\documentclass[11pt, letterpaper]{article}

\usepackage{cluster}
\usepackage{verbatim}
\newcommand{\transp}{\top}
\newcommand{\Cost}{\mathsf{Cost}}

\newcommand{\trace}{\Tr}
\newcommand{\vNE}{\mathsf{vNE}}
\newcommand{\trosp}{\mathbb{S}_{\ge 0}}
\newcommand{\QRL}{\mathsf{QRE}}
\newcommand{\QRE}{\QRL}
\newcommand{\wopts}{W^*}
\newcommand{\wopt}{W^*}
\newcommand{\mopt}{M^*}

\newcommand{\fan}[1][k]{\mathcal{F}_{#1}} 
\newcommand{\phib}[1][1]{\Phi_b (#1)}

\newcommand{\projalg}{\mathsf{FantopeProjection}}
\newcommand{\tipest}{\mathsf{InnerProductTrace}}
\newcommand{\ipest}{\mathsf{InnerProductEstimation}}
\newcommand{\solvalg}{\mathsf{PackingCoveringDecision}}
\newcommand{\packred}{\mathsf{ApproxCost}}

\newcommand{\wts}[1][t]{w^{(#1)}}
\newcommand{\Mt}[1][t]{M^{(#1)}}
\newcommand{\Wt}[1][t]{W^{(#1)}}
\newcommand{\wMt}[1][t]{\wt{M}^{(#1)}}
\newcommand{\wWt}[1][t]{\wt{W}^{(#1)}}
\newcommand{\St}[1][t]{S^{(#1)}}

\newcommand{\omt}[1][t]{\omega^{(t)}}
\newcommand{\thet}[1][t]{\theta^{(t)}}

\newcommand{\gt}[1][t]{\gamma^{(#1)}}
\newcommand{\bt}[1][t]{\beta^{(#1)}}
\newcommand{\bpt}[1][t]{\beta^{\prime (#1)}}

\newcommand{\st}[1][t]{\sigma^{(#1)}}
\newcommand{\taut}[1][t]{\tau^{(#1)}}
\newcommand{\Vt}[1][t]{V^{(#1)}}
\newcommand{\vt}[1][t]{v^{(#1)}}
\newcommand{\zt}[1][t]{z^{(#1)}}
\newcommand{\yt}[1][t]{y^{(#1)}}

\newcommand{\Ft}[1][t]{F^{(#1)}}
\newcommand{\Gt}[1][t]{G^{(#1)}}
\newcommand{\phit}[1][t]{\phi^{(#1)}}
\newcommand{\psit}[1][t]{\psi^{(#1)}}

\newcommand{\epd}{\eps^\dagger}
\newcommand{\dt}[1][t]{\xi^{(#1)}}
\newcommand{\dd}{\delta^\dagger}

\def\authornotes{1pt}

\ifdim\authornotes=1pt
    \newcommand{\ynote}[1]{\footnote{\color{ForestGreen}Yeshwanth: #1}}
    \newcommand{\snote}[1]{\footnote{\color{Fuchsia}Sidhanth: #1}}
    \newcommand{\mnote}[1]{\footnote{\color{Periwinkle}Morris: #1}}
\else
    \newcommand{\ynote}[1]{}
    \newcommand{\snote}[1]{}
    \newcommand{\mnote}[1]{}
\fi

\usepackage{tikz}
\begin{document}

\title{List Decodable Mean Estimation in Nearly Linear Time}
\author{Yeshwanth Cherapanamjeri\thanks{EECS Department, University of California Berkeley.  \texttt{yeshwanth@berkeley.edu}.} \and Sidhanth Mohanty\thanks{EECS Department, University of California Berkeley.  \texttt{sidhanthm@cs.berkeley.edu}.  Supported by NSF grant CCF-1718695.} \and Morris Yau\thanks{EECS Department, University of California Berkeley.  \texttt{morrisyau@berkeley.edu}.  Supported by NSF grant CCF-1718695.}}
\date{\today}
\maketitle

\begin{abstract}
Learning from data in the presence of outliers is a fundamental problem in statistics.  Until recently, no computationally efficient algorithms were known to compute the mean of a high dimensional distribution under natural assumptions in the presence of even a small fraction of outliers.  In this paper, we consider robust statistics in the presence of overwhelming outliers where the majority of the dataset is introduced adversarially.  With only an $\alpha < 1/2$ fraction of ``inliers'' (clean data) the mean of a distribution is unidentifiable.  However, in their influential work, \cite{CSV17} introduces a polynomial time algorithm recovering the mean of distributions with bounded covariance by outputting a succinct list of $O(1/\alpha)$ candidate solutions, one of which is guaranteed to be close to the true distributional mean; a direct analog of 'List Decoding' in the theory of error correcting codes.  In this work, we develop an algorithm for list decodable mean estimation in the same setting achieving up to constants the information theoretically optimal recovery, optimal sample complexity, and in nearly linear time up to polylogarithmic factors in dimension.  Our conceptual innovation is to design a descent style algorithm on a nonconvex landscape, iteratively removing minima to generate a succinct list of solutions.  Our runtime bottleneck is a saddle-point optimization for which we design custom primal dual solvers for generalized packing and covering SDP's under Ky-Fan norms, which may be of independent interest.                            
\end{abstract}

\thispagestyle{empty}
\setcounter{page}{0}
\newpage

\tableofcontents

\thispagestyle{empty}
\setcounter{page}{0}
\newpage

\section{Introduction}
\label{sec:int}

Estimating the mean of data is a cardinal scientific task.  The population mean can be shifted arbitrarily by a single outlier, a problem which is compounded in high dimensions where outliers can conspire to destroy the performance of even sophisticated estimators of central tendency.  Robust statistics, beginning with the works of Tukey and Huber \cite{Tukey60,Huber64}, endeavors to design, model, and mitigate the effect of data deviating from statistical assumptions \cite{huber11}.        
         
A canonical model of data corruption is the Huber contamination model \cite{Huber64}.  Let $I(\mu)$ be a probability distribution parameterized by $\mu$.  We say a dataset $X_1,X_2,...,X_N$ is $\alpha$-Huber contaminated for some constant $\alpha \in [0,1]$ if it is drawn i.i.d from
$$X_1,X_2,...,X_N \sim \alpha \mathcal{I}(\mu) + (1 - \alpha)\mathcal{O}$$
where $\mathcal{O}$ is an arbitrary outlier distribution which can be  adversarial and dependent on $\mathcal{I}(\mu)$.  The goal is to estimate $\mu$ with an estimator $\hat{\mu}$ such that the two are close with respect to a meaningful metric.  The Huber contamination model captures the setting where only an $\alpha$ fraction of the dataset is subject to statistical assumptions.  One would hope to design estimators $\hat{\mu}$ for which $\alpha$ is as small as possible thereby tolerating the largest fraction of outliers--a quantity known as the breakdown point .  The study of estimators with large breakdown points is the focus of a long and extensive body of work, which we do not attempt to survey here.  For review see \cite{huber11,HRRS86}.  

A first observation, is that the breakdown point of a single estimator must be smaller than $\frac{1}{2}$.  For concreteness, consider the problem of estimating the mean of a standard normal. The adversary can set up a mixture of $\frac{1}{\alpha}$ standard normals for which the means of the mixture components are far apart.  This intrinsic difficulty also gives rise to a natural notion of recovery in the presence of overwhelming outliers.  Instead of outputting a single estimator, consider outputting a list of candidate estimators $\mathcal{L}  = \{\hat{\mu}_1, \hat{\mu}_2, ..., \hat{\mu}_{\frac{1}{\alpha}}\}$ with the guarantee that the true $\mu$ is amongst the elements of the list.  This is the setting of 'List Decodable Learning' \cite{BBV08,CSV17}, analogous to list decoding in the theory of error correcting codes.           

In their influential work \cite{CSV17} introduces list decodable learning in the context of robust statistics.  They consider the problem of estimating the mean $\mu$ of a $d$-dimensional distribution $\mathcal{I}(\mu)$ with a bounded covariance $\Cov(\mathcal{I}(\mu)) \preceq \sigma^2I$ for a constant $\sigma$ from $N = \frac{d}{\alpha}$ samples.  Their algorithm recovers a list $\mathcal{L}$ of $O(\frac{1}{\alpha})$ candidate means with the guarantee that there exists a $\hat{\mu}^* \in \mathcal{L}$ achieving the recovery guarantee $\norm{\hat{\mu}^* - \mu} \leq O\left(\sigma\sqrt{\frac{\log\left(\frac{1}{\alpha}\right)}{\alpha}}\right)$ with high probability $1 - \frac{1}{\poly(d)}$.  Furthermore, their algorithm is 'efficient', running in time $\poly(N,d,\frac{1}{\alpha})$ via the polynomial time solvability of ellipsoidal convex programming.  

\subsection{Results}
\label{ssec:results}

Our first contribution is an algorithm for list decodable mean estimation of covariance bounded distributions, which outputs a list $\mathcal{L}$ of length $O(\frac{1}{\alpha})$, achieving (up to constants) the information theoretically optimal recovery $O(\frac{\sigma}{\sqrt{\alpha}})$, with linear sample complexity $N = \frac{d}{\alpha}$, and running in nearly linear time  $\widetilde{O}(Nd\poly(\frac{1}{\alpha}))$ where $\widetilde{O}$ omits logarithmic factors in $d$.  For the matching minimax $\Omega(\frac{\sigma}{\sqrt{\alpha}})$ lower bound see \cite{DKSlist18}.  Formally, we state our main theorem. 

\begin{theorem} \label{thm:main-thm-list-dec-mean}
    Let $\mathcal{I}(\mu)$ be a distribution in $\R^d$ with unknown mean $\mu \in \R$ and bounded covariance $\Cov(\mathcal{I}(\mu)) \preceq \sigma^2I$ for a constant $\sigma \in \R^+$.  Let $\mathcal{I} := \{x_1,x_2,...,x_{\alpha N}\}$ be a dataset in $\R^d$ drawn i.i.d from $\mathcal{I}(\mu)$.  An adversary then selects an arbitrary dataset in $\R^d$ denoted $\mathcal{O} := \{x'_1, x'_2, ..., x'_{(1-\alpha)N}\}$ which in particular, may depend on $\mathcal{I}$.  The algorithm is presented with the full dataset $X := \mathcal{I} \cup \mathcal{O}$.  For any $N \geq \frac{d}{\alpha}$, \cref{alg:outputlist} outputs a list $\mathcal{L}  = \{\hat{\mu}_1, \hat{\mu}_2, ..., \hat{\mu}_{O(\frac{1}{\alpha})}\}$ of length $O(\frac{1}{\alpha})$ such that there exists a $\hat{\mu}^* \in \mathcal L$ satisfying $\norm{\hat{\mu}^* - \mu} \leq O(\frac{\sigma}{\sqrt{\alpha}})$ with high probability $1 - \frac{1}{d^{10}}$.  Furthermore, the algorithm runs in time $\tilde{O}(Nd\poly(\frac{1}{\alpha}))$.         
\end{theorem}

For precise constants and failure probability see \cref{sec:algorithms}.  At a high level, we define a nonconvex cost function for which $\mu$ is an approximate minimizer and build a 'descent style' algorithm to find $\mu$.  As with most nonconvex algorithms, our approach is susceptible to falling in suboptimal minima. Our key algorithmic insight is that our algorithm fails to descend the cost function exactly when a corresponding dual procedure succeeds in "sanitizing" the dataset by removing a large fraction of outliers --- a win-win.  

\paragraph{Applications of List Decoding} 
First observed in \cite{CSV17}, the list decoding problem lends itself to applications for which our algorithm offers immediate improvements.  Firstly, it is perhaps surprising that a succinct list of estimators can be procured from a dataset overwhelmed by outliers.  Perhaps more surprising is that the optimal candidate mean can be isolated from the list $\mathcal L$ with additional access to a mere $\log(\frac{1}{\alpha})$ clean samples drawn from $\mathcal{I}(\mu)$.  This "semi-supervised" learning is compelling in settings where large quantities of data are collected from unreliable providers (crowdsourcing, multiple sensors, etc.). Although it is resource intensive to ensure the cleanliness of a large dataset, it is easier to audit a small, in our case $\log(\frac{1}{\alpha})$, set of samples for cleanliness.  Given access to this small set of samples as side information, our algorithm returns estimators for mean estimation with breakdown points higher than $\frac{1}{2}$ in nearly linear time.  

Faster list decodable mean estimation also accelerates finding planted partitions in semirandom graphs.  In particular, consider the problem where $G$ is a directed graph where the (outgoing) neighborhoods of an $\alpha$ fraction of vertices $S$ are random while the neighborhoods of the remaining vertices are arbitrary, and the goal is to output $O(1/\alpha)$ lists such that one of them is ``close'' to $S$.  Our algorithm for list decodable mean estimation implies a faster algorithm for this problem as well.

Lastly, list decodable mean estimation is a superset of learning mixture models of bounded covariance distributions with minimum mixture weight $\alpha$.  By treating a single cluster as the inliers, one can recover the list of means comprising the mixture model. Notably, this can be done without any separation assumptions between the mixture components and is robust to outliers.    

\paragraph{Fast Semidefinite Programming:}
Rapidly computing our cost function necessitates the design of new packing/covering solvers for Positive Semidefinite Programs (SDP) over general Fantopes (the convex hull of the projection matrices).  Positive SDP's have seen remarkable success in areas spanning quantum computing, spectral graph theory, and approximation algorithms (See \cite{ahk,ALO16,jambu} and the references therein).  Informally, a packing SDP computes the fractional number of ellipses that can be packed into a spectral norm ball which involves optimization over the spectrahedron.  A natural question is whether the packing concept can be extended to balls equipped with general norms, say the sum of the top $k$ eigenvalues (the Ky Fan norm), where for $k = 1$ we recover the oft studied spectral norm packing.  We use results from Loewner's theory of operator monotonicity and operator algebras to design fast, and as far as we know the first solvers for packing/covering positive SDP's under Ky Fan norms (See \cref{thm:solthm}). 

\subsection{Related Work}
\paragraph{Robust Statistics}

Robust statistics has a long history \cite{Tukey60,Tukey75,Huber64,Hampel71}.  This extensive body of work develops the theory of estimators with high breakdown points, of influence functions and sensitivity curves, and of designing robust M-estimators.  See \cite{huber11,HRRS86}.  However, little was understood about the computational aspects of robustness which features prominently in high dimensional settings.   

Recent work in theoretical computer science \cite{DKKLMS16,LRV16} designed the first algorithms for estimating the mean and covariance of high dimensional gaussians tolerating a constant fraction of outliers in polynomial time $poly(N,d,\frac{1}{1 - \alpha})$.  Since then, a flurry of work has emerged studying robust regression \cite{KKM18,DKSregression18}, sparse robust regression \cite{BDLS17,DKKPS19}, fast algorithms for robustly estimating mean/covariance \cite{CDG,DHL19,CDGW19}, statistical query hardness of robustness \cite{DKSquery17}, worst case hardness \cite{HL19}, robust graphical models \cite{CDKS18}, and applications of the sum of squares algorithm to robust statistics \cite{KSS18}.  See survey \cite{DKsurvey19} for an overview.  

\paragraph{List Decodable Learning}

Despite the remarkable progress in robust statistics for large $\alpha$ contamination, progress on the list decoding problem has been slower. This is partially owed to the intrinsic computational hardness of the problem.  Even for the natural question of list decoding the mean of a high dimensional gaussian, \cite{DKSlist18} exhibits a quasipolynomial time lower bound against Statistical Query algorithms for achieving the information theoretically optimal recovery of $\Theta\big(\sqrt{\log(\frac{1}{\alpha})}\big)$.  This stands in contrast to large $\alpha$ robust mean estimation where nearly linear time algorithms \cite{CDG} achieve optimal recovery.  

In light of this hardness, a natural question is to determine whether polynomial time algorithms can at least approach the optimal recovery for list decoding the mean of a gaussian.  In a series of concurrent works \cite{KS17} \cite{DKSlist18}, develop the first algorithms approaching the $\Theta\big(\sqrt{\log(\frac{1}{\alpha})}\big)$ recovery guarantee. At a high level, both papers achieve recovery $O(\frac{\sigma}{\alpha^{c/k}})$ for different fixed constants $c > 1$ in time $\poly(\frac{d}{\alpha})^{O(k)}$ for $k$ a positive integer greater than $2$. The \cite{DKSlist18} algorithm, known as the "multi-filter", is a spectral approach reasoning about high degree polynomials of the moments of data. Furthermore, the "low degree" multi-filter achieves a suboptimal $O\big(\sqrt{\frac{\log(\frac{1}{\alpha})}{\alpha}}\big)$ recovery guarantee for list decoding the mean of subgaussian distributions, which is fast and may be of practical value.  \cite{KS17} develop a convex hierarchy (sum of squares) style approach, which achieve similar  guarantees for more general distributional families satisfying a poincare inequality.  In particular for list decoding the mean of bounded covariance distributions they achieve the optimal $O(\frac{1}{\sqrt{\alpha}})$ guarantee via the polynomial time solvability of convex concave optimization.  Finally, \cite{DKSlist18,KS17} and a concurrent work \cite{HLmixture18} develop tools for reasoning about the high degree moments of data to break the longstanding  "single-linkage" barrier in clustering mixtures of spherical gaussians.           

In other statistical settings a series of concurrent works \cite{RY20,KKK19} demonstrate information theoretic impossibility for list decoding regression even under subguassian design.  Similar barriers arise in the context of list decodable subspace recovery \cite{RYsubspace20,BK20} where it is information theoretically impossible to list decode a dataset for which an $\alpha$ fraction is drawn from a subgaussian distribution in a subspace.  Indeed, since list decoding is a superset of learning mixture models, these hardness considerations stem from barriers in learning mixtures of linear regressions and subspace clustering.  On the other hand, the above works also construct polynomial time, $d^{\poly(\frac{1}{\alpha})}$, algorithms for regression and subspace recovery for Gaussian design and Gaussian subspaces respectively, which holds true for a larger class of "certifiably anticoncentrated" distributions.  

In this backdrop of computational and statistical hardness, and given the practical value of robust statistics, it is a natural challenge to design list decoding algorithms that are both fast and statistically optimal.  The current work is a step in this direction.    

\paragraph{SDP Solvers}
There has been much recent interest in designing fast algorithms for positive SDP solvers due to the ubiquity of their application in approximation algorithms. We do not attempt to survey the full breadth of these results and their applications in this section. We refer the interested reader to \cite{jambu,ALO16,peng2012faster,ahk} for more context on these developments. We will restrict ourselves to the following class of SDPs relevant to our work:
 
\begin{equation*}
    \begin{gathered}
    \max \sum_{i = 1}^n w_i \\
    \text{s.t } \sum_{i = 1}^n w_i A_i \preccurlyeq I \tag{Gen-Pack} \label{eq:genpack}\\
    \left\|\sum_{i = 1}^n w_i B_i\right\|_k \leq k
\end{gathered}
\end{equation*}
where $A_i \in \psd[l]$ and $B_i\in \psd[m]$. While existing fast solvers \cite{peng2012faster,ALO16,jambu} only apply to the above setting when $k = 1$, we generalize the approach of \cite{peng2012faster} to for all $k$ with running times scaling at most polynomiall in $k$. In particular, we show for small values of $k$, \ref{eq:genpack} can be solved in nearly linear time for a broad range of settings including ours and inherits the parallel, width-independent properties of \cite{peng2012faster}. See \cref{thm:solthm} for the exact statement of the result. However, carrying out this generalization brings with it a host of technical challenges which are explained in more detail in \cref{sec:techset} including a more refined analysis of the power method and a novel technique to bound errors incurred in a hard-thresholding operator due to approximate eigenvector computation.

\paragraph{Semirandom Graph Inference}  The study of problems that are typically computationally hard in the worst case in semirandom graph models was initiated by \cite{blum1995coloring} and perpetuated by \cite{feige2001heuristics}.  A specific problem of interest to us studied by \cite{feige2001heuristics} for which nearly optimal algorithms were given by \cite{mckenzie2020new} is the \emph{semirandom independent set} problem where the set of edges between a planted independent set and the remaining (adversarially chosen) graph come from a randomized model.  In a similar vein \cite{CSV17} studies a \emph{planted partition} where instead of an independent set the given graph is some other sparse random graph (albeit directed).  Our results improve upon the statistical guarantees of \cite{CSV17} as well as give faster algorithms, however both \cite{CSV17} and our work fall short of capturing the results of \cite{mckenzie2020new} due to the directed model we work in.  However, we believe the hurdle is a technical point rather than an inherent shortcoming of our approach.

\paragraph{Sample Complexity:} The following lemma of \cite{CSV17} achieves linear sample complexity which suffices for our algorithm.  

\begin{lemma}[{\cite[Proposition 1.1]{CSV17}}]
Suppose $\calI(\mu)$ is a distribution on $\R^d$ with mean $\mu$ and covariance $\Cov(I(\mu)) \preceq \sigma^2I$ for a constant $\sigma > 0$.  Then given $n \geq d$ samples from $\calI(\mu)$, with probability $1 - \exp(\frac{-n}{64})$ there exists a subset $\calI \in [n]$ of size $|\calI| \geq \frac{n}{2}$ such that $\norm{\frac{1}{|I|}\sum_{i \in \calI} (x_i - \mu)(x_i - \mu)^T} \preceq 24\sigma^2I$ 
\end{lemma}

Taking $N = O\left(\frac{d}{\alpha}\right)$, for the rest of the paper we will adjust $\sigma$ by a constant and assume the inlier set $I$ satisfies $\norm{\frac{1}{|I|}\sum_{i \in \calI} (x_i - \mu)(x_i - \mu)^T} \preceq \sigma^2I$. 

\paragraph{Notation: } We will frequently use $\psd$ to denote the set of positive semidefinite matrices with dimension $n$. For $A \in \psd$, we will frequently write the ordered eigenvalue decomposition of $A = \sum_{i = 1}^n \lambda u_iu_i^\top$ with $\lambda_1 \geq \dots \geq \lambda_n$ and for any matrix, $M$, $\sigma_i(M)$ denotes its $i^{th}$ singular value. The \emph{Ky--Fan matrix $k$-norm} of a matrix, $M$, is defined as the sum of the top-$k$ singular values of $M$; i.e $\norm{M}_k := \sum_{i=1}^k \sigma_i(M)$.  Notably, $\norm{\cdot}_1$ is the operator norm and $\norm{\cdot}_d$ is the trace norm.  However, we will stick with $\|\cdot\|$ for operator norm and $\|\cdot\|_*$ for trace norm. Continuing along these lines, we also define the $\ell$-\emph{Fantope}, denoted by $\mc{F}_\ell$ and characterized as $\mc{F}_\ell = \{M \in \psd: \Tr M = \ell \text{ and } \norm{M} \leq 1\}$. Finally, given $\{b_i \geq 0\}_{i = 1}^N$, we define the set $\Phi_b (\gamma) = \{w_i \geq 0: \sum_{i}^N w_i = \gamma \text{ and } w_i \leq b_i\}$ and $\Phi_b = \Phi_b(1)$. For a set of vectors, $V = [v_1, \dots, v_k]$, we will use $\mc{P}_V^\perp$ to denote the projection onto the orthogonal subspace of the span of $v_i$.

\paragraph{Organization: } Our paper is organized as follows: In \cref{sec:techset}, we outline the key ideas underlying the design of our algorithm for list-decodable mean estimation, our solver for the generalized class of Packing/Covering SDPs considered in this paper and the technical challenges involved in doing so. Then, in \cref{sec:alg,sec:infgraph}, we formally describe and analyze our algorithm for list-decodable mean estimation and its application to the semirandom graph model considered in \cite{CSV17}. \cref{sec:pow,sec:genpcsolver,sec:fastproj} contain our refined power method analysis, a formal description and analysis of our solver and the hard thresholding based operator required to implement the solver in nearly-linear time. Finally, \cref{sec:appsketch,sec:miscellaneous,sec:fstsolv} contain supporting results required by the previous sections.
\section{Techniques}
\label{sec:techset}

First we present an inefficient algorithm for list decodable mean estimation.  Although it is inefficient, it captures the core ideas and foreshadows the difficulties encountered by our efficient algorithm.  At a high level, the inefficient algorithm greedily searches through the dataset for subsets of points with small covariance with the goal of finding the subset of inliers.          
 
\paragraph{Inefficient Algorithm:} Our inefficient algorithm is a list decoding analogy to the nonconvex weight minimization procedure first proposed in \cite{DKKLMS16} \footnote{There they directly design a separation oracle for the objective \cref{eq:objective} for $\alpha > \frac{2}{3}$, which yields polynomial time guarantees for robust mean estimation via the ellipsoid algorithm.  It is plausible that a similar approach could yield polynomial time algorithms for list decodable mean estimation, but use of the ellipsoid algorithm would preclude the possibility of fast algorithms so we do not pursue that avenue here. }. Let $\calL$ be a list initialized to be the empty set.  Let $b$ be a vector initialized to be $(\frac{2}{\alpha N}, \frac{2}{\alpha N}, ..., \frac{2}{\alpha N}) \in \R^N$.  The algorithm iterates the following loop for $\frac{2}{\alpha}$ iterations.  

\begin{enumerate}
    \item First, solve the optimization problem \begin{align} \label{eq:objective}
    \hat{w} = \argmin_{w \in \Phi_b}  \left\|\sum_{i = 1}^N w_i (x - \mu(w)) (x - \mu(w))^{\transp} \right\| \end{align}
    Where $\mu(w) = \sum_{i=1}^N w_ix_i$
    \item Second, append $\hat{\mu} = \sum_{i=1}^N \hat{w}_i x_i$ to $\calL$
    \item Third, update $b$ such that $b_i = b_i - \hat{w}_i$
\end{enumerate}

 We claim the algorithm outputs a list $\calL$ of length $\frac{2}{\alpha}$ and that there exists a $\hat{\mu}^* \in L$ satisfying $\norm{\hu^* - \mu} \leq O(\frac{\sigma}{\sqrt{\alpha}})$.
Next we outline the proof of correctness.  

\paragraph{Proof Outline:} We proceed by contradiction and assume $\norm{\hu - \mu} \geq \frac{10\sigma}{\sqrt \alpha}$ for all $\hu \in \calL$.  Consider the first iteration. The scaled indicator of the inliers $\frac{1}{\alpha N}\mathbb{1}[i \in \calI]$ is feasible for \cref{eq:objective}. Thus, we have $\norm{ \sum_{i = 1}^N \hat{w}_i (x - \mu(\hat w)) (x - \mu(\hat w))^{\transp} } \leq 4\sigma^2$.  It is a fact that given two subsets of the data that are both covariance bounded, if the means of the subsets are far apart then the subsets don't overlap substantially.  This fact extends beyond subsets and holds true even for the soft weights that we are considering here.  See \cref{fact:resilience}.  Applying this fact we conclude $\sum_{i \in \calI}\hat w \leq \frac{\alpha}{2}$.  By assumption, subsequent iterations of the algorithm continue to output $\hu$ far away from the true mean so a substantial fraction of the inlier weight is preserved enabling the above argument to go through repeatedly.  Formally, this would be argued inductively, see \cref{cor:maincorollary}.  Thus, at every iteration $\sum_{i \in \calI}\hat w \leq \frac{\alpha}{2}$. Notice that any algorithm that removes more outlier weight than inlier weight at a ratio $\sum_{i \in O}\hat{w}_i \geq \frac{2}{\alpha}\sum_{i \in I}\hat{w}_i$ will eventually remove all the outlier weight leaving more than $\frac{1}{2}$ of the inlier weight intact.  Since the total inlier weight is initialized to be $\sum_{i \in \calI}b_i = 2$, we have at the second to last iteration a dataset comprised entirely of inliers which implies $\norm{\hu - \mu} \leq \frac{10\sigma}{\sqrt \alpha}$, which is a contradiction.

\paragraph{Sanitizing the Dataset: } Abstracting the guarantees of our inefficient algorithm, we say that an algorithm "sanitizes" a dataset if it outputs a tuple $(\hu, \hat w)$ where $\sum_{i=1}^N\hat w \geq \Omega(1)$ satisfying the following conditions.  If $\norm{\hu - \mu} \geq O(\frac{\sigma}{\sqrt \alpha})$ then $\sum_{i \in O}\hat{w}_i \geq \frac{2}{\alpha}\sum_{i \in I}\hat{w}_i$.  Any algorithm that sanitizes the dataset iteratively, is guaranteed to succeed as a list decoding algorithm.  This is made formal in \cref{sec:algorithms}.

\paragraph{Descent Style Formulation:} The optimization problem \cref{eq:objective} is nonconvex and hard to solve directly.  A novel approach to minimizing \cref{eq:objective} is to replace $\mu(w)$ with a parameter $\nu$ and define a cost function $f(\nu)$.  First introduced in \cite{CDG} in the context of robust mean estimation and later in robust covariance estimation \cite{CDGW19} consider the function $f(\nu)$ defined as follows:
\begin{align*}
    f(\nu) \coloneqq \min_{w \in \Phi_b}  \left\|\sum_{i = 1}^N w_i (x - \nu) (x - \nu)^{\transp}\right\|
\end{align*}
where $b_i = \frac{1}{\alpha N}$ for all $i \in [N]$.  This formulation has two appealing aspects.  Firstly, the cost function can be computed efficiently via convex concave optimization.  Indeed, the operator norm can be replaced by the maximization over its associated fantope $\mathcal{F}_1$ 
\begin{align*}
    f(\nu) \coloneqq \min_{w \in \Phi_b} \max_{M \in \mathcal{F}_1} \left\langle M, \sum_{i = 1}^N w_i (x_i - \nu) (x_i - \nu)^{\transp} \right \rangle.
\end{align*}
Secondly, for $\alpha > \frac{2}{3}$ (robust mean estimation), a crucial insight of \cite{CDG} is that $f(\nu)$ approximates the squared distance from $\nu$ to the mean $\mu$.  Then a good estimate of the mean is the minimizer of the cost.  
\begin{align} \label{eq:descendcost}
    \hat{\mu} := \argmin_{\nu \in \R^d} f(\nu) \approx \argmin_{\nu \in \R^d} \norm{\nu - \mu}^2
\end{align}
In their setting the minimization in \cref{eq:descendcost} can be performed by a descent style algorithm.   

Substantial challenges arise when designing such a cost function for list decodable mean estimation.  Chiefly, the inliers are unidentifiable from the dataset so there is no function of the data that approximates the distance to the true mean. Our solution is to design a function that either approximates the distance to the true mean, or when the approximation is poor, prove there exists a corresponding dual procedure that sanitizes the dataset.  This win-win observation can be made algorithmic and is the subject of \cref{sec:algorithms}  

\subsection{Our Approach}
\label{ssec:ourapproach}
\paragraph{Designing Cost:} We make extensive use of the Fantope \cite{Dattoro05}, the convex hull of the rank $\ell$ projection matrices. This set of matrices, denoted $\calF_\ell$, is a tight relaxation for simultaneous rank and orthogonality constraints on the positive semidefinite cone.  This also makes it amenable to semidefinite optimization.  We define 
\begin{equation*}
    \mathcal{F}_\ell = \{M \in \R^{d \times d}: 0 \preceq M \preceq I \text{ and } \Tr (M) = \ell\}.
\end{equation*}
Optimization over the Fantope provides a variational characterization of the principal subspace of a symmetric matrix $B \in \R^{d \times d}$. Indeed the \emph{Ky Fan Theorem}, states that the \emph{Ky Fan Norm} defined to be the sum of the $\ell$ largest eigenvalues of a psd matrix is equal to 
\begin{align}
    \norm{B}_\ell := \sum_{i=1}^\ell \lambda_i(B) = \max_{Q^TQ = I_\ell}\langle B, QQ^T\rangle = \max_{M \in \calF_\ell} \langle B,M\rangle
 \end{align}
Here the first equality is an extremal property known as \emph{Ky Fan's Maximum Principle}, and the second equality follows because the rank $\ell$ projection matrices are extremal points of $\calF_\ell$.  See \cite{OW92}.  We use this principle to generalize the min-max problem considered in the previous section. Let $Cost_{X,b,\ell}(\nu): \R^d \rightarrow \R^+$ be defined 
\begin{equation}\label{eq:cost}
    Cost_{X,b,\ell}(\nu) = \min_{w \in \Phi_b} \max_{M \in \mathcal{F}_\ell} \inp{\sum_{i = 1}^N w_i (X_i - \nu)(X_i - \nu)^\top}{M} = \min_{w \in \Phi_b} \norm*{\sum_{i = 1}^N w_i (X_i - \nu)(X_i - \nu)^\top}_\ell
\end{equation}

We call the above min-max formulation the dual and the associated minimizer $w^*$ the dual minimizer or dual weights.  By Von Neumann's min max theorem we have 
\begin{equation*}
    Cost_{X,b,\ell}(\nu) = \max_{M \in \mathcal{F}_\ell} \min_{w \in \Phi_b}  \inp{\sum_{i = 1}^N w_i (X_i - \nu)(X_i - \nu)^\top}{M} 
\end{equation*}
Where we refer to the maximizer $M^*$ as the primal maximizer.  For the remainder of this section we will set $\ell = \frac{1}{\alpha}$ and $b = (\frac{1}{\alpha N}, ..., \frac{1}{\alpha N}) \in \R^N$ and drop the subscripts in $Cost(\cdot)$.

\paragraph{An Easier Problem: } To aid in exposition, we illustrate our algorithmic approach on the simpler and well understood problem of finding the $k = \frac{1}{\alpha}$ means of data drawn from a mixture of $k$ bounded covariance distributions. That is $x_1,...,x_N \thicksim \frac{1}{k}\sum_{i = 1}^k \mc{D}(\mu_i)$ for a distribution $D(\mu) \preceq I$ with means $\{\mu_i\}_{i=1}^k$ and let $T_i$ denote the set of points in each cluster $i \in [k]$. Consider a vector $\nu$ that is further than $O(\sqrt{k})$ away from all the means $\{\mu_i\}_{i=1}^k$. By standard duality arguments, we see that $Cost(\nu)$ is a good approximation to the distance to the closest cluster center denoted $\mu^*$ comprised of points $T^*$. Furthermore, $\mu^* - \nu$ is almost completely contained in the top-$O(k)$ singular subspace of $M^*$, denoted $V$. We may now project all the data points onto the affine subspace $V$ offset to $\nu$ forming the set $X' :=\{\Pi_V(x_i - \nu)\}_{i=1}^N$. The second observation is that a randomly chosen point $\bar x \in T^*$  satisfies $\norm{\Pi_V(\bar x - \mu^*)} \leq O(\sqrt{k})$  with constant probability. Due to the fact that $\mu^* - \nu$ is almost completely contained in $V$, we get by picking a set of $p = \tilde{O}(k)$ random data points $ R := \{x'_1, x'_2,..., x'_p\} \in X'$, that there exists a point $\hat{x}' \in R$ substantially closer to $\mu^*$ than $\nu$ with high probability $1 - \frac{1}{\poly(d)}$. We can efficiently certify this progress by computing the value of $Cost(\hat{x}')$. By iterating this procedure we converge to within $O(\sqrt{k})$ of the mean of a cluster center.    

\paragraph{List Decoding Main Lemma: } In analogy to clustering, one should hope that for any $\nu \in R^d$ further than $O(\frac{\sigma}{\sqrt{\alpha}})$ from $\mu$, that $Cost(\nu) \approx \norm{\nu - \mu}^2$.  Although this is impossible, it turns out that when it is false, there exists a corresponding "dual procedure" for outputting a sanitizing tuple.  More precisely, we claim that either $0.4\norm{\nu - \mu}^2 \leq Cost(\nu) \leq 1.1\norm{\nu - \mu}^2$, or a simple procedure outputs a set of weights $\hat w$ identifying vastly more outliers than inliers i.e $\sum_{i \in \calO}\hat{w}_i \geq \frac{\alpha}{2}\sum_{i \in \calI}\hat{w}_i$, or both.    

The dual procedure is as follows.  Let $\widehat{\Sigma} := \sum_{i=1}^n w^*_i (x_i - \nu)(x_i - \nu)^T$ be the weighted second moment matrix centered at $\nu$. Let $V$ be the top $O(\frac{1}{\alpha})$ eigenspace of $\widehat\Sigma$. We project the dataset onto the affine subspace $V$ with offset $\nu$. We then sort the points $\{\Pi_V(x_i - \nu)\}_{i=1}^N$ by Euclidean lengths.  This sorting determines an ordering of the weights $w^*_1,...,w^*_N$. We pass through the sorted list, and find the smallest $m \in [N]$ such that $\sum_{i=1}^m w^*_i \geq 0.5$.  We set $\hat{w}_i = w^*_i$ for $i = 1,...,m$ and $\hat{w}_i = 0$ for $i > m$. The following lemma guarantees $\sum_{i \in \calO}\hat{w}_i \geq \frac{\alpha}{2}\sum_{i \in \calI}\hat{w}_i$. 

\begin{lemma}\label{lem:nonalgorithmic} (Nonalgorithmic \cref{lem:dualprocedure} with Exact Cost Evaluation) Let $\nu \in \R^d$ be any vector satisfying $\norm{\nu - \mu} \geq O(\frac{\sigma}{\sqrt \alpha})$.  Let $Cost_{X,b,\ell}(\nu)$ be defined as in \cref{eq:cost} for $b = (\frac{1}{\alpha N}, ..., \frac{1}{\alpha N}) \in \R^N$ and $\ell = \frac{1}{\alpha}$. Let $w^* \in \Phi_b$ be the corresponding dual minimizer.  Let $\hat w$ be defined as follows
$$\hat{w} := \argmin\limits_{p_i \in [0,w^*_i] \text{ and }\norm{p}_1 \geq 0.5} \sum_{i=1}^N p_i\left\|\Pi_V(x_i-\nu)\right\|$$ 
for $V$ the subspace defined above.  Then either the cost is a constant factor approximation to the distance to the true mean, $0.4\norm{\nu - \mu}^2 \leq Cost_{X,b,\ell}(\nu) \leq 1.1\norm{\nu - \mu}^2$, or $\hat w$ identifies a set of weights with vastly more outliers than inliers, $\sum_{i \in \calO}\hat{w}_i \geq \frac{\alpha}{2}\sum_{i \in \calI}\hat{w}_i$ (or both). 
\end{lemma}

In \cref{lem:dualprocedure} we state the algorithmic version of the above lemma.  There it is important to take into account technicalities involving the approximate evaluation of $Cost_{X,b,\ell}(\cdot)$, and provide a procedure for making progress when the cost is a constant approximation $\norm{\nu - \mu}^2$.  This will be done in a manner akin to the procedure for clustering described earlier.  Nevertheless, \cref{lem:nonalgorithmic} captures the key guarantee that ensures our main algorithms in \cref{sec:algorithms} succeeds.

\subsection{Generalized Packing/Covering Solvers and Improved Power Method Analysis}
\label{ssec:genpacktech}

We start by considering the simpler problem of computing $Cost_{X, b, 1}(\nu)$. The approach taken in \cite{CDG} is to reduce the problem to a packing SDP via the introduction of an additional parameter $\lambda$; specifically, they solve the following packing SDP:

\begin{equation*}
    \label{eqn:packptz}
    \begin{gathered}
    \max_{w_i \geq 0} \sum_{i = 1}^N w_i \\
    \text{s.t } \left\|\sum_{i = 1}^N w_i (X_i - \nu)(X_i - \nu)^\top\right\| \leq \lambda \tag{Packing-SDP}\\
    w_i \leq b_i
\end{gathered}
\end{equation*}
for which there exist fast linear-time solvers \cite{peng2012faster,ALO16}. It can be shown that the value of the above program when viewed as a function of $\lambda$ is monotonic, continuous and attains the value $1$ precisely when $\lambda = Cost_{X, b, 1}(\nu)$. Therefore, by performing a binary search over $\lambda$, one obtains accurate estimates of $w^*$ and $Cost_{X, b, 1}(\nu)$. 

However, a similar approach for the problem of compute, $Cost_{X, b, \ell} (\nu)$ results in the following SDP:

\begin{gather*}
    \max_{w_i \geq 0} \sum_{i = 1}^N w_i \\
    \text{s.t } \norm{\sum_{i = 1}^N w_i (X_i - \nu)(X_i - \nu)^\top}_\ell \leq \lambda \\
    w_i \leq b_i
\end{gather*}
which does not fall into the standard class of packing SDPs. We extend and generalize fast linear time solvers for packing/covering SDPs from \cite{peng2012faster} to this broader class of problems. However, this generalization is not straightforward. 

To demonstrate the main difficulties, we will delve more deeply into the solver from \cite{peng2012faster} and state the packing/covering primal dual pairs they consider:

\begin{table}[h]
\centering
\begin{tabularx}{\textwidth}{ >{\centering\arraybackslash}X | >{\centering\arraybackslash}X }
  Covering (Primal) & Packing   (Dual) \\
  \begin{gather*}
      \min_{M} \Tr{M} \\
      \text{Subject to: } \inp{A_i}{M} \geq 1\\
      M \succcurlyeq 0
  \end{gather*} & 
  \begin{gather*}
      \max_{w_i \geq 0} \sum_{i = 1}^N w_i \\
      \text{Subject to: } \sum_{i = 1}^N w_i A_i \preccurlyeq I
  \end{gather*}
\end{tabularx}
\end{table}
\noindent where $A_i \in \psd[m]$. The solver of \cite{peng2012faster} start by initializing a set of weights, $w_i$, feasible for \ref{eqn:packptz}. Subsequently, in each iteration, $t$, they first compute the matrix 

\begin{equation*}
 P_1 = \frac{\exp (\sum_{i = 1}^N w_i A_i)}{\Tr{\exp (\sum_{i = 1}^N w_i A_i)}}.
\end{equation*}

The algorithm then proceeds to increment the weights of all $i$ such that $\inp{P_1}{A_i} \leq (1 + \eps)$ for a user defined accuracy parameter, $\eps$, by a multiplicative factor. Intuitively, these indices correspond to ``directions'', $A_i$, along which $\sum_{i = 1}^N w_i A_i$ is small and therefore, their weights can be increased in the dual formulation. By incorporating a standard regret analysis from \cite{DBLP:journals/jacm/AroraK16} for the matrices, $P_1$, they show that one either outputs a primal feasible, $M$, with $\Tr M \leq 1$ or a dual feasible $w$ with $\sum_{i = 1}^N w_i \geq (1 - \eps)$.

The construction of our solver follows the same broad outline as in \cite{peng2012faster}. While the regret guarantees we employ are a generalization of those used in \cite{peng2012faster}, they still follow straightforwardly from standard regret bounds for mirror descent based algorithms \cite{DBLP:journals/corr/abs-1909-05207}. Instead, the main challenge of our solver is computational. The matrix $P^{(t)}$ can be viewed as a maximizer to $f(X) = \inp{X}{F} - \inp{X}{\log X}$ where $F = \sum_{i = 1}^N w_i A_i$ in $\mc{F}_1$. In our setting, we instead are required to compute the maximizer of $f$ in $(\mc{F}_\ell / \ell)$ which we show is given by the following: Let $H$ and $\tau^*$ be defined as:

\begin{equation*}
    H = \exp(F) = \sum_{i = 1}^m \lambda_i u_i u_i^\top \text{ with } \lambda_1 \geq \lambda_2 \ldots \geq \lambda_m > 0 \text{ and } \tau^* = \max_{\tau > 0} \lbrb{\frac{\tau^*}{\sum_{i = 1}^m \min (\tau, \lambda_i)} = \frac{1}{\ell}}.
\end{equation*}
Then, we have:
\begin{equation*}
    P_\ell = \argmax_{\ell X \in \mc{F}_\ell} f(X) = \frac{1}{\sum_{i = 1}^m \min(\lambda_i, \tau^*)}\cdot \sum_{i = 1}^m \min(\lambda_i, \tau^*) u_iu_i^\top.
\end{equation*}
While the matrix, $P_1$, can be efficiently estimated by Taylor series expansion of the exponential function (see \cite{DBLP:journals/jacm/AroraK16}), we need to estimate $P_\ell$ which is given by a careful truncation operation on $\exp (F)$. Note that given access to the exact top-$\ell$ eigenvectors and eigenvalues of $\exp(F)$, one can efficiently obtain a good estimate of $P_\ell$. However, the main technical challenge is establishing such good estimates given access only to \emph{approximate} eigenvectors and eigenvalues of a truncated Taylor series approximation of $H$. 

One of the main insights of our approach is that instead of analyzing the truncation operator directly, we instead view the matrix. $P_\ell$ as being the maximizer of $g(X, F) = \inp{F}{X} - \inp{X}{\log X}$ with respect to $X$. We then subsequently show that maximizer of $g(X, \wt{F})$ is close to $P_\ell$ for some $\wt{F}$ close to $F$ which makes crucial use of the fact that $\log X$ is operator monotone. Our second main piece of insight is that if our approximate eigenvectors and eigenvalues, denoted by $(\wh{\lambda_i}, v_i)_{i = 1}^\ell$ satisfy:

\begin{equation*}
    (1 - \eps) \sum_{i = 1}^l \wh{\lambda_i} v_i v_i^\top + \mc{P}_V^\perp H \mc{P}_V^\perp \preccurlyeq H \preccurlyeq (1 + \eps) \sum_{i = 1}^l \wh{\lambda_i} v_i v_i^\top + \mc{P}_V^\perp H \mc{P}_V^\perp,
\end{equation*}
where $V$ is the subspace spanned by the $v_i$ and $\mc{P}^\perp_V$ is the projection onto the orthogonal subspace of $V$, then our approximate truncation operator can be viewed as the \emph{exact} maximizer of $g(X, \wt{F})$ for some $\wt{F}$ close to $F$. From the previous discussion, this means that our truncation operator operating on the approximate eigenvectors $v_i$ is a good estimate of $P_\ell$. However, standard analysis of methods for the computation of eigenvalues and eigenvectors do not yield such strong guarantees \cite{DBLP:conf/nips/ZhuL16,DBLP:conf/nips/MuscoM15}. The final contribution of our work is a refined analysis of the power method that yields the required stronger guarantees which is formally stated in \cref{thm:main-svd-guarantee}.\footnote{While our guarantees scale with $1 / \eps$ as opposed to $1 / \sqrt{\eps}$ as in \cite{DBLP:conf/nips/MuscoM15}, we suspect this dependence may be improved using techniques from \cite{DBLP:conf/nips/MuscoM15}.}
\section{Preliminaries}
\subsection{Linear algebra}

The following can be found in \cite[Example 13(iii)]{Cha15}:
\begin{fact}\label{fact:log-op-monotone}
    Suppose $A$ and $B$ are positive semidefinite matrices such that $A\psdge B$, then $\log A \psdge \log B$.
\end{fact}



\subsection{Optimization}
\begin{definition}[Strongly convex]
    Let $\calS$ be any convex set.  We say a function $f:\calS\to \R$ is $\alpha$-strongly convex with respect to norm $\|\cdot\|_{\square}$ if
    \[
        f(y) \ge f(x) + \grad f(x)^{\top}(y-x) + \frac{\alpha}{2}\|y-x\|_{\square}^2.
    \]
    Similarly, we call $f$ $\alpha$-strongly concave with respect to $\|\cdot\|_{\square}$ if
    \[
        f(y) \le f(x) + \grad f(x)^{\top}(y-x) - \frac{\alpha}{2}\|y-x\|_{\square}^2.
    \]
\end{definition}

\begin{definition}
    We will use $\trosp$ to denote the set $\{X\psdge 0:\Tr(X)\le 1\}$ and $\trosp^m$ for its restriction to $m\times m$ matrices.
\end{definition}

\begin{definition}
    The \emph{von Neumann entropy} $\vNE:\trosp\to\R$ is defined as follows:
    \[
        \vNE(X) \coloneqq -\langle X,\log X\rangle.
    \]
\end{definition}
A key property of von Neumann entropy we use is:
\begin{fact}[{\cite[Corollary 3]{yu2013strong}}]
    \label{fac:vnconv}
    $\vNE$ is $1$-strongly concave with respect to the trace norm.
\end{fact}

\begin{definition}
    The \emph{quantum relative entropy} $\QRL(X,Y)$ where $X\in\trosp$ and $Y$ is any PSD matrix is defined as
    \[
        \QRL(X,Y) \coloneqq \langle X, \log X\rangle - \langle X, \log Y\rangle.
    \]
\end{definition}
The following is an immediate consequence of \pref{fac:vnconv}:
\begin{fact}    \label{fact:vndiv-conc}
    Let $P$ be any positive semidefinite matrix. The function $f(X)\coloneqq\QRL(X,P)$ is $1$-strongly convex on $\trosp$.
\end{fact}
We emphasize that we slightly deviate from the convention that von Neumann entropy and quantum relative entropy are defined only on PSD matrices of trace \emph{exactly} $1$.
\section{Algorithm for List-Decodable Mean Estimation} \label{sec:algorithms}
\label{sec:alg}
To begin our discussion on algorithms, we must first state the formal error guarantees of evaluating $Cost_{X,b,\ell}(\nu)$.

\begin{lemma} \label{lem:informalapprox}(Cost Approximation Error) There exists an algorithm $ApproxCost_{X,b,\ell}(\nu)$ that outputs a triple $(\theta,\bar w)$ for $\theta \in \R$, and $\bar w \in \Phi_b(1 - \delta)$ satisfying $\theta := \max_{M \in \mathcal{F}_\ell} f(\bar w,M) \leq \min\limits_{w \in \Phi_b(1)} \max\limits_{M \in \calF_\ell} f(w,M)$ with failure probability $\epsilon$.  Furthermore, the algorithm runs in time $\widetilde O(Nd\poly(\ell,\frac{1}{\delta},\log(\frac{1}{\epsilon}), \log(N + d)))$.  
\end{lemma}

\begin{remark} If $\delta = 0$ we would compute cost exactly.  This is computationally expensive so we take $\delta$ to be $0.01$.  For ease of reading, one can first set $\delta = 0$ with the understanding that the algorithmic lemmas succeed for small $\delta$.       
\end{remark}

\begin{definition}[Sanitizing Tuple]
    For a dataset $X = \{x_1,\dots,x_N\}$, an inlier set $I\subseteq [N]$ with $|I|=\alpha N$, budgets $b_1,\dots,b_N\in\left[0,\frac{2}{\alpha N}\right]$, we say that $(\hat{\mu}, \hat{w})$ is a \emph{sanitizing tuple} for $(X, I, b)$ if the it satisfies:
    \begin{enumerate}
        \item If $\norm{\hat{\mu} - \mu} \geq 2 \cdot 10^3 \frac{\sigma}{\sqrt{\alpha}}$, then $\sum_{i \in I} \hat{w}_i \leq \frac{\alpha}{4}$.
        \item For all $i$: $0\le \hat{w}_i \le b_i$, $\|\hat{w}\|_1\ge 0.5$.
    \end{enumerate}
\end{definition}

In a slight abuse of notation, we will often use ApproxCost$_{X,b,\ell}(\nu)$ to refer to the corresponding $\theta$.  We defer the proof of \pref{lem:informalapprox} to \cref{lem:minmaxmain}.  Our goal in this section will be to prove that the algorithm $OutputList(X,b)$ satisfies our list decoding guarantees and that the algorithm $DescendCost(X,b)$ outputs a sanitizing tuple $(\hu,\hat w)$.  

\begin{theorem} \label{thm:maintheorem} 
    (Descend Cost Sanitizes Dataset) Let $X = \{x_1,...,x_N\}$ be a dataset with an inlier set $I$ of size $|I| = \alpha N$ satisfying $Cov_I(x) \preceq \sigma^2I$. Let $\mu = \sum_{i \in I}x_i$.  Let $\sum_{i \in I}b_i \geq 1$ and $b_i \in [0, \frac{2}{\alpha N}]$.  Then for ApproxCost$_{X,b,\ell}(\cdot)$ satisfying the guarantees of \pref{lem:informalapprox}, DescendCost$(X,b)$ outputs a sanitizing tuple $(\hu,\hat w)$  with probability at least $1 - \frac{1}{d^{10}}$.
\end{theorem}

\begin{proof} Recall that when the algorithm terminates $\theta^{(t)}\le\sigma^2$ or $\theta^{(t+1)}\ge 0.5\cdot\theta^{(t)}$.  We prove that $(\hat{\mu}, \hat{w})$ is a sanitizing tuple when $\theta^{(t)}\le\sigma^2$ in \pref{lem:basecase}, and when $\theta^{(t+1)}\ge 0.5\cdot\theta^{(t)}$ in \pref{lem:dualprocedure}.
\end{proof}

\begin{restatable}[Main Corollary]{corollary}{maincorollary}
\label{cor:maincorollary} 
Let $X = \{x_1,...,x_N\}$ be a dataset with an inlier set $I$ of size $|I| = \alpha N$ satisfying $Cov_I(x) \preceq \sigma^2I$. Let $\mu = \sum_{i \in I}x_i$.  $OutputList(X,\alpha)$ returns a list $\calL$ of length $O(\frac{1}{\alpha})$ such that there exists $\hu^* \in \calL$ satisfying $\norm{\hu^* - \mu} \leq r\frac{\sigma}{\sqrt{\alpha}}$ with high probability $1 - \frac{1}{d^{10}}$
\end{restatable}

The proof of the corollary is elementary and similar to the proof of correctness for the inefficient algorithm.  For example, see \cref{sec:miscellaneous}.

\paragraph{ Algorithm Description} In the ensuing sections we will drop iteration indices when they are clear from context.  The OutputList algorithm takes as input $(X,\alpha)$ and iteratively runs the DescendCost subroutine.  DescendCost sanitizes the dataset and returns a tuple $(\hu, \hat w)$ comprised of a candidate mean $\hu \in \R^d$ and a set of weights $\hat w$ to be removed in the next iteration of OutputList.  DescendCost takes as input $(X,b)$ where $b$ is a weight upper bound vector initialized to be $b_i = \frac{2}{\alpha N}$ for all $i \in [N]$.  The inlier weight is initialized to be $\sum_{i \in \calI} b_i = 2$ as an overconservative way of dealing with the fact that up to an $\frac{\alpha}{4}$ fraction of the inlier weight can be removed per iteration of OutputList.  The first subroutine within DescendCost is a WarmStart procedure which takes in $(X,\alpha)$ and outputs $(\theta,\nu,\bar{w})$ where $ApproxCost_{X,b,\ell}(\nu) = \theta \leq O(d\sigma^2)$ and $\bar{w}$ is the set of dual weights. 

In the main loop of DescendCost we shift the dataset $X$ so that the origin is at $\nu^{(t)}$ and let $X^{(t)}$ denote this shifted dataset.  Let $\hat\Sigma$ be the weighted empirical covariance according to the weights $\bar{w}$ output by $ApproxCost_{X,b,\ell}(\nu^{(t)})$.  We let $V$ be the top $\ell := O(k)$ eigenspace of $\hat{\Sigma}$ and let $\Pi_V$ denote the corresponding projection operator.  If there exists a point $\nu^{(t+1)}$ in $V$ with a substantially smaller cost i.e $\theta^{(t+1)} \leq 0.5\theta^{(t)}$, then a brute force search through $V$ would lower cost at an exponential rate and in no more than $O(\log(d))$ iterations the cost would be smaller than $\sigma^2$.  The algorithm then terminates, outputting the final iterate $\nu^{(t)}$ and the dual weights $\bar{w}^{(t)}$.  It is simple to show that $(\nu^{(t)},\bar{w}^{(t)})$ is a sanitizing tuple which we prove in \pref{lem:basecase}.  

A brute force search through $V$ is inefficient.  Instead, we project the dataset $X^{(t)}$ onto $V$.  We evaluate the cost at $p = O(\frac{\log(d)}{\log(\frac{1}{1 - \alpha})})$ randomly chosen projected datapoints.  We choose $\nu^{(t+1)}$ to be the projected datapoint with the smallest cost.  Although the projected datapoints are by no means an exhaustive search of the subspace $V$, if this procedure fails to make sufficient progress i.e $\theta^{(t+1)} \geq 0.5\theta^{(t)}$ then a corresponding procedure uses the dual weights to output a sanitizing tuple.  

The weight removal procedure is as follows.   We sort the vectors $\{\Pi_V(x_i-\nu^{(T)})\}_{i=1}^N$ by euclidean norm.  This sorting determines an ordering of the weights $\bar{w}^{(T)}_1,...,\bar{w}^{(T)}_N$. We pass through the sorted list, and find the smallest $m \in [N]$ such that $\sum_{i=1}^m \bar{w}^{(T)}_i \geq 0.5$.  We set $\hat{w}_i = \bar{w}^{(T)}_i$ for $i = 1,...,m$ and $\hat{w}_i = 0$ for $i > m$.  Finally we output $(\nu^{(T)},\hat w)$ as the sanitizing tuple.  

\begin{algorithm}
\SetAlgoLined
\KwIn{Set of points  $X = \{x_1,...,x_N\}$ in $\R^d$, inlier fraction $\alpha \in [0,\frac{1}{2}]$}
    $b_i := \frac{2}{\alpha n} \text{ }\forall i \in [N], \;  L := \{ \}$\;
    \While{$\norm{b}_1 > 0$} {
    $(\hu, \hat w) =  \textsf{DescendCost}(X, b)$\;
    $L = L\cup\{\hu\}$\;
    $b =  b - \hat w$\;
    }
\KwOut{$L = \{\hu_1,...,\hu_q\}$ a list of $q \leq \frac{4}{\alpha}$ candidate means}

\caption{OutputList($X,\alpha$)}
\label{alg:outputlist}
\end{algorithm}

\begin{restatable}[Unsuccessful Descent Implies Weight Removal Succeeds]{lemma}{dualprocedure}
\label{lem:dualprocedure} 
 \textsf{DescendCost}(X,b) is given dataset $X$ and weight upper bound $b$ satisfying the assumptions of \pref{thm:maintheorem}.  If at iteration $t$, $\theta^{(t+1)} \geq 0.5 \theta^{(t)}$ then $DescendCost(X,b)$ outputs  
$(\hu, \hat{w})$ satisfying $\|\hu - \mu\| \le r\frac{1}{\sqrt{\alpha}}$ or $\sum_{i \in I}\hat{w}_i \le \frac{\alpha}{4}$.
\end{restatable}

Proving \pref{lem:dualprocedure} is the primary objective of the remainder of this section.  We will make use of the following two lemmas and defer their proofs to the appendix. The first \pref{lem:descent} states that if $(\nu^{(t)}, \bar{w}^{(t)})$ is not a sanitizing tuple then the descent procedure succeeds.

\begin{restatable}[Unsanitized Tuple Implies Success in Descending Cost]{lemma}{descent}
\label{lem:descent}
Let $(X,b)$ satisfying the assumptions of \pref{thm:maintheorem}. If at iteration $t$ of DescendCost(X,b), the tuple $(\nu^{(t)}, \bar{w}^{(t)})$ is not a sanitizing tuple, then $\theta^{(t+1)} \leq 0.04\norm{\mu - \nu^{(t)}}^2 + c k\sigma^2$ for $c = 10^5$ with high probability $1 - \frac{1}{d^{10}}$.

\end{restatable}

We will also need \pref{lem:removal} which states that if the cost is a constant factor smaller than $\norm{\mu - \nu^{(t)}}$ then the weight removal procedure outputs a sanitizing tuple. 

\begin{restatable}[Small Cost Implies Weight Removal Succeeds]{lemma}{removal}
\label{lem:removal}
Let $(X,b)$ satisfy the conditions of \pref{thm:maintheorem}. Let $\nu^{(t)}$ be the iterate at iteration $t$.  If $ApproxCost_{X,b,\ell}(\nu^{(t)}) \leq \zeta \norm{\mu - \nu^{(t)}}^2$ for $\zeta \in [0,0.4]$.  Then the weight removal procedure outputs $\hat w$ satisfying $\norm{\hat w}_1 \geq 0.5$ and $\sum_{i \in I}\hat{w}_i \leq \frac{\alpha}{4}$. 

\end{restatable}

\begin{algorithm}
\SetAlgoLined
\KwIn{Dataset $X$, weight budget $b$}
\KwResult{A tuple $(\hu,\hat w)$ for $\hu \in \R^d$ and $\hat w_i \in [0,b_i]$ for all $i \in [N]$ and $\norm{\hat w}_1 \geq 0.5$}
\textbf{Fixed Constants: }$k = 1/\alpha,~\ell = 100\cdot k ,~p=\frac{10\log(d)}{\log(\frac{1}{1 - \alpha})}$\;
$\theta_0 = \infty$\;
$(\theta^{(1)}, \nu^{(1)}, \bar{w}^{(1)}) = \text{WarmStart}(X,\alpha)$\;
$t = 1$\;
\While{ $\theta^{(t)}  \geq \sigma^2 \text{ and } \theta^{(t)} \leq 0.5\cdot \theta^{(t-1)}$}{
    \text{Comment: Descent Procedure}\;
    $\hat{\Sigma}^{(t)} =  \frac{1}{\norm{\bar{w}^{(t)}}_1}\sum_{i\in [N]} \bar{w}_i^{(t)} \left(x_i - \nu^{(t)}\right)\left(x_i - \nu^{(t)}\right)^{\transp}$\;
    $V^{(t)} = \text{top }\ell\text{ eigenspace of } \hat{\Sigma}, \quad \Pi_{V^{(t)}} \text{ be associated projection operator}$\;
    $Y^{(t)} = \{\nu^{(t)}+\Pi_{V^{(t)}}(x_i-\nu^{(t)}):i\in [N]\}$\;
    Pick $y_1,\dots,y_p$ independently and uniformly from $Y^{(t)}$\;
    $(\theta^j,\bar{w}^j) = \text{ApproxCost}_{X,b,\ell}(y_j)$ for all $j \in [p]$\;
    $j^* = \argmin_{j \in [p]} \theta^j$\;
    $(\theta^{(t+1)},\bar{w}^{(t+1)}) = (\theta^{j^*},\bar{w}^{j^*})$\;
    $\nu^{(t+1)} = y_{j^*}$\;
    $t=t+1$\;
    }
\eIf{$\theta^{(t)} < \sigma^2$}{
\text{Comment: Base Case for Successful Descent}\;
 $\hat w = \bar{w}^{(t)}$\;
 $\hu = \nu^{(t)}$\;
}{
\text{Comment: Weight Removal Procedure for Failed Descent}\;
$T=t-1$\;
$\text{Let $x_1,\dots,x_n$ be an ordering of $X$ so that }$\;
$\left\|\Pi_{V^{(T)}}(x_1-\nu^{(T)})\right\| \le \left\|\Pi_{V^{(T)}}(x_2-\nu^{(T)})\right\| \le \ldots \left\|\Pi_{V^{(T)}}(x_n-\nu^{(T)})\right\|$\;
$i^* = \min \{i > 0: \sum_{j = 1}^i \bar{w}_i^{(T)} \geq 0.5\}$\;
$\hat{w}_i = \bm{1} \{i \leq i^*\} \bar{w}_i^{(T)}$\;
$\hu = \nu^{(T)}$\;
}
\KwOut{$(\hu,\hat{w})$}
\caption{DescendCost($X,b$)}
\label{alg:descendcost}
\end{algorithm}

Using the above two lemmas we prove \pref{lem:dualprocedure}.
\begin{proof}(Proof of \pref{lem:dualprocedure})
Firstly, we observe that at any given iteration, if $\norm{\nu^{(t)} - \mu} \leq r\frac{\sigma}{\sqrt{\alpha}}$ then either the descent makes progress or the weight removal procedure outputs a sanitizing tuple.  Likewise, if $\sum_{i \in I}\bar{w}^{(t)}_i \leq \frac{\alpha}{4}$ then either descent makes progress or weight removal outputs a sanitizing tuple.  So without loss of generality, we assume $\norm{\nu^{(t)} - \mu} \geq r\frac{\sigma}{\sqrt{\alpha}}$ and  $\sum_{i \in I}\hat{w}^{(t)}_i \geq \frac{\alpha}{4}$. 

Using \pref{lem:descent} for $\ell > 100k$ and $\norm{\mu - \nu^{(t)}} \geq r\frac{\sigma}{\sqrt{\alpha}}$ we have $\theta^{(t+1)} \leq 0.2\norm{\mu - \nu^{(t)}}^2$.  Thus for $\theta^{(t)} \geq 0.4\norm{\mu - \nu^{(t)}}^2$ we would make sufficient progress for $DescendCost(X,b)$ to progress to the next iteration.  Otherwise,the termination condition $\theta^{(t+1)} \geq 0.5 \theta^{(t)}$ implies $\theta^{(t)} \leq 0.4\norm{\mu - \nu^{(t)}}^2$.  Using \pref{lem:removal} we conclude that the weight removal procedure outputs a sanitizing tuple.  This concludes the proof of correctness for the $DescendCost$ algorithm.   
\end{proof}

\subsection{Analysis I: Descending Cost}
In this section we prove \cref{lem:descent}.  

\descent*
\begin{proof} 

Moving forward we drop the iteration indices unless otherwise specified. We let $\nu$ refer to $\nu^{(t)}$ and let $(\theta, \bar{w}) := ApproxCost_{X,b,k}(\nu^{(t)})$.  We define $\hat{\Sigma}$, $V$, and $\Pi_V$ accordingly. Note that since $\bar{w} \in \Phi_b(1 - \delta)$ we have $\theta \leq \norm{\hat\Sigma}_\ell \leq \frac{\theta}{1 - \delta}$. 

Our first step is to prove that $\mu$ has a large component in $V$ i.e
\begin{align} \label{eq:dualprocedureproject}
\|\Pi_{V}(\mu-\nu) \|^2 \ge (1 - \frac{4k}{\ell}) \|(\mu-\nu)\|^2 - c_1\ell\sigma^2    
\end{align}
For some fixed constant $c_1$.  By definition of projection we have 
\begin{align} \label{eq:dualprocedure1}
     \norm{\Pi_{V}(\mu - \nu)}^2 = \norm{\mu - \nu}^2 - \norm{\Pi^\perp_{V}(\mu - \nu)}^2
\end{align}
so it suffices to upper bound $\norm{\Pi^\perp_{V}(\mu - \nu)}^2$ by $\frac{4k}{\ell}\norm{\mu - \nu}^2$.  Let $\vphi \coloneqq \frac{\Pi_V^\perp(\mu-\nu)}{\|\Pi_V^\perp(\mu-\nu)\|}$.  Towards these ends, we observe the following 

\begin{align*}
    \norm{\mu - \nu}^2 \geq \theta - 10\ell\sigma^2 \geq (1 - \delta)\ell \cdot \sigma_\ell(\hat{\Sigma}) - 10\ell\sigma^2 \geq (1 - \delta)\ell \cdot \vphi^{\transp}\hat{\Sigma}\vphi - 10\ell\sigma^2
\end{align*}

Where the first inequality follows by \pref{lem:costupper}, the second inequality follows by the fact that $\theta \leq \norm{\hat\Sigma}_\ell \leq \frac{\theta}{1 - \delta}$, the third inequality is because the $\ell$'th singular value of an $\ell$-SVD captures greater variance in the covariance matrix than any vector orthogonal to the top $\ell$ eigenspace.  Let $w = \frac{1}{\norm{\bar w}_1} \bar w$. Continuing to lower bound, we obtain 
\begin{align*}
    &= (1 - \delta)\ell \cdot \sum_{i=1}^N w_i \langle \vphi, x_i-\nu\rangle^2- 10\ell\sigma^2
    \\& \ge (1 - \delta)\ell \cdot \sum_{i \in I} w_i \langle \vphi, x_i-\nu\rangle^2 - 10\ell\sigma^2
    \\&= (1 - \delta)\ell \cdot \sum_{i \in I}^N w_i \langle \vphi, x_i-\mu + \mu -\nu\rangle^2 - 10\ell\sigma^2
    \\&= (1 - \delta)\ell \cdot \sum_{i \in I}^N w_i \left(\langle\vphi,x_i-\mu\rangle^2 - 2\langle\vphi, \mu - x_i\rangle\langle\vphi,\mu-\nu\rangle + \langle\vphi,\mu-\nu\rangle^2\right)- 10\ell\sigma^2
\end{align*}
Here, the inequality follows by dropping squared terms. Now using the inequality $2ab \leq \frac{a^2}{c^2} + c^2b^2$ for $a = \langle \vphi, \mu - \nu\rangle$, $b = \langle \vphi, \mu - x_i\rangle$, $c = 10$ we obtain   
\begin{align*}
&\geq (1 - \delta)\ell \cdot \sum_{i \in I}w_i\left( 0.99\left\|\Pi_{V}^{\perp}(\mu-\nu)\right\|^2 - 99\langle\vphi,x_i-\mu\rangle^2\right) - 10\ell\sigma^2
\\&= (1 - \delta)\ell \cdot (\sum_{i \in I} w_i)(0.99\left\|\Pi_{V}^{\perp}(\mu-\nu)\right\|^2) - (1 - \delta)\ell \cdot 99\sum_{i \in I} w_i \langle \vphi, x_i - \mu \rangle^2 - 10\ell\sigma^2
\end{align*}

We use the fact that $\sum_{i \in I}w_i \geq \sum_{i \in I}\bar{w}_i \geq \frac{1}{4k}$ to lower bound the first term.  And we use the fact $w_i \leq \frac{2}{(1-\delta)\alpha N}$ to lower bound the second term to obtain 

\begin{align} \label{eq:dualprocedure2}
\geq 0.99\frac{\ell}{(1 - \delta)4k}\left\|\Pi_{V}^{\perp}(\mu-\nu)\right\|^2 - (1 - \delta)\ell \cdot 99\sum_{i \in I}\frac{2}{(1-\delta)\alpha N} \langle \vphi, x_i - \mu \rangle^2 - 10\ell\sigma^2
\end{align}

Consider the second term $\sum_{i \in I}\frac{1}{\alpha N} \langle \vphi, x_i - \mu \rangle^2$.  We can upper bound it by the fact that the inliers are covariance bounded $\sum_{i \in I}\frac{1}{\alpha N} \langle \vphi, x_i - \mu \rangle^2\leq \sigma^2$. Plugging this bound into \pref{eq:dualprocedure2} we obtain 

\begin{align*}
\norm{\mu - \nu}^2 \geq \pref{eq:dualprocedure2} \geq 0.99 \frac{\ell }{(1-\delta)4k}\left\|\Pi_{V}^{\perp}(\mu-\nu)\right\|^2 - 600\ell\sigma^2 - 10\ell\sigma^2 \geq \frac{\ell }{4k}\left\|\Pi_{V}^{\perp}(\mu-\nu)\right\|^2 - c_1\ell\sigma^2
\end{align*}

For a fixed constant $c_1$.  Rearranging the LHS and RHS we upper bound 
\begin{align} \label{eq:dualprocedure3}
 \left\|\Pi_{V}^{\perp}(\mu-\nu)\right\|^2 \leq \frac{4k}{\ell }\norm{\mu - \nu}^2 + c_1\ell\sigma^2     
\end{align}
Plugging into \pref{eq:dualprocedure1} we obtain 
$ \|\Pi_{V}(\mu-\nu) \|^2 \ge (1 - \frac{4k}{\ell}) \|(\mu-\nu)\|^2 - c_1\ell\sigma^2$ 
as desired.  

 Now consider a set of $M := \alpha N$ inliers $\calI$. given a rank $\ell$ projection of the inliers, the number of points within a $10\sigma\sqrt{\ell}$ ball of $\Pi_V(\mu - \nu)$ is larger than $\frac{99}{100}M$ by Chebyshev's inequality.  Selecting $p := O(\frac{\log(d)}{\log(\frac{1}{1 - \alpha})})$ random datapoints $R := \{x'_j\}_{j=1}^p \subset X$, there exists an $x' \in R$ satisfying 
 \begin{align} \label{eq:vempala-wang}
     \norm{ \Pi_V(x' - \mu)} \leq 10\sigma\sqrt{\ell}
 \end{align}
 with probability greater than $1 - \frac{1}{d^{10}}$.  Here we aim for a $\frac{1}{d^{10}}$ failure probability so that by union bound over the $O(\frac{1}{\alpha}\log(d))$ iterations of the algorithm we continue to succeed with high probability. Thus we have 

\begin{align*}
 \theta^{(t+1)} &= \min_{x_i \in R} \text{ApproxCost}(\nu + \Pi_V(x_i - \nu)) \\
 & \leq  \norm{\nu + \Pi_V(x' - \nu)  - \mu}^2 + 10\ell\sigma^2 \\
    &= \norm{\Pi_V(x' - \nu)  - (\mu - \nu)}^2 + 10\ell\sigma^2 \\ 
    &=  \norm{\Pi_V(x' - \nu) - \Pi_V(\mu - \nu) + \Pi_V(\mu - \nu)  - (\mu - \nu)}^2 + 10\ell\sigma^2\\
    &= \norm{\Pi_V(x' - \mu) - \Pi^\perp_V(\mu - \nu)}^2 + 10\ell\sigma^2
\end{align*}
Where the first equality is by definition of $\theta^{(t+1)}$, and the inequality follows by \pref{lem:costupper}.  The above is then:
\begin{multline*}
= \norm{\Pi_V (x' - \mu)}^2 +  \norm{\Pi^\perp_V(\mu - \nu)}^2 +  10\ell\sigma^2 \\
\leq  100\ell\sigma^2 + \frac{4k}{\ell}\norm{\mu - \nu}^2 + c_1\ell\sigma^2 + 10\ell\sigma^2 
= \frac{4k}{\ell}\norm{\mu - \nu}^2 + c_3\ell\sigma^2
\end{multline*}
 Here, the second inequality follows by applying \pref{eq:vempala-wang} to the first term and \pref{eq:dualprocedure3} to the second term.  The last equality follows for $c_3 = 1000$.  Thus we have concluded  $\theta^{(t+1)} \leq \frac{4k}{\ell}\norm{\mu - \nu^{(t)}}^2 + c\ell\sigma^2$ for some fixed constant $c$ with high probability $1 - \frac{1}{d^{10}}$ as desired .   
\end{proof}   

\subsection{Analysis II: Removing Weights}
In this section we prove \cref{lem:removal}
\removal*

\begin{proof}
 Let $w := \frac{\bar{w}}{\norm{\bar{w}}_1}$ which can be viewed as a probability distribution over $i \in [N]$ for which we can further define the expectation $\mathbb{E}_w[\cdot]$ and probability $\Pr_w[\cdot]$ associated with $w$.  We also use the notation $\Pr_{i \in I}[\cdot]$ and $\mathbb{E}_{i \in I}[\cdot]$ to denote fraction of inliers satisfying an event and the average of a function over the inliers. 
 
We need to prove two facts. Firstly, $\Pr_w[\norm{\Pi_V(x_i - \nu)}^2 < 0.8\norm{\nu - \mu}^2] \geq 0.5$, and secondly $\Pr_{i \in I}[\norm{\Pi_V(x_i - \nu)}^2 < 0.8\norm{\nu - \mu}^2] \leq \frac{\alpha}{4}$.  Taken together, this implies that a sort of the list $\{\norm{\Pi_V(x_i - \nu)}\}_{i = 1}^N$ succeeds in isolating at least $0.5$ weight where $\sum_{i \in I}\hat{w}_i \leq \frac{\alpha}{4}$.  We use Markov's inequality to prove the first statement.    

\begin{align} \label{eq:dualprocedure4}
\Pr_w[\norm{\Pi_V(x_i - \nu)}^2 < 0.8\norm{\nu - \mu}^2]  = 1 - \Pr_w[\norm{\Pi_V(x_i - \nu)}^2 \geq 0.8\norm{\nu - \mu}^2]
\geq 1 - \frac{ \mathbb{E}_{w}\|\Pi_V(x_i-\nu)\|^2}{0.8\norm{\nu - \mu}^2}
\end{align} 
We can upper bound the expectation by 
\begin{align}
    \mathbb{E}_{w}\|\Pi_V(x_i-\nu)\|^2 &= \sum_{i = 1}^N w_i\left\|\Pi_V(x_i-\nu)\right\|^2 = \trace\left(\Pi_V\hat{\Sigma}\Pi_V\right) = \sum_{i=1}^\ell \sigma_i\left(\hat{\Sigma}\right) = \norm{\hat\Sigma}_\ell
    \\ &\leq \frac{1}{1 - \delta}\theta \le \frac{0.4}{1 - \delta}\|\nu-\mu\|^2
\end{align}
Where the first inequality is  $\theta \leq \norm{\hat\Sigma}_\ell \leq \frac{\theta}{1 - \delta}$ and the second inequality is by assumption \newline $\theta^{(t)} \leq \zeta\norm{\mu - \nu}^2$ for $\zeta \leq 0.4$.
Plugging into \pref{eq:dualprocedure4} we obtain $\Pr_w[\norm{\Pi_V(x_i - \nu)}^2 < 0.8\norm{\nu - \mu}^2] \geq 0.5$

Now we prove the second statement $\Pr_{i \in I}[\norm{\Pi_V(x_i - \nu)}^2 < 0.8\norm{\nu - \mu}^2] \leq \frac{\alpha}{4}$.  Let $x \in I$ be an inlier.  Let $\rho := \frac{\Pi_V(\mu - \nu)}{\norm{\Pi_V(\mu - \nu)}}$.  We have that $\|\Pi_V(x-\nu)\|^2$ is lower bounded by 
\begin{align*}
\|\Pi_V(x-\nu)\|^2 &\geq \langle x-\nu,\rho\rangle^2 = \langle x-\mu,\rho\rangle^2 - 2\langle \mu - x, \rho\rangle\langle \mu-\nu,\rho\rangle + \langle\mu-\nu, \rho\rangle^2\\
&\ge 0.99\langle\mu-\nu, \rho\rangle^2 - 99\langle x-\mu,\rho\rangle^2  
\end{align*}

Where the first inequality follows because $\rho$ is a unit vector in $V$.  The second inequality follows by the fact that $2ab \leq \frac{a^2}{c^2} + c^2b^2$ for $a = \langle \mu - \nu,\rho\rangle$, $b = \langle \mu - x, \rho\rangle$, $c = 10$.  We further lower bound by
\begin{align*}
&\geq 0.99\|\Pi_{V}(\mu-\nu) \|^2 - 99\langle x-\mu,\rho\rangle^2\\
&\ge 0.99(1 - \frac{4k}{\ell}) \|\mu-\nu\|^2 - c_1\ell\sigma^2 - 99\langle x-\mu,\rho\rangle^2\\ 
&\ge  0.9\|\mu-\nu\|^2 - 99\langle x-\mu,\rho\rangle^2
\end{align*}
The first inequality follows by definition of $\rho$, the second inequality follows from \pref{eq:dualprocedureproject} \footnote{A fine point is that we need the assumption that $(\nu, \bar w)$ is not a sanitizing tuple to apply \pref{eq:dualprocedureproject}. This holds because if $(\nu, \bar w)$ were a sanitizing tuple, our lemma would be trivially true}.  Moving on, the last inequality follows from plugging in parameters $\ell = 100 k$ and the given $\|\mu-\nu\| \geq r \frac{\sigma}{\sqrt{\alpha}}$ for $r = 2\cdot 10^3$.    

Plugging this lower bound for $\|\Pi_V(x-\nu)\|^2 $ into  $\Pr_{i \in I}[\norm{\Pi_V(x_i - \nu)}^2 < 0.8\norm{\nu - \mu}^2]$ we obtain 
\begin{align*}
\Pr_{i \in I}[\norm{\Pi_V(x_i - \nu)}^2 < 0.8\norm{\nu - \mu}^2]
\leq  \Pr_{i \in I}[0.9\|(\mu-\nu)\|^2 - 100\langle x_i-\mu,\rho\rangle^2 < 0.8\norm{\nu - \mu}^2]  \\
\end{align*}
Rearranging we obtain 
\begin{align*}
= \Pr_{i \in I}[\langle x_i-\mu,\rho\rangle^2 > 10^{-3}\norm{\nu - \mu}^2] 
\end{align*}
Via Markov's inequality: 
\begin{align*}
    \Pr_{i \in I}[\langle x_i-\mu,\rho\rangle^2 > 10^{-3}\norm{\nu - \mu}^2] \leq \frac{\mathbb{E}_{i \in I}[\langle x_i-\mu,\rho\rangle^2] }{10^{-3}\norm{\nu - \mu}^2}
    \leq \frac{\sigma^2}{10^{-3}\norm{\nu - \mu}^2} \leq \frac{\alpha}{4}
\end{align*}
Where the second inequality follows from applying the bounded covariance of the inliers.  The last inequality follows from the assumption that  $\norm{\nu - \mu} \geq r\frac{\sigma}{\sqrt \alpha}$ for $r = 2\cdot 10^3$.  
\end{proof} 

 \section{Fantope optimization in nearly linear time}  
\label{sec:genpcsolver}

In this section, we will design a solver for solving the following class of generalized packing/covering SDPs that we will need to solve in the course of our algorithm.
\begin{table}[h]
\centering
\begin{tabularx}{\textwidth}{ >{\centering\arraybackslash}X | >{\centering\arraybackslash}X }
  Covering (Primal) & Packing   (Dual) \\
  \begin{gather*}
      \min_{M, N} \Tr{M} + \Tr{W} \\
      \text{Subject to: } \inp{M}{A_i} + \inp{N}{B_i} \geq 1\\
      M \succcurlyeq 0,\ W \succcurlyeq 0\\
      \norm{W} \leq \frac{\Tr{W}}{k}
  \end{gather*} & 
  \begin{gather*}
      \max_{w_i} \sum_{i = 1}^n w_i \\
      \text{Subject to: } \sum_{i = 1}^n w_i A_i \preccurlyeq I \\
      \norm*{\sum_{i = 1}^n w_i B_i}_k \leq k \\
      0 \leq w_i
  \end{gather*} \\
\end{tabularx}
\end{table}

In the above setup, we have $A_i, M \in \psd[l]$ and $B_i, W \in \psd[m]$ and furthermore, we will assume that the matrices are given to us in factorized form; that is, $A_i = C_i C_i^\top$ with $C_i \in \mathbb{R}^{l \times r_i}$ and $B_i = D_iD_i^\top$ with $D_i \in \mathbb{R}^{m \times s_i}$ which will allow us to bound the runtimes of our procedure. The analysis will proceed along the lines of \cite{peng2012faster} however, as opposed to the standard multiplicative weights regret bound from \cite{DBLP:journals/jacm/AroraK16}, we will need to use a new regret bound to account for the nature of the sets we have chosen. As in \cite{peng2012faster}, we restrict ourselves to solving the $\eps$-decision version of the above pair of problems:

\begin{problem}
    \label{prob:epsdec}
    Find either a dual feasible, $w$, with $\sum_{i = 1}^n w_i \geq (1 - \epsilon)$ or a dual feasible $(M, W)$, satisfying $\Tr M + \Tr W \leq 1 + \epsilon$. 
\end{problem}

\subsection{The Regret Guarantee}
\label{ssec:reg}

In this subsection, we will establish a regret guarantee useful for designing fast solvers for our class of SDPs. First, let $\mc{S}$ defined as:

\begin{equation*}
    \mc{S} = \{(M, W) \in (\psd[l], \psd[m]): M \succcurlyeq 0,\ W \succcurlyeq 0,\ \Tr{M} + \Tr{W} = 1,\ \norm{W} \leq \Tr{W} / k\}.
\end{equation*}

The game takes place over $T$ rounds where for each round $t \in 1, \dots, T$:

\begin{enumerate}
    \item The player plays two psd matrices $(M_t, W_t) \in \mc{S}$. 
    \item The environment then reveals two gain matrices $(F_t, G_t)$ with $\norm{F_t} \leq 1$ and $\norm{G_t} \leq 1$ and the player achieves a gain of $\inp{F_t}{M_t} + \inp{G_t}{W_t}$. 
    \item The goal of the player is to minimize their total regret:
    \begin{equation*}
        \mc{R}(T) = \max_{(M, W) \in \mc{S}} \sum_{t = 1}^T \inp{F_t}{M} + \inp{G_t}{W} - \sum_{t = 1}^T \inp{F_t}{M_t} + \inp{G_t}{W_t}.
    \end{equation*}
\end{enumerate}

We will first provide a regret guarantee for the following strategy where in each iteration $(M_t, W_t)$ are defined for $\eta > 0$ by:
\begin{equation}
    \label{eq:mdUp}
    (M_t, W_t) = \argmax_{(M, W) \in \mc{S}} \eta \sum_{i = 1}^{t - 1} \inp{F_i}{M} + \inp{G_i}{W} + \vNE (M) + \vNE (W)
\end{equation}
Before we move on to the regret bound, we will require \cref{fac:vnconv}:
\begin{corollary}
    \label{cor:vnconvgen}
    The function $f(M, W) = \vNE(M) + \vNE(W)$ is $1$-strongly concave with respect to the following norm:
    \begin{equation*}
        \norm{(M, W)}_{\mc{S}} = \norm{M}_\ast + \norm{W}_\ast
    \end{equation*}
    on the set $\mc{S}$.
\end{corollary}
\begin{proof}
    This is easily seem by embedding the set $\mc{S}$ in $\psd[l + m]$ by the following mapping 
    $$(M,W) \to \begin{bmatrix}M & 0 \\ 0 & W\end{bmatrix}$$
    and subsequently noting that the $\norm{\cdot}_{\mc{S}}$ and $f$ coincide with the trace norm and the von Neumann entropy on $\psd[l + m]$.
\end{proof}

We will now state a standard regret guarantee (See, for example, Theorem 5.2 from \cite{DBLP:journals/corr/abs-1909-05207}) for the update rule defined in Equation~\ref{eq:mdUp}:

\begin{lemma}
    \label{lem:regG}
    For a sequence of gain matrices, $\{(F_t, G_t)\}_{t = 1}^T$ satisfying $\norm{F_t} \leq 1$ and $\norm{G_t} \leq 1$, the update rule defined in Equation~\ref{eq:mdUp} satisfies:
    \begin{equation*}
        \max_{(M, W) \in \mc{S}} \sum_{t = 1}^T \inp{F_t}{M} + \inp{G_t}{W} - \sum_{t = 1}^T \inp{F_t}{M_t} + \inp{G_t}{W_t} \leq 2\eta (T) + \frac{\log (l + m)}{\eta}.
    \end{equation*}
\end{lemma}

\begin{proof}
    The lemma follows immediately from Theorem 5.2 in \cite{DBLP:journals/corr/abs-1909-05207}.
\end{proof}

We will use the following corollary in the analysis of our solver:

\begin{corollary}
    \label{cor:regapx}
    Let $\{(F_t, G_t)\}_{t = 1}^T$ be any sequence of gain matrices satisfying $\norm{F_t} \leq 1$ and $\norm{G_t} \leq 1$ and let $(M_t, W_t)$ be defined as in Equation~\ref{eq:mdUp} and suppose that $(\Mt{}, \Wt{})$ satisfy:
    \begin{equation*}
        \norm{\Mt{} - M_t} + \norm{\Wt{} - W_t} \leq \eps.
    \end{equation*}
    Then, we have:
    \begin{equation*}
        \max_{(M, W) \in \mc{S}} \sum_{t = 1}^T \inp{F_t}{M} + \inp{G_t}{W} - \sum_{t = 1}^T \inp{F_t}{\Mt{}} + \inp{G_t}{\Wt{}} \leq (2\eta + \eps) (T) + \frac{\log (l + m)}{\eta}.
    \end{equation*}
\end{corollary}

\begin{proof}
    The corollary follows from the fact that for each $1 \leq t \leq T$, we have:
    
    \begin{align*}
        \inp{F_t}{\Mt{}} + \inp{G_t}{\Wt{}} &= \inp{F_t}{M_t} + \inp{G_t}{W_t} + \inp{F_t}{\Mt{} - M_t} + \inp{G_t}{\Wt{} - W_t} \\
        &\geq \inp{F_t}{M_t} + \inp{G_t}{W_t} - \norm{\Mt{} - M_t}_{\ast} - \norm{\Wt{} - W_t}_{\ast} \\
        &\geq \inp{F_t}{M_t} + \inp{G_t}{W_t} - \eps
    \end{align*}
    
    where the first inequality follows from Matrix-H\"olders inequality. 
\end{proof}

\subsection{Analysis of the Solver}
\label{ssec:solvana}

In this subsection, we formally introduce our solver and incorporate the regret analysis from the previous subsection into its analysis. We first introduce the following notation:

\begin{align*}
    \psit{} &= \sum_{i \in [n]} \wts_i A_i\qquad \phit{} = \sum_{i \in [n]} \wts_i B_i && \text{Weighted Constraint matrices}
\end{align*}

Our algorithm and its subsequent analysis follow along the lines of \cite{peng2012faster}:

\begin{algorithm}[H]
\SetAlgoLined

\KwIn{Constraint Matrices $A_i = C_iC_i^\top$ \text{ and }  $B_i = D_iD_i^\top$, Tolerance $\eps$, Failure Probability $\delta$}
\KwOut{Either primal feasible, $(M^*, W^*)$ or dual feasible $x^*$ satisfying \cref{prob:epsdec}}

$K \gets \frac{1 + \log (n + l + m)}{\eps}, \alpha \gets \frac{\epd}{(1 + 10 \eps)Kk}, R \gets \frac{512 \log (n + l + m) K k}{\epd\eps}, \eps^\prime \gets \frac{\eps^2}{2048 k \log (n + l + m)} $, $\dd \gets \delta / (5R)$\;
$t \gets 0, x_i^0 \gets 1 / (n(\Tr A_i + \Tr B_i))$

\While{$\norm{\wts{}}_1 \leq K$ and $t < R$} {
    $t \gets t + 1$\;
    $\omt{} \gets (\psit[t-1] - \psit[0]),\ \thet{} \gets (\phit[t-1] - \phit[0])$\;
    $(\gt{}, \bt{}, \bpt{}, \taut, \{(\st_i, \vt_i)\}_{i = 1}^k) \gets \projalg (\omt{}, 2Kk, \thet{}, \epd, \dd)$ \label{lin:proj}\;
    $\Vt{} \gets [\vt_1, \dots, \vt_k]$\;
    
    $\yt{} \gets \ipest(\omt{}, 2Kk, \{A_i = C_i C_i^\top\}, \epd, \dd)$\;\ 
    $\zt{} \gets \ipest(\thet{}, 2Kk, \{\mc{P}_{\Vt{}}^\perp D_i D_i^\top \mc{P}_{\Vt{}}^\perp\}, \epd, \dd)$\;
    
    $\St{} \gets \{i \in [n]: \gt{} \yt_i + \bpt{} \zt_i + \bt{} \sum_{j = 1}^k \min(\st_j, \taut{}) (\vt_j)^\top D_iD_i^\top \vt_j \leq (1 + \eps)\}$\label{lin:iprod} \;
    $\wts \gets \wts[t - 1] + \alpha \wts[t - 1]_{\St{}}$\;
}

\uIf{$\norm{\wts}_1 \geq K$} {
    \KwRet{$w^* = \frac{1}{(1 + 10 \eps)K} \cdot \wts{}$ as dual feasible point}\;
}
\Else{
    $M^* \gets t^{-1} \sum_{i = 1}^t \gt{} \exp (\psit[t - 1])$\;
    $W^* \gets t^{-1} \sum_{i = 1}^t (\bt{}\sum_{i = 1}^k \min (\st_i, \taut{}) \vt_i (\vt_i)^\top + \bpt{} \mc{P}_{\Vt{}}^\perp \exp (\phit[t - 1]) \mc{P}_{\Vt{}}^\perp)$\;
    \KwRet{$(M^*, W^*)$ as primal feasible point}\;
}
\caption{PackingCoveringSolver}
\label{alg:packsol}
\end{algorithm}

For the rest of the proof, we will assume that the algorithm terminates at the end of the $T^{th}$ loop for some $T \leq R$. For ease of exposition, we now define the following variables:

\begin{align*}
    \Ft{} &= \frac{\alpha}{\epd} \sum_{i \in \St{}} \wts[t - 1]_i A_i\qquad \Gt{} = \frac{\alpha}{\epd} \sum_{i \in \St{}} \wts[t - 1]_i B_i && \text{Gain matrices} \\
    \dt{} &= \alpha \wts[t-1]_{\St{}} \implies \wts = \wts[0] + \sum_{i = 1}^t \dt[i] && \text{Update variables}  \\
    \Mt{} &= \gt{} \exp (\psit[t - 1]) && \text{Approximate $M$ Projection} \\
    \Wt{} &= \bt{}\sum_{i = 1}^k \min (\taut{}, \st{}) \vt_i(\vt_i)^\top + \bpt{} \mc{P}_{\Vt{}}^\perp \exp (\phit[t - 1]) \mc{P}_{\Vt{}}^\perp && \text{Approximate $W$ Projection} \\
    (\wMt{}, \wWt{}) &= \argmax_{(M, W) \in \mathcal{S}} \inp{\omt{}}{M} + \inp{\thet{}}{W} + \vNE(M) + \vNE(W) && \text{True Projections}
\end{align*}

Note that $\Mt{}$ and $\Wt{}$ are meant to be approximations to $\wMt{}$ and $\wWt{}$ respectively and the correctness of these projections is guaranteed by \cref{thm:fast-projection-main}. Also, observe that $\omt{} = \epd\sum_{i = 1}^{t - 1}\Ft[i]$ and $\thet{} = \epd\sum_{i = 1}^{t - 1}\Gt[i]$. In the next few lemmas proving the correctness of \cref{alg:packsol}, we will simplify presentation by making the following assumptions. We will prove in the main theorem of the section that these assumptions hold with the desired probability.

\begin{assumption}
    \label{as:solvasm}
    We assume the following about the running of \cref{alg:packsol} for all $t \in [T]$:
    
    \begin{enumerate}
        \item The projections $(\Mt{}, \Wt{})$ satisfy $\norm{\Wt{}} \leq \Tr(\Wt{}) / k$ and satisfy:
        
        \begin{equation*}
            \norm{\Mt{} - \wMt{}}_\ast + \norm{\Wt{} - \wWt{}}_\ast \leq \epd.
        \end{equation*}
        
        \item The estimates, $\yt{}$ and $\zt{}$ satisfy for all $i \in [n]$:
        
        \begin{gather*}
            (1 - \epd) \inp{\Mt{}}{A_i} \leq \gt{} \yt_i \leq (1 + \epd) \inp{\Mt{}}{A_i} \\
            (1 - \epd) \inp{\Wt{}}{B_i} \leq \bpt{} \zt_i + \bt{} \sum_{j = 1}^k \min(\st_j, \taut{}) (\vt_j)^\top D_iD_i^\top \vt_j \leq (1 + \epd) \inp{\Wt{}}{B_i}.
        \end{gather*}
        
    \end{enumerate}
\end{assumption}

We also make the following non-probabilistic assumptions about the problem:
\begin{assumption}
\label{as:probasm}
    We assume that the problem instance satisfies:
    \begin{enumerate}
            
            \item For each, $i$, we have:
            
            \begin{equation*}
                \Tr (A_i) \leq (l + m + n)^5 \text{ and } \Tr (B_i) \leq (l + m + n)^5.
            \end{equation*}
            
            \item $T > 0$
    \end{enumerate}
\end{assumption}

The following three claims are analogues of Claims 3.3-3.5 from \cite{peng2012faster}:

\begin{claim}
    \label{clm:init}
    Assume \cref{as:probasm}. Then, we have:
    
    \begin{equation*}
        \norm{\psit[0]} \leq 1 \qquad \norm{\phit[0]}_k \leq k.
    \end{equation*}
\end{claim}

\begin{proof}
    We have:
    \begin{align*}
        \norm{\psit[0]} &\leq \Tr{\psit[0]} \leq \sum_{i = 1}^n \frac{1}{n \Tr A_i} \Tr A_i = 1. \\
        \norm{\phit[0]}_k &\leq \Tr{\phit[0]} \leq \sum_{i = 1}^n \frac{1}{n \Tr B_i} \Tr B_i = 1.
    \end{align*}
\end{proof}

\begin{claim}
    \label{clm:deltabnd}
    Assume \cref{as:solvasm,as:probasm}. Then, for $t = 1, \dots, T$:
    \begin{equation*}
        \inp{\Ft{}}{\Mt{}} + \inp{\Gt{}}{\Wt{}} \leq \frac{1 + 2\eps}{\epd} \cdot \norm{\dt{}}_1.
    \end{equation*}
\end{claim}

\begin{proof}
    We have from the definition of $\St$ and \cref{as:solvasm}:
    \begin{align*}
        \inp{\Ft{}}{\Mt{}} + \inp{\Gt{}}{\Wt{}} &= \frac{\alpha}{\epd} \sum_{i \in \St{}} \wts[t - 1]_i (\inp{A_i}{\Mt{}} + \inp{B_i}{\Wt{}}) \leq \frac{(1 + \eps)}{\epd} \sum_{i \in \St{}} \dt_i \\
        &\leq (1 - \epd)^{-1} \frac{1 + \eps}{\epd} \cdot \norm{\dt{}}_1 \leq \frac{(1 + 2 \eps)}{\epd} \cdot \norm{\dt{}}_1.
    \end{align*}
\end{proof}

\begin{claim}
    \label{clm:dualbnd}
    Assume \cref{as:solvasm,as:probasm}. Then, for $t = 0, \dots, T$:
    \begin{equation*}
        \norm{\wts{}}_1 \leq (1 + \eps)K.
    \end{equation*}
\end{claim}

\begin{proof}
    It suffices to prove the claim for $t = T$ as for $t < T$, the claim is true from the fact that the while loop continued till the next iteration. Now, we have from the fact that $\dt_i \leq \alpha \wts[t - 1]_i$:
    
    \begin{equation*}
        \sum_{i = 1}^n \wts[T]_i = \sum_{i = 1}^n (\wts[T - 1]_i + \dt[T]_i) \leq \sum_{i = 1}^n (1 + \alpha) \wts[T - 1]_i \leq (1 + \alpha) \norm{\wts[T - 1]}_1 \leq (1 + \eps) K.
    \end{equation*}
\end{proof}

We start with the following decomposition of $\psit{}$:
\begin{equation}
    \label{eq:psidec}
    \psit{} = \sum_{i = 1}^n \wts_i A_i = \sum_{i = 1}^n \wts[0]_i A_i + \sum_{i = 1}^n \sum_{j = 1}^t \dt[j]_i A_i = \sum_{i = 1}^n \wts[0]_i A_i + \sum_{j = 1}^t \sum_{i = 1}^n \dt[j]_i A_i = \psit[0] + \epd \sum_{j = 1}^t \Ft[j].
\end{equation}
Similarly, we get for $\phit{}$

\begin{equation}
    \label{eq:phidec}
    \phit{} = \sum_{i = 1}^n \wts_i B_i = \sum_{i = 1}^n \wts[0]_i B_i + \sum_{i = 1}^n \sum_{j = 1}^t \dt[j]_i B_i = \sum_{i = 1}^n \wts[0]_i B_i + \sum_{j = 1}^t \sum_{i = 1}^n \dt[j]_i B_i = \psit[0] + \epd \sum_{j = 1}^t \Gt[j].
\end{equation}

\begin{lemma}
\label{lem:specbnd}
    Under \cref{as:solvasm,as:probasm}, we have for every $t = 0, \dots, T$:
    \begin{equation*}
        \psit{} \preccurlyeq (1 + 10 \eps) K \qquad \phit{} \preccurlyeq (1 + 10 \eps) K \cdot k.
    \end{equation*}
\end{lemma}

\begin{proof}
    As in the proof of Lemma~3.2 in \cite{peng2012faster}, we will prove the claim via strong induction on $t$. We have from the definitions of $\alpha$ and $\epd$:
    \begin{align*}
        \Ft{} &= \frac{1}{\epd} \sum_{i = 1}^n \dt_i A_i \preccurlyeq \frac{\alpha}{\epd} \sum_{i = 1}^n \wts[t - 1]_i A_i = \frac{\alpha}{\epd} \psit[t - 1] \preccurlyeq I \\ 
        \Gt{} &= \frac{1}{\epd} \sum_{i = 1}^n \dt_i B_i \preccurlyeq \frac{\alpha}{\epd} \sum_{i = 1}^n \wts[t - 1]_i B_i = \frac{\alpha}{\epd} \phit[t - 1] \preccurlyeq I.
    \end{align*}
    We can now apply the results of \cref{cor:regapx} and the definition of $\St{}$ along with \cref{as:solvasm} to obtain:
    \begin{align*}
        \norm{\sum_{j = 1}^t \Ft[j]} &\leq \max_{(M, W) \in \mc{S}} \sum_{j = 1}^t \inp{\Ft[j]}{M} + \inp{\Gt[j]}{W} \leq 3\epd (T) + \frac{\log (l + m)}{\epd} + \sum_{j = 1}^T \inp{\Ft[j]}{\Mt[j]} + \inp{\Gt[j]}{\Wt[j]} \\
        &= 3\epd (T) + \frac{\log (l + m)}{\epd} + \frac{\alpha}{\epd} \sum_{j = 1}^T \inp{\sum_{i \in \St[j]} \wts[j-1]_i A_i}{\Mt[j]} + \inp{\sum_{i \in \St[j]} \wts[j-1]_i B_i}{\Wt[j]} \\
        &= 3\epd (T) + \frac{\log (l + m)}{\epd} + \frac{1}{\epd} \sum_{j = 1}^T \sum_{i \in \St[j]} \alpha \wts[j-1]_i  (\inp{A_i}{\Mt[j]} + \inp{B_i}{\Wt[j]}) \\
        &\leq 3\epd (T) + \frac{\log (l + m)}{\epd} + \frac{(1 + \eps)(1 + \epd)}{\epd} \sum_{j = 1}^T \sum_{i \in \St[j]} \dt[j]_i \\
        &= 3\epd (T) + \frac{\log (l + m)}{\epd} + \frac{(1 + 2\eps)}{\epd} \norm{\wts{}}_1.
    \end{align*}
    Similarly, we get:
    
    \begin{equation*}
        \frac{\norm{\sum_{j = 1}^t \Gt[j]}_k}{k} \leq \max_{(M, W) \in \mc{S}} \sum_{j = 1}^t \inp{\Ft[j]}{M} + \inp{\Gt[j]}{W} \leq 3\epd (T) + \frac{\log (l + m)}{\epd} + \frac{(1 + 2\eps)}{\epd} \norm{\wts{}}_1.
    \end{equation*}
    From the previous two inequalities, we get for $\psit{}$ from Equation~\ref{eq:psidec}:
    
    \begin{equation*}
        \norm{\psit{}} \leq \norm{\psit[0]} + \epd \norm{\sum_{j = 1}^t \Ft[j]} \leq (1 + 10\eps)K.
    \end{equation*}
    Finally, we get for $\phit{}$ from Equation~\ref{eq:phidec}:
    \begin{equation*}
        \norm{\phit{}}_k \leq \norm{\phit[0]}_k + \epd \norm{\sum_{j = 1}^t \Gt[j]}_k \leq (1 + 10\eps)K \cdot k.
    \end{equation*}
\end{proof}

\begin{lemma}
    \label{lem:priexit}
    Under \cref{as:solvasm,as:probasm}, \cref{alg:packsol} terminates with $\norm{\wts[R]}_1 \leq K$, we have for all $i \in [n]$:
    \begin{equation*}
        \inp{A_i}{M^*} + \inp{B_i}{W^*} \geq 1.
    \end{equation*}
\end{lemma}

\begin{proof}
    Suppose for the sake of contradiction, that there exists $i \in [n]$ such that:
    \begin{equation*}
        \inp{A_i}{M^*} + \inp{B_i}{W^*} < 1 \implies \frac{1}{R} \cdot \sum_{j = 1}^R \inp{A_i}{\Mt[j]} + \inp{B_i}{\Wt[j]} < 1.
    \end{equation*}
    Now, let $U$ denote the steps in algorithm where the dual variable, $\wts_i$ was incremented. From \cref{as:solvasm}, we get that $\wts_i$ is at least incremented for every iteration in the set $Y$ defined as:
    \begin{equation*}
        Y = \{j: \inp{A_i}{\Mt[j]} + \inp{B_i}{\Wt[j]} \leq (1 + \eps)(1 + \epd)^{-1}\}
    \end{equation*}
    By Markov's inequality and the definition of $\epd$, we must have $\abs{Y} \geq \frac{\eps}{2(1 + \eps)} \cdot R$. We must have as $\wts_i$ is incremented by a factor of $(1 + \alpha)$ each time:
    \begin{equation*}
        \wts[R]_i \geq \wts[0]_i (1 + \alpha)^{\abs{U}} \geq \wts[0]_i (1 + \alpha)^{\abs{Y}} \geq \wts[0]_i (1 + \alpha)^{\frac{\eps}{2(1 + \eps)} \cdot R} \geq \wts[0]_i \exp \lbrb{\frac{\alpha}{4} \abs{Y}} \geq (n + l + m)^4,
    \end{equation*}
    which is a contradiction. This concludes the proof of the lemma.
\end{proof}

\begin{lemma}
    \label{lem:asumpholds}
    Assume \cref{as:probasm}. Then, \cref{as:solvasm} holds in the running of \cref{alg:packsol} with probability at least $1 - \delta$. Furthermore, the total runtime of \cref{alg:packsol} is at most:
    \begin{equation*}
        O\lprp{(t_C + t_D + l + m) \poly\lprp{k, \frac{1}{\eps}, \log \frac{1}{\delta}, \log (l + m + n)}}
    \end{equation*}
    where $t_{C_i}$ and $t_{D_i}$ denote the time taken to compute one matrix vector multiplication with $C_i$ and $D_i$ respectively, $t_C = \sum_{i = 1}^n t_{C_i}$ and $t_D = \sum_{i = 1}^n t_{D_i}$.
\end{lemma}

\begin{proof}
    We will prove that \cref{as:solvasm} hold by induction on the number of steps of the Algorithm. Our induction hypothesis will be that \cref{as:solvasm} hold with probability $\dd (3t)$ up to iteration $t$. The hypothesis is trivially true at $t = 0$. Now, we will inductively prove that the assumptions hold true when $t = q + 1$ given that they hold at $t = 1 \dots q$. We start by computing a bound on the matrices $\omt{}$ and $\thet{}$. We have by the application of \cref{lem:specbnd} up to iteration $q$ that:
    \begin{equation*}
        \norm*{\omt[q + 1]} \leq \norm*{\psit[q]} \leq 2K.
    \end{equation*}
    Similarly, we have for $\thet{}$:
    \begin{equation*}
        \norm*{\thet[q + 1]} \leq \norm*{\phit[q]} \leq 2Kk.
    \end{equation*}
    Therefore, the upper bounds computed on $\norm{\omt[q + 1]}$ and $\norm{\omt[q + 1]}$ remain valid even in iteration $q + 1$. Therefore, conditioned on \cref{as:solvasm} holding true for iteration $q$, the conclusions of \cref{thm:fast-projection-main,lem:inpest} hold for Algorithms $\projalg$ and $\ipest$ for iteration $q + 1$ with probability at least $1 - 3\dd$. Hence, \cref{as:probasm} hold for iteration $q + 1$ with probability at least $(1 - 3\dd)(1 - 3\dd q) \geq 1 - (3\dd)(q + 1)$. 
    
    The runtime guarantees follow from the runtime guarantees in \cref{lem:inpest,thm:fast-projection-main} along with the fact that $2Kk$ is $O(\poly (\frac{1}{\eps}, \log (l + m + n), k))$ and matrix-vector multiplies with $\omt{}$ and $\thet{}$ can be implemented in time $O(t_C)$ and $O(t_D)$ respectively. And furthermore, a matrix vector product for all the $C_i$ and $\mc{P}_{\Vt{}}^\perp D_i$ required by \cref{lem:inpest} can be implemented in time $O(t_C)$ and $O(mk + t_D)$ respectively as for any vector $v$, computing $v^\top \mc{P}_{\Vt{}}^\perp$ takes $O(mk)$ time and subsequently, the resultant is multiplied with each of the $D_i$.
\end{proof}

We now conclude with the main theorem of the section.

\begin{theorem}
    \label{thm:solthm}
    There exists an Algorithm, $\solvalg$, which when given an instance of \cref{prob:epsdec}, with $A_i = C_i C_i^\top$, $B_i = D_iD_i^\top$, error tolerance $\eps \geq \frac{1}{n^2}$ and failure probability $\delta$, runs in time:
    
    \begin{equation*}
        O\lprp{(t_C + t_D + l + m) \poly\lprp{k, \frac{1}{\eps}, \log \frac{1}{\delta}, \log (l + m + n)}}
    \end{equation*}
    where $t_{C_i}$ and $t_{D_i}$ are the time taken to perform a matrix-vector product with $C_i$ and $D_i$ respectively and $t_C = \sum_{i = 1}^n t_{C_i}$ and $t_D = \sum_{i = 1}^n t_{D_i}$, and outputs a correct answer to \cref{prob:epsdec} with probability at least $1 - \delta$.
\end{theorem}

\begin{proof}
    We first discard $A_i$ and $B_i$ for those indices $i$ satisfying,
    
    \begin{equation*}
        \Tr (A_i) \geq (n + l + m)^5 \text{ or } \Tr (B_i) \geq (n + l + m)^5.
    \end{equation*}
    
    We will now run \cref{alg:packsol} instantiated with error parameter set to $\eps / 20$ and failure probability $\delta$. We first quickly address the case where $T = 0$. In this case, it must be that $\norm{\wts[0]}_1 > \frac{20(1 + \log (n + m + l))}{\eps}$. In this case, $w^*$ returned by the algorithm satisfies by definition $\sum_{i = 1}^n w_i \geq (1 - \eps/2)$ and furthermore, by \cref{clm:init}, is a valid dual solution. In this case, we can simply output $\hat{w} = w^*$ as a valid answer to \cref{prob:epsdec}.
    
    Now, after discarding the above two cases, we have that \cref{as:probasm} hold for the input passed to \cref{alg:packsol}. We have from \cref{lem:asumpholds} that \cref{alg:packsol} runs in time:
    
    \begin{equation*}
        O\lprp{(t_C + t_D + l + m) \poly\lprp{k, \frac{1}{\eps}, \log \frac{1}{\delta}, \log (l + m + n)}}
    \end{equation*}
    and that \cref{as:solvasm} hold in the running of \cref{alg:packsol} with probability at least $1 - \delta$. Conditioned on this event, we consider two possible cases:
    \begin{enumerate}
        \item The algorithm returns a dual solution, $w^\ast$.
        \item The algorithm returns a primal solution, $(M^\ast, W^\ast)$.
    \end{enumerate}
    In the first case, we have by \cref{lem:specbnd} and the definition of $w^\ast$ that $w^\ast$ is a feasible dual solution and furthermore, that $\sum_{i = 1}^n w^\ast \geq 1 - \eps / 2$ from our setting of the arguments to \cref{alg:packsol}. In this case, we simply define $\hat{w} = w^\ast$ for indices that are included in the input to \cref{alg:packsol} and $0$ for the discarded indices. Clearly, $\hat{w}$ is feasible dual solution to the original $\eps$-decision problem.

    In the second case, we construct a new primal solution, $(\wh{M}, \wh{W}) = (M^* + I / (n + l + m)^5, W^* / (n + l + m)^5)$. Note that for our bounds on $\eps$ and $n$, the trace of $(\wh{M}, \wh{W})$ from \cref{as:solvasm} is at most:
    \begin{align*}
        \Tr{\wh{M}} + \Tr{\wh{W}} &= \Tr{W^*} + \Tr{W^*} + 2 (n + l + m)^{-5} = T^{-1} \left(\sum_{t = 1}^T \Tr{\Mt{}} + \Tr{\Wt{}}\right) + 2 (n + l + m)^{-5} \\
        &\leq 1 + T^{-1} \lprp{\sum_{t = 1}^T \norm{\Mt{} - \wMt{}}_\ast + \norm{\Mt{} - \wMt{}}_\ast} \leq 1 + \eps
    \end{align*}
    and furthermore, from \cref{lem:priexit}, $(\wh{M}, \wh{W})$ satisfies all the primal constraints for the indices passed to \cref{alg:packsol} and finally for any discarded index, $i$, we have:
    \begin{equation*}
        \inp{A_i}{\hat{M}} + \inp{B_i}{\hat{W}} \geq \Tr (A_i) / (n + l + m)^5 + \Tr (B_i) / (n + l + m)^5 \geq 1.
    \end{equation*}
    Furthermore, from \cref{as:solvasm} since $W^*$ satisfied $\norm{W^*} \leq \Tr W^* / k$, we have 
    \begin{equation*}
        \norm{\hat{W}} = \norm{W^*} + (n + l + m)^{-5} \leq \frac{\Tr W^*}{k} + (n + l + m)^{-5} \leq \frac{\Tr \wh{W}}{k}.
    \end{equation*}
    Therefore, $(\hat{M}, \hat{W})$ is a valid primal solution to the original $\eps$-decision problem. Now, the run time guarantees follow from the fact that the run-time is dominated by the running of \cref{alg:packsol} and the probabilistic guarantees follow from the fact that \cref{alg:packsol} runs correctly with probability at least $1 - \delta$ as established previously. 
\end{proof}
\section{Power Method Analysis}
\label{sec:pow}

\begin{algorithm}[H]
\SetAlgoLined


\KwIn{PSD Matrix $A$, Accuracy $\eps$, Failure Probability $\delta$}
\KwOut{$\norm{v} = 1$}

$t \gets O((\log d + \log 1 / \delta + \log 1/\eps) / \eps)$\;

$\bg \gets \mc{N}(0, \Id)$\;
$v \gets A^t \bg / \norm{A^t \bg}$\;

\KwRet{$v$}
\caption{PowerMethod}
\label{alg:pow}
\end{algorithm}

\begin{algorithm}[H]
\SetAlgoLined

\KwIn{PSD Matrix $A$, Number of Components $m$, Accuracy $\eps$, Failure Probability $\delta$}
\KwOut{$\{(v_i, \lambda_i)\}_{i = 1}^m$ with $\wt{A} = \sum_{i = 1}^m \lambda_i v_iv_i^\top + \mc{P}_{V}^\perp A \mc{P}_{V}^\perp$ with $V = [v_1, \dots, v_m]$}

$A_0 \gets A$\;
\For{$i = 1:m$}{
    $v_i \gets \text{PowerMethod}(A_{i - 1}, \eps / (2m), \delta / (2m))$\;
    $\lambda_i \gets v_i^\top A_{i - 1}v_i$\;
    $A_i \gets \mc{P}_{v_i}^\perp A_{i - 1} \mc{P}_{v_i}^\perp$\;
}

\KwRet{$\{(v_i, \lambda_i)\}_{i = 1}^m$}
\caption{PCA}
\label{alg:pca}
\end{algorithm}

\noindent In this section $A$ is a $d\times d$ positive semidefinite matrix with eigenvalues $\lambda_1\ge\dots\ge\lambda_d \geq 0$ and a corresponding basis of eigenvectors $\phi_1,\dots,\phi_d$.  Let $\wt{\phi}_1,\dots,\wt{\phi}_{\ell}$ be an orthogonal basis of unit vectors obtained as the output of running the \cref{alg:pca} on matrix $A$, dimension $\ell\in[d]$, accuracy $\eps>0$, and failure probability $\delta>0$.  Let $\wt{A}$ be the matrix defined as follows:
\[
    \wt{A} \coloneqq \sum_{i=1}^{\ell}\langle\wt{\phi}_i,A\wt{\phi_i}\rangle\wt{\phi}_i\wt{\phi}_i^{\top} + \left(\Id-\sum_{i=1}^{\ell}\wt{\phi}\wt{\phi}^{\top}\right)A\left(\Id-\sum_{i=1}^{\ell}\wt{\phi}\wt{\phi}^{\top}\right)
\]
In this section we will prove:
\begin{theorem} \label{thm:main-svd-guarantee}
    With probability at least $1-O(\ell\delta)$, $(1-\eps)^{\ell}\wt{A}\psdle A\psdle (1+\eps)^{\ell}\wt{A}$.
\end{theorem}
We will first prove \cref{thm:main-svd-guarantee} when $\ell=1$ and then use it to prove the theorem for general $\ell$. Thus, for the rest of this section we will assume $\wt{A}$ is equal to:

\[
    \langle\wt{\phi_1},A\wt{\phi}_1\rangle\wt{\phi}_1\wt{\phi}_1^{\top} + (\Id-\wt{\phi}_1\wt{\phi}_1^{\top})A(\Id-\wt{\phi}_1\wt{\phi}_1^{\top}).
\]

We will now prove the following lemma.
\begin{lemma}   \label{lem:top-sing-vector}
    Suppose $\ell = 1$, with probability at least $1-O(\delta)$, $(1-\eps)\wt{A}\psdle A\psdle (1+\eps)\wt{A}$.
\end{lemma}
Proving \cref{lem:top-sing-vector} amounts to showing for any $x\in\R^d$:
\[
    |x^{\top}(A-\wt{A})x| \le \eps x^{\top}Ax. \numberthis \label{eq:goal-svd-rayleigh}
\]
Thus, we analyze the left hand side of the above expression.  A short calculation reveals that
\[
    A-\wt{A} = (\Id-\wt{\phi}_1\wt{\phi}_1^{\top})A\wt{\phi}_1\wt{\phi}_1^{\top} + \wt{\phi}_1\wt{\phi}_1^{\top}A(\Id-\wt{\phi}_1\wt{\phi}_1^{\top}).
\]
Let $x$ be any vector in $\R^d$.  We write $x$ and $\wt{\phi}_1$ in the basis of eigenvectors of $A$ as follows:
\begin{align*}
    \wt{\phi}_1 &= \sum_{i=1}^d c_i\phi_i\\
    x &= \sum_{i=1}^d \alpha_i\phi_i.
\end{align*}
Then
\begin{align*}
    \frac{x^{\top}(A-\wt{A})x}{2} &= (x^{\top}A\wt{\phi}_1\wt{\phi}_1^{\top}x - x^{\top}\wt{\phi}_1\wt{\phi}_1^{\top}A\wt{\phi}_1\wt{\phi}_1^{\top}x)\\
    &= \left(\sum_{i=1}^d\alpha_i c_i\right)\left(\sum_{i=1}^d\lambda_ic_i\alpha_i\right) - \left(\sum_{i=1}^d\alpha_ic_i\right)^2\left(\sum_{i=1}^d\lambda_ic_i^2\right)\\
    &= \left(\sum_{i=1}^d\alpha_i c_i\right)\left(\sum_{i=1}^d\lambda_ic_i\alpha_i-\left(\sum_{i=1}^d\alpha_ic_i\right)\left(\sum_{i=1}^d\lambda_ic_i^2\right)\right)\\
    &= \left(\sum_{i=1}^d\alpha_i c_i\right)\left(\sum_{i=1}^dc_i\alpha_i\left(\lambda_i- \sum_{j=1}^d \lambda_jc_j^2\right)\right)\\
    &= \left(\sum_{i=1}^d\alpha_i c_i\right)\left(\sum_{i=1}^dc_i\alpha_i\left(\lambda_i- \wt{\phi}_1^\top A\wt{\phi}_1\right)\right) \numberthis \label{eq:rayleigh-diff}
\end{align*}
So far we have not used the fact that $\wt{\phi}_1$ is the output of \cref{alg:pca}.  In particular, our progress so far which is recorded in \cref{eq:rayleigh-diff} holds true for \emph{arbitrary} $\wt{\phi}_1$.  We now discuss and prove the relevant properties of $\wt{\phi}_1$ we use for showing \cref{lem:top-sing-vector} (more specifically, for showing \cref{eq:goal-svd-rayleigh}).

\subsection{Probability and norm bounds}
Recall that $\phi_1,\dots,\phi_d$ is the basis of eigenvectors of $A$.
\begin{definition}
    We call a vector $g\in\R^d$ \emph{$\delta$-tempered} if $|\langle g, \phi_1\rangle|\ge\delta$.
\end{definition}

\begin{proposition} \label{prop:gauss-temper}
    Let $\bg\sim\calN(0,\Id)$.  $\bg$ is $\delta$-tempered except with probability $C\delta$ for some absolute constant $C$.
\end{proposition}
\begin{proof}
    $\langle\bg,\phi_1\rangle$ is distributed as a scalar standard Gaussian random variable and hence
    \[
        \Pr[\langle\bg,\phi_1\rangle\in[-\delta,\delta]]\le\frac{1}{\sqrt{2\pi}}\delta.
    \]
\end{proof}
The following can be found in \cite[Theorem 2.1.12]{Tao11}:
\begin{lemma}   \label{lem:gauss-conc}
    Let $\bg\sim\calN(0,\Id_d)$.  Then except with probability $C\exp(-ct^2)$ for absolute constants $C,c>0$,
    \[
        \|\bg\| \in [\sqrt{d}-t,\sqrt{d}+t].
    \]
\end{lemma}
From \cref{alg:pow} $\wt{\phi}_1$ is equal to $\frac{A^{t}g}{\|A^tg\|}$ for $t=\Theta\left(\frac{\log d+\log\frac{1}{\delta}+\log\frac{1}{\eps}}{\eps}\right)$.   Additionally, from \cref{prop:gauss-temper}, \cref{lem:gauss-conc} $g$ is a $\delta$-tempered vector and
\[
    \|g\|\in\left[\sqrt{d}-\zeta\sqrt{\log(1/\delta)}, \sqrt{d}+\zeta\sqrt{\log(1/\delta)}\right]
\]
except with probability $C\delta$ for some absolute constant $\zeta > 0$.  For the rest of this section we will assume that $g$ is indeed $\delta$-tempered and has norm in the above range.  For $S\subseteq\R$ let $\Pi_{S}$ be the projection matrix onto the eigenspace of eigenvalues in $S$.  In particular,
\[
    \Pi_S = \sum_{i:\lambda_i\in S} \phi_i\phi_i^{\top}.
\]
\begin{proposition}\label{prop:small-proj-low-eigenspace}
    $\|\Pi_{[0,(1-\eps)\lambda_1]}\wt{\phi}_1\|\le\eps$.
\end{proposition}
\begin{proof}
    We start by expressing $g$ in the basis $\{\phi_1,\dots,\phi_d\}$ as
    \[
        g = \sum_{i=1}^d \hat{g}_i\phi_i,
    \]
    which means
    \[
        A^t g = \sum_{i=1}^d \lambda_i^t\hat{g}_i\phi_i.  
    \]
    Thus,
    \begin{align*}
        \|\Pi_{[0,(1-\eps)\lambda_1]}\wt{\phi_1}\|^2 &= \frac{\sum_{i:\lambda_i\le(1-\eps)\lambda_1}\lambda_i^{2t}\hat{g}^2}{\|A^t g\|^2}\\
        &\le \frac{\lambda_1^{2t}\sum_{i:\lambda_i\le(1-\eps)\lambda_1}\hat{g}_i^2}{\lambda_1^{2t}\hat{g}_1^2} \\
        &\le (1-\eps)^{2t}\frac{\|g\|^2}{\delta^2}\\
        &\le (1-\eps)^{2t}\frac{2d+2\zeta^2\log(1/\delta)}{\delta^2}
    \end{align*}
    Since $t=L\frac{\log d+\log\frac{1}{\delta}+\log\frac{1}{\eps}}{\eps}$, we can choose constant $L$ large enough so that the above is bounded by $\eps^2$.
\end{proof}

Let $S_\eps = \{i: \lambda_i \geq (1 - \eps) \lambda_1\}$ and $T = \{i: \lambda_i > 0\}$. We now establish the following result:

\begin{proposition}
    \label{prop:small-ev-bnd}
    We have:
    \begin{equation*}
        \sum_{i \in T \setminus S_\eps} c_i^2 \cdot \abs*{\frac{\lambda_1}{\lambda_i}} \leq \frac{\eps^2}{d^2}.
    \end{equation*}
\end{proposition}

\begin{proof}
    First fix one particular $i \in T \setminus S_\eps$ and a in the proof of \cref{prop:small-proj-low-eigenspace}:
    
    \begin{align*}
        c_i^2 \abs*{\frac{\lambda_1}{\lambda_i}} \leq \abs*{\frac{\hat{g}_i}{\hat{g}_1}}^2 \abs*{\frac{\lambda_i}{\lambda_1}}^{2t - 1} \leq C\lprp{\frac{\sqrt{d} + \zeta \sqrt{\log 1 / \delta}}{\delta}}^2 (1 - \eps)^{2t - 1} \leq \frac{\eps^2}{d^3}
    \end{align*}
    
    where the first inequality follows from the fact that $\abs{c_i} \leq \abs*{\frac{\lambda_i^t \hat{g}_i}{\lambda_1^t \hat{g}_1}}$, the second inequality follows from the assumption that $\hat{g}$ is $g$-tempered and has a bounded norm and the final inequality from our definition of $t$. By summing up over all the terms, the statement of the proposition follows.
\end{proof}

\begin{proposition} \label{prop:approx-eigval-top-sing}
    $(1-2\eps)\lambda_1 \le \wt{\phi}_1^{\top}A\wt{\phi}_1 \le \lambda_1$.
\end{proposition}
\begin{proof}
    The upper bound follows from $\lambda_1$ being the maximum eigenvalue of $A$.  As a consequence of \cref{prop:small-proj-low-eigenspace} and the fact that $\|\wt{\phi}_1\|=1$,
    \[
        \|\Pi_{[(1-\eps)\lambda_1,\lambda_1]}\wt{\phi}_1\|\ge \sqrt{1-\eps^2}.
    \]
    and thus
    \[
        \wt{\phi}_1^{\top}A\wt{\phi}_1 \ge (1-\eps)\lambda_1\left(1-\eps^2\right) \ge (1-2\eps)\lambda_1.
    \]
\end{proof}

\begin{remark}  \label{rem:eig-lower-bound}
    The same proof as \cref{prop:approx-eigval-top-sing} also shows that:
    \[
        \wt{\phi}_k^{\top} A\wt{\phi}_k \ge (1-2\eps)\left\|\left(\Id-\sum_{i=1}^{k-1}\wt{\phi}_i\wt{\phi}_i^{\top}\right)A\left(\Id-\sum_{i=1}^{k-1}\wt{\phi}_i\wt{\phi}_i^{\top}\right)\right\|
    \]
    and since $\left\|\left(\Id-\sum_{i=1}^{k-1}\wt{\phi}_i\wt{\phi}_i^{\top}\right)A\left(\Id-\sum_{i=1}^{k-1}\wt{\phi}_i\wt{\phi}_i^{\top}\right)\right\|$ is decreasing in $k$, it must be true that for $k=1,\dots,\ell$:
    \[
        \wt{\phi}_k^{\top} A\wt{\phi}_k \ge (1-2\eps)\left\|\left(\Id-\sum_{i=1}^{\ell}\wt{\phi}_i\wt{\phi}_i^{\top}\right)A\left(\Id-\sum_{i=1}^{\ell}\wt{\phi}_i\wt{\phi}_i^{\top}\right)\right\|.
    \]
\end{remark}

\subsection{Wrapup and proof of \cref{thm:main-svd-guarantee}}
In this section we will first prove \cref{lem:top-sing-vector} and then prove \cref{thm:main-svd-guarantee}.
\begin{proof}[Proof of \cref{lem:top-sing-vector}]
    As in \cref{prop:small-ev-bnd}, we will define the sets $S_{\eps}\coloneqq\{i:\lambda_i\ge(1-\eps)\lambda_1\}$ and $T = \{i: \lambda_i > 0\}$.
    Recall that it suffices to prove:
    \[
        |x^{\top}(A-\wt{A})x|\le 8\eps x^{\top}Ax.
    \]
    Note that we may write $\wt{\phi}_1 = \sum_{i \in T} c_i \phi_i$ as we run at least one iteration of the power method which ensures that $\wt{\phi}_1$ is in the row/column space of $A$. Therefore, we can assume that the sums in \cref{eq:rayleigh-diff} only go over the elements in $T$. Using this as our starting point, we have:
    \begin{align*}
        \frac{1}{2}|x^{\top}(A-\wt{A})x| &= \left|\left(\sum_{i\in T}\alpha_i c_i\right)\left(\sum_{i\in T} c_i\alpha_i\left(\lambda_i- \wt{\phi}_1^\top A\wt{\phi}_1\right)\right)\right| \\
        &\le \left|\sum_{i \in S_\eps} \alpha_ic_i \sum_{i\in S_{\eps}}\alpha_ic_i(\lambda_i-\wt{\phi}_1^{\top}A\wt{\phi}_1)\right| + \left|\sum_{i \in S_\eps}\alpha_ic_i \sum_{i\in T \setminus S_{\eps}}\alpha_ic_i(\lambda_i-\wt{\phi}_1^{\top}A\wt{\phi}_1)\right| \\
        &\qquad + \left|\sum_{i \in T \setminus S_\eps}\alpha_ic_i \sum_{i\in S_{\eps}}\alpha_ic_i(\lambda_i-\wt{\phi}_1^{\top}A\wt{\phi}_1)\right| + \left|\sum_{i \in T \setminus S_\eps}\alpha_ic_i \sum_{i\in T \setminus S_{\eps}} \alpha_ic_i(\lambda_i-\wt{\phi}_1^{\top}A\wt{\phi}_1)\right| \numberthis \label{eq:split-sum}
    \end{align*}
    
    We start by bounding the first term in \cref{eq:split-sum}:
    
    \begin{align*}
        \abs*{\sum_{i \in S_\eps} \alpha_ic_i \sum_{i\in S_{\eps}}\alpha_ic_i(\lambda_i-\wt{\phi}_1^{\top}A\wt{\phi}_1)} &\leq 2\eps \lambda_1 \lprp{\sum_{i \in S_\eps} \alpha_i c_i}^2\\
        &\leq 2\eps \lambda_1 \lprp{\sum_{i \in S_\eps} \frac{c_i^2}{\lambda_i}} \lprp{\sum_{i \in S_\eps} \alpha_i^2 \lambda_i} \\
        &\leq 2\frac{\eps}{1 - \eps} (\sum_{i \in S_\eps} c_i^2) \lprp{\sum_{i \in T} \alpha_i^2 \lambda_i}\\
        &\leq 3\eps x^\top A x
    \end{align*}
    where the first inequality follows from \cref{prop:approx-eigval-top-sing} and the definition of the set $S_\eps$, the second inequality follows from Cauchy-Schwarz, the third follows again from the definition of the set $S_\eps$ and the final inequality from the fact that $\sum_i c_i^2 = 1$. 
    
    For the next term in \cref{eq:split-sum}, we have:
    
    \begin{align*}
        \abs*{\sum_{i \in S_\eps}\alpha_ic_i \sum_{i\in T \setminus S_{\eps}}\alpha_ic_i(\lambda_i-\wt{\phi}_1^{\top}A\wt{\phi}_1)} &= \abs*{\sum_{i \in S_\eps}\alpha_ic_i} \cdot \abs*{\sum_{i\in T \setminus S_{\eps}}\alpha_ic_i(\lambda_i-\wt{\phi}_1^{\top}A\wt{\phi}_1)}\\
        &\leq \lambda_1 \lprp{\sum_{i \in S_\eps} \abs{\alpha_i c_i}} \lprp{\sum_{i \in T \setminus S_\eps} \abs{\alpha_i c_i}} \\
        &= \lprp{\sum_{i \in S_\eps} \sqrt{\lambda_1} \abs{\alpha_i c_i}} \lprp{\sum_{i \in T \setminus S_\eps} \sqrt{\lambda_1} \abs{\alpha_i c_i}}\\
        &\leq \lprp{\sum_{i \in S_\eps} \sqrt{\lambda_1} \abs{\alpha_i c_i}} \lprp{\sum_{i \in T \setminus S_\eps} \sqrt{\lambda_1} \abs{\alpha_i c_i}} \\
        &\leq \lprp{\sum_{i \in S_\eps} \alpha_i^2 \lambda_i}^{1/2} \lprp{\sum_{i \in S_\eps} c_i^2 \frac{\lambda_1}{\lambda_i}}^{1/2} \lprp{\sum_{i \in T \setminus S_\eps} \alpha_i^2 \lambda_i}^{1/2} \lprp{\sum_{i \in T\setminus S_\eps} c_i^2 \frac{\lambda_1}{\lambda_i}}^{1/2} \\
        &\leq x^\top A x (1 - \eps)^{-1} \frac{\eps}{d}\\
        & \leq \frac{\eps}{4} x^\top Ax
    \end{align*}
    where the second inequality follows from Cauchy-Schwarz and the final inequality follows from the definition of $S_\eps$ and \cref{prop:small-ev-bnd}. 
    
    For the third term in \cref{eq:split-sum}, we have:
    
    \begin{equation*}
        \abs*{\sum_{i \in T \setminus S_\eps}\alpha_ic_i \sum_{i\in S_{\eps}}\alpha_ic_i(\lambda_i-\wt{\phi}_1^{\top}A\wt{\phi}_1)} \leq \lambda_1 \lprp{\sum_{i \in T \setminus S_\eps} \abs{\alpha_i c_i}} \lprp{\sum_{i \in S_\eps} \abs{\alpha_i c_i}}
    \end{equation*}
    
    and the proof proceeds as before. For the final term, we have from Cauchy-Schwarz and \cref{prop:small-ev-bnd}:
    
    \begin{align*}
        \abs*{\sum_{i \in T \setminus S_\eps}\alpha_ic_i \sum_{i\in T \setminus S_{\eps}}\alpha_ic_i(\lambda_i-\wt{\phi}_1^{\top}A\wt{\phi}_1)} &\leq \lambda_1 \lprp{\sum_{i \in T \setminus S_\eps} \abs{\alpha_i c_i}} \lprp{\sum_{i \in T\setminus S_\eps} \abs{\alpha_i c_i}} \\
        &\leq \lprp{\sum_{i \in T \setminus S_\eps} c_i^2 \lprp{\frac{\lambda_1}{\lambda_i}}} \lprp{\sum_{i \in T \setminus S_\eps} \alpha_i^2 \lambda_i} \\
        &\leq x^\top Ax \cdot \frac{\eps^2}{d^2}.
    \end{align*}
    
    Putting the bounds on the four terms on \cref{eq:split-sum}, we get the desired result.
\end{proof}

\begin{proof}[Proof of \cref{thm:main-svd-guarantee}]
    Our proof proceeds by induction.  When $\ell=1$, \cref{lem:top-sing-vector} gives us the desired statement.  Suppose we wish to prove the statement for $\ell = m$ and suppose our goal statement is true for $\ell = m-1$.  Recall that the algorithm first computes vectors $\wt{\phi}_1,\dots,\wt{\phi}_{m-1}$, and then runs \pref{alg:pca} on
    \[
        E\coloneqq\left(\Id-\sum_{i=1}^{m-1}\wt{\phi}_i\wt{\phi}_i\right)A\left(\Id-\sum_{i=1}^{m-1}\wt{\phi}_i\wt{\phi}_i\right)
    \]
    to obtain $\wt{\phi}_m$.  Let us define:
    \[
        \wt{E}\coloneqq \langle\wt{\phi}_m,E\wt{\phi}_m\rangle\wt{\phi}_m\wt{\phi}_m^{\top}+ \left(\Id-\wt{\phi}_m\wt{\phi}_m^{\top}\right)E\left(\Id-\wt{\phi}_m\wt{\phi}_m^{\top}\right)
    \]
    From \cref{lem:top-sing-vector}:
    \[
        (1-\eps)\wt{E}\psdle E \psdle (1+\eps)\wt{E}    \numberthis \label{eq:psd-ind-hyp}
    \]
    We now observe that
    \[
        \langle\wt{\phi}_m,E\wt{\phi}_m\rangle = \langle\wt{\phi}_m,A\wt{\phi}_m\rangle
    \]
    since $\wt{\phi}_m$ is orthogonal to the space spanned by $\wt{\phi}_1,\dots,\wt{\phi}_{m-1}$, and further note that
    \[
        \left(\Id-\wt{\phi}_m\wt{\phi}_m^{\top}\right)E\left(\Id-\wt{\phi}_m\wt{\phi}_m^{\top}\right) = \left(\Id-\sum_{i=1}^{m}\wt{\phi}_i\wt{\phi}_i\right)A\left(\Id-\sum_{i=1}^{m}\wt{\phi}_i\wt{\phi}_i\right)
    \]
    which lets us rewrite $\wt{E}$ as
    \[
        \wt{E} = \langle\wt{\phi}_m,A\wt{\phi}_m\rangle\wt{\phi}_m\wt{\phi}_m^{\top}+ \left(\Id-\sum_{i=1}^m\wt{\phi}_i\wt{\phi}_i^{\top}\right)A\left(\Id-\sum_{i=1}^m\wt{\phi}_i\wt{\phi}_i^{\top}\right).
    \]
    Now,
    \[
        \wt{A} = \wt{E} + \sum_{i=1}^{m-1}\wt{\phi}_i\wt{\phi}_i^{\top}\langle\wt{\phi}_i,A\wt{\phi}_i\rangle.
    \]
    Define $F$ as
    \[
        F \coloneqq E + \sum_{i=1}^{m-1}\wt{\phi}_i\wt{\phi}_i^{\top}\langle\wt{\phi}_i,A\wt{\phi}_i\rangle.
    \]
    Adding the PSD inequality
    \[
        (1-\eps)\sum_{i=1}^{m-1}\wt{\phi}_i\wt{\phi}_i^{\top}\langle\wt{\phi}_i,A\wt{\phi}_i\rangle \psdle \sum_{i=1}^{m-1}\wt{\phi}_i\wt{\phi}_i^{\top}\langle\wt{\phi}_i,A\wt{\phi}_i\rangle \psdle \sum_{i=1}^{m-1}\wt{\phi}_i\wt{\phi}_i^{\top}\langle\wt{\phi}_i,A\wt{\phi}_i\rangle
    \]
    to \cref{eq:psd-ind-hyp} along with the definitions of $\wt{A}$ and $F$ gives us
    \[
        (1-\eps)\wt{A} \psdle F \psdle (1+\eps)\wt{A} \label{eq:penult-IH}
    \]
    From our induction hypothesis,
    \[
        (1-\eps)^{m-1}F\psdle A \psdle (1+\eps)^m F
    \]
    from which we can deduce
    \[
        (1-\eps)^m\wt{A}\psdle A \psdle (1+\eps)^m\wt{A}.
    \]
    Hence our induction is complete and our goal statement is proved.
    \end{proof}

\section{Fast Projection on Fantopes}
\label{sec:fastproj}

In this section we define $\calS$ as follows
\[
    \mc{S} = \{(M, W) \in (\psd[\ell], \psd[m]): M \succcurlyeq 0,\ W \succcurlyeq 0,\ \Tr{M} + \Tr{W} = 1,\ \norm{W} \leq \Tr{W} / k\}.
\]
We will be concerned with solving the following optimization problem:
\begin{equation}
    \label{eq:pcomp-orig}
    (\mopt, \wopt) = \argmax_{(M, W) \in \mc{S}} \inp{F}{M} + \inp{G}{W} + \vNE (M) + \vNE (W).
\end{equation}

\begin{algorithm}[H]
\SetAlgoLined


\KwIn{Gain Matrices $F, G$}
\KwOut{$(M^*, W^*) = \argmax_{(M, W) \in \mathcal{S}} \inp{F}{M} + \inp{G}{W} + \vNE (M) + \vNE (W)$}

$Q \gets \exp (F)$\;
$H \gets \exp (G)$\;

$(v_i, \sigma_i)_{i = 1}^k \gets \mc{PCA}_k (H)$\;

$t \gets \Tr\lprp{\mc{P}_{V_k}^\perp H \mc{P}_{V_k}^\perp}$\;
Let $\tau^*$ be such that $\frac{\tau^*}{t + \sum_{i = 1}^k \min(\sigma_i, \tau^*)} = \frac{1}{k}$\;

$Z_1 = \Tr (Q),\ Z_2 = t + \sum_{i = 1}^k \min (\sigma_i, \tau^*)$\;
$\wh{M} = Z_1^{-1} Q,\ \wh{W} = Z_2^{-1} (\sum_{i = 1}^k \min(\tau^*, \sigma_i) v_iv_i^\top + \mc{P}_{V_k}^\perp H \mc{P}_{V_k}^\perp)$\;

$\gamma = \log Z_1,\ \zeta = \log Z_2 + k^{-1} \sum_{i = 1}^k (\log \sigma_i - \log \min (\sigma_i, \tau^*))$\;

$(M^*, W^*) = \lprp{\frac{e^\gamma}{e^\gamma + e^\zeta} \wh{M}, \frac{e^\zeta}{e^\gamma + e^\zeta} \wh{W}}$

\KwRet{Main output: $(M^*, W^*)$, Ancillary output: $\gamma,\zeta,\tau,(v_i,\sigma_i)_{i=1}^k$}
\caption{FastProjection} 
\label{alg:projI}
\end{algorithm}

In this subsection, we will prove that \pref{alg:projI} correctly computes the optimizer to \pref{eq:pcomp-orig}.  To do this, we will first analyze the following simpler problem:

\begin{equation}
    \label{eq:psimp}
    \wopts = \argmax_{\substack{W \succcurlyeq 0, \Tr{W} = 1 \\ \norm{W} \leq 1 / k}} \inp{G}{W} + \vNE (W).
\end{equation}

\begin{remark}  \label{rem:proj-def-p}
    Henceforth, we use $p(G)$ to denote $W^*$.
\end{remark}

We first prove that the optimizer, $\wopts$, of \pref{eq:psimp} has the same eigenvectors as that of $G$. 

\begin{lemma}
    \label{lem:psimpevec}
    Given $G \succcurlyeq 0$, the optimizer, $\wopts$, of \pref{eq:psimp} has the same eigenvectors as $G$.
\end{lemma}

\begin{proof}
    Let $\sigma_1 \geq \dots \geq \sigma_m \geq 0$ and $\lambda_1 \geq \dots \lambda_m \geq 0$ denote the eigenvalues of $\wopts$ and $G$ respectively. Now, we have the von Neumann's trace inequality:
    
    \begin{equation*}
        \inp{\wopts}{G} - \inp{\wopts}{\log \wopts} \leq \sum_{i = 1}^m (\sigma_i \lambda_i - \sigma_i \log \sigma_i)
    \end{equation*}
    with equality when the eigenvectors of $\wopts$ corresponding to the eigenvalue $\sigma_i$ coincide with the eigenvectors of $G$ for the eigenvalue $\lambda_i$. Therefore, the optimizer $\wopts$ must share the same set of eigenvectors as $G$.
\end{proof}

\begin{lemma}
    \label{lem:psimpeval}
    Given, $G \succcurlyeq 0$ and let $H = \exp G$ with eigenvalue decomposition $H = \sum_{i = 1}^m \lambda_i u_iu_i^\top$. Let $\nu^*$ be defined as follows:
    
    \begin{equation*}
        \nu^* = \min \lbrb{\nu > 0: \frac{\nu}{\sum_{i = 1}^m \min (\nu, \lambda_i)} \leq \frac{1}{k}}.
    \end{equation*}
    
    Then, the optimizer, $\wopts$, of \pref{eq:psimp} is given by:
    
    \begin{equation*}
        \wopts = \frac{\sum_{i = 1}^m \min (\nu^*, \lambda_i) u_iu_i^\top}{\sum_{i = 1}^m \min (\nu^*, \lambda_i)}.
    \end{equation*}
\end{lemma}

\begin{proof}
    From \pref{lem:psimpevec}, we know that the eigenvectors for $\wopts$ and $G$ and hence, $H$, coincide. Let $\sigma_1^*, \dots, \sigma_m^*$ denote the eigenvalues of $\wopts$ corresponding to the eigenvectors $u_1, \dots, u_m$. Then, we see from \pref{eq:psimp} that:
    
    \begin{gather*}
        (\sigma_1^*, \dots, \sigma_m^*) = \argmax_{(\sigma_1, \dots, \sigma_m) \geq 0} \sum_{i = 1}^m \sigma_i \log \lambda_i - \sigma_i \log \sigma_i \\
        \sum_{i = 1}^m \sigma_i = 1 \\
        \sigma_i \leq \frac{1}{k} \tag{Prog} \label{eq:psprog}
    \end{gather*}
    Since, the above optimization problem is convex, we compute its Lagrangian (Note that we must set $\alpha_i \geq 0$):
    \begin{equation*}
        \mc{L} (\{\sigma_i\}, \{\alpha_i\}, \beta) = \sum_{i = 1}^m \sigma_i (\log \lambda_i - \log \sigma_i + \beta - \alpha_i) - \beta + \sum_{i = 1}^m \alpha_i / k.
    \end{equation*}
    Now, picking $\beta' < - \max_{i \in [m]} \abs{\log \lambda_i}$. Note that $\log \lambda_i$ are the eigenvalues of $G$ and hence $\beta$ is finite. We now have:
    \begin{equation*}
        \max_{\{\sigma_i\} \geq 0} \mc{L} (\{\sigma_i\}, 0, \beta') = \max_{\{\sigma_i\} \geq 0} \sum_{i = 1}^m \sigma_i (\log \lambda_i - \log \sigma_i + \beta') - \beta' \leq \frac{m}{e} - \beta'
    \end{equation*}
    by noting that $-x\log x$ is maximized at $x = 1 / e$. Note that the above conclusion holds true for $\alpha$ satisfying $\max_{i} \alpha_i \leq - \max_{i \in [m]} \abs{\log \lambda_i} - \beta'$. Therefore, Slaters' condition holds for both the primal problem, \ref{eq:psprog}, and its dual. Furthermore, the optimal value of \ref{eq:psprog} is bounded as both \ref{eq:psprog} and its dual have a feasible point with finite objective value. Therefore, strong duality holds for \ref{eq:psprog} and its dual and their optimal value is attained. Let $\{\sigma_i^*\}$ and $(\{\alpha_i^* \geq 0\}, \beta^*)$ denote the primal and dual optimal points respectively. Note that by a simple exchange argument $\sigma^*_i \neq 0$. Therefore, the KKT conditions apply to \ref{eq:psprog} and we get:
    \begin{equation*}
        \log \lambda_i - \log \sigma_i^* + \beta^* - \alpha_i^* = 0 \implies \sigma_i^* = e^{\beta^* - \alpha^*_i} \lambda_i.
    \end{equation*}
    From the condition of primal feasibility, we get that $e^{\beta^*} = (\sum_{i = 1}^m e^{-\alpha^*_i} \lambda_i)^{-1}$. Also, note that we get from complementary slackness that $\alpha^*_i > 0$ implies that $\sigma^*_i = 1 / k$. Additionally, from complementary slackness, we obtain that $\sigma^*_i \geq \sigma^*_j$ for $i \geq j$. 
    
    Let $l = \#\{i: \alpha^*_i > 0\}$. We first tackle the case where $l = 0$. In this case, the optimizer is simply $\wopts = H / \Tr (H)$ and the statement of the lemma is true. 
    
    Now assume that $l > 0$. Let us now consider the function, $f$, defined as:
    \begin{equation*}
        f(\nu) = \frac{\nu}{\sum_{i = 1}^m \min (\lambda_i, \nu)}.
    \end{equation*}
    When $\lambda_m > 0$ which holds in this case, $f(\nu)$ is a strictly increasing, continuous function of $\nu$ in the interval $[\lambda_m, \infty)$ and its value increases from $1 / m$ to $\infty$. For $i, j \in [l]$, we have $\sigma^*_i = \sigma^*_j = 1/k$ by complementary slackness and therefore $e^{-\alpha^*_i} \lambda_i = e^{-\alpha^*_j} \lambda_j = \wh{\nu}$. For $i \notin [l]$, we have $\sigma^*_i = e^{\beta^*} \min(\lambda_i, \wh{\nu})$ as we have $\sigma^*_i = e^{\beta^*} \lambda_i \leq e^{\beta^*} \wh{\nu} = 1/k$. From the previous two statements, we have $\sigma^*_i = e^{\beta^*} \min(\wh{\nu}, \lambda_i)$ for all $i \in [m]$. Finally, we have $f(\wh{\nu}) = 1 / k$ from complementary slackness which implies that $\wh{\nu} = \nu^*$ as $f$ is strictly increasing and continuous. Which implies that the optimal value of $\sigma^*_i$ is given by $\sigma^*_i = \min(\lambda_i, \nu^*) / (\sum_{j = 1}^m \min (\lambda_j, \nu^*))$, thus proving the lemma.
    
\end{proof}

Finally, we will now show how to use solutions to \pref{eq:psimp} to obtain solutions to the following:

\begin{equation}
    \label{eq:pcomp}
    (\mopt, \wopt) = \argmax_{(M, W) \in \mc{S}} \inp{F}{M} + \inp{G}{W} + \vNE (M) + \vNE (W).
\end{equation}
The result is detailed in the following lemma:
\begin{lemma}
    \label{lem:pcomp}
    Let $F, G \succcurlyeq 0$ and let $Q = \exp (F)$, $Z_1 = \Tr (\exp (F))$, $H = \exp (G)$ with eigenvalue decomposition $H = \sum_{i = 1}^m \lambda_i u_iu_i^\top$ and $\nu^*$ and $Z_2$ be defined as:
    
    \begin{equation*}
        \nu^* = \max \lbrb{\nu > 0: \frac{\nu}{\sum_{i = 1}^m \min(\nu, \lambda_i)} \leq \frac{1}{k}}, \qquad Z_2 = \sum_{i = 1}^m \min(\lambda_i, \nu^*).
    \end{equation*}
    Then, the optimizers, $(\mopt, \wopt)$, of Equation~\pref{eq:pcomp} are given by:
    \begin{equation*}
        \mopt = \frac{e^{\gamma}}{e^{\gamma} + e^{\zeta}} \cdot \wh{M},~ \wopt = \frac{e^{\zeta}}{e^{\gamma} + e^{\zeta}} \cdot \wh{W}
    \end{equation*}
    where $\gamma = \log Z_1$, $\zeta = \log (Z_2) + k^{-1} \sum_{i = 1}^k (\log (\lambda_i) - \log \min (\lambda_i, \nu^*))$ and $\wh{W}$  and $\wh{M}$ are defined as:
    
    \begin{equation*}
        \wh{M} = \frac{Q}{\Tr (Q)} \text{ and } \wh{W} = \frac{\sum_{i = 1}^m \min(\lambda_i, \nu^*) u_i u_i^\top}{\sum_{i = 1}^m \min(\lambda_i, \nu^*)}.
    \end{equation*}
\end{lemma}

\begin{proof}
    Let $(\mopt, \wopt)$ denote the solutions of \pref{eq:pcomp} and let $\alpha = \Tr{\mopt}$. Then, we must have:
    
    \begin{equation*}
        \mopt = \argmax_{\substack{M \succcurlyeq 0 \\ \Tr (M) = \alpha}} \inp{F}{M} - \inp{M}{\log M} \text{ and } \wopt = \argmax_{\substack{W \succcurlyeq 0 \\ \Tr (W) = 1 - \alpha \\ \norm{W} \leq \Tr W / k}} \inp{G}{W} - \inp{W}{\log W}.
    \end{equation*}
    Now consider the case where $\alpha > 0$. For the first equation, we have:
    \begin{align*}
        \mopt &= \argmax_{\substack{M \succcurlyeq 0 \\ \Tr (M) = \alpha}} \inp{F}{M} - \inp{M}{\log M} \\
        &= \argmax_{\substack{M \succcurlyeq 0 \\ \Tr (M) = \alpha}} \inp{F}{M / \alpha} - \inp{M / \alpha}{\log M} \\
        &= \argmax_{\substack{M \succcurlyeq 0 \\ \Tr (M) = \alpha}} \inp{F}{M / \alpha} - \inp{M / \alpha}{\log (M/\alpha)} + \log 1 / \alpha\\
        &= \argmax_{\substack{M \succcurlyeq 0 \\ \Tr (M) = \alpha}} \inp{F}{M / \alpha} - \inp{M / \alpha}{\log (M/\alpha)} \\
        &= \alpha \argmax_{\substack{M \succcurlyeq 0 \\ \Tr (M) = 1}} \inp{F}{M} - \inp{M}{\log M} = \alpha \wh{M}.
    \end{align*}
    When $\alpha = 0$, the conclusion of the previous manipulation is trivially true. By a similar manipulation, from \pref{lem:psimpeval} we have $\wopt = (1 - \alpha) \wh{W}$. We now have:
    \begin{align*}
        \inp{\alpha \wh{M}}{F} - \inp{\alpha \wh{M}}{\log \alpha \wh{M}} = \alpha Z_1^{-1} (\inp{Q}{F} - \inp{Q}{F}) + \alpha \log Z_1 - \alpha \log \alpha = \alpha \log Z_1 - \alpha \log \alpha.
    \end{align*}
    
    We now proceed for a similar computation for $\wh{W}$:
    \begin{align*}
        &\inp{(1 - \alpha) \wh{W}}{G} - \inp{(1 - \alpha) \wh{W}}{\log (1 - \alpha) \wh{W}} \\
        &= (1 - \alpha) Z_2^{-1} \lprp{\sum_{i = 1}^m \min(\lambda_i, \nu^*) \log \lambda_i - \min (\lambda_i, \nu^*) \log (\min(\lambda_i, \nu^*))} \\
        &\qquad + (1 - \alpha) \log Z_2 - (1 - \alpha) \log (1 - \alpha) \\
        &= (1 - \alpha) Z_2^{-1} \lprp{\sum_{i = 1}^k \min(\lambda_i, \nu^*) \log \lambda_i - \min (\lambda_i, \nu^*) \log (\min(\lambda_i, \nu^*))} \\
        &\qquad + (1 - \alpha) \log Z_2 - (1 - \alpha) \log (1 - \alpha) \\
        &= (1 - \alpha) \frac{1}{k} \lprp{\sum_{i = 1}^k \log \lambda_i - \log (\min(\lambda_i, \nu^*))} + (1 - \alpha) \log Z_2 - (1 - \alpha) \log (1 - \alpha)
    \end{align*}
    where the second-to-last equality follows because at most $k$ of the $\lambda_i$ are greater than $\nu^*$ and the final inequality follows from the fact that $\lambda_i \geq \nu^*$ implies that $\min(\lambda_i, \nu^*) = \nu^* / Z_2 = 1 / k$. By putting the previous two results together, we get that:
    \begin{equation*}
        \alpha = \argmax_{\beta \in [0, 1]} \beta \gamma  + (1 - \beta) \zeta - \beta \log \beta - (1 - \beta) \log (1 - \beta)
    \end{equation*}
    whose optimal value is given by $\alpha = e^{\gamma} / (e^{\gamma} + e^{\zeta})$ which concludes the proof the lemma.
\end{proof}

\subsection{Fast Approximate Projection}
\label{ssec:appxproj}
It is unclear how to \emph{exactly} solve the optimization problem \pref{eq:pcomp-orig} fast, so we give an algorithm that outputs a solution close to the exact optimizer in trace norm.  As a first step, we give an algorithm to approximately solve the optimization problem \pref{eq:psimp}.  In the algorithm below, not all matrices are explicitly computed and we obtain an implicit representation of $W^*$ rather than an explicit $m\times m$ matrix.  For simplicity of exposition, we defer the details of this implicit representation to later subsections.  In the algorithm below, when we say $\Tr_{\eps}(H)$, we mean running the trace estimation algorithm from \pref{cor:esttr} from $H$, which with high probability produces a $(1\pm\eps)$-approximation of the trace.

\begin{algorithm}[H]
\SetAlgoLined


\KwIn{Gain Matrix $G$}
\KwOut{$W^*\approx_{\eps}\argmax_{W\psdge 0,\Tr(W)=1,\|W\|\le1/k} \inp{G}{W} \vNE (W)$}

$\wh{W} \gets \exp (G)$\;

$(v_i, \sigma_i)_{i = 1}^k \gets \text{PCA}_k (\wh{W})$\;

$H \gets (1-2\eps)\calP_{V_k}^{\perp}\wh{W}\calP_{V_k}^{\perp}$\;

$T \gets \Tr_{\eps}\lprp{H}$\;
$\wt{\tau} \gets \text{ solution to }kt = (1-\eps)T+\sum_{i=1}^k\min\{\sigma_i,t\}~~~t\in[\sigma_k,\infty)$\;

$W^* = \frac{1-4k\eps}{k\wt{\tau}}\left(\sum_{i=1}^k\min\{\sigma_i,\wt{\tau}\}v_iv_i^{\top} + H\right)$

\KwRet{$\text{Main output: $W^*$, Ancillary output: $\wt{\tau},(v_i,\sigma_i)_{i=1}^k$}$}
\caption{SimpleApproximateProjection}
\label{alg:projApp}
\end{algorithm}

Our first goal is to show that the output of \pref{alg:projApp} on input $G$ is close in trace norm to $p(G)$ where $p$ is as defined in \pref{rem:proj-def-p}.  Concretely, we prove:
\begin{theorem}[Simplified version of \pref{thm:main-thm-simp-proj}]
    Let $G$ be a positive semidefinite matrix, and let $W^*$ be the output of \pref{alg:projApp} on input $G$.  Then:
    \[
        \|p(G)-W^*\|_* \le 4\sqrt{k\eps} + 9k\eps.
    \]
\end{theorem}
The full statement of the above, which states some more technical properties of $W^*$ can be found in \pref{thm:main-thm-simp-proj}.

\subsection{Closeness in trace norm I: projections of spectrally similar matrices}    \label{sec:tn-closeness-1}
In this section, let $A$ and $\wt{A}$ be positive semidefinite matrices such that
\[
    (1-\eps)\wt{A}\psdle A\psdle (1+\eps)\wt{A}.
\]
for some $0<\eps<1/2$.\footnote{$\wt{A}$ will be a matrix obtained via the power iteration based PCA algorithm}
\begin{lemma}   \label{lem:logs-close}
    $\|\log \wt{A} - \log A\| \le 4\eps$.
\end{lemma}
\begin{proof}
From \pref{fact:log-op-monotone}, we can conclude that
\[
    \log\wt{A} + \log(1-\eps)\cdot\Id \psdle \log A \psdle \log\wt{A} + \log(1+\eps)\cdot\Id
\]
which means
\[
    \|\log\wt{A}-\log A\| \le |\log(1+\eps)| + |\log(1-\eps)| \le 4\eps.
\]
where the last inequality follows from $\eps < 1/2$.
\end{proof}
\noindent For the rest of this section, let $\displaystyle M\coloneqq\arg\min_{\substack{X\psdge 0\\ \Tr(X)=1\\ \|X\|\le1/k}}\QRL(X,A)$ and let $\displaystyle \wt{M}\coloneqq\arg\min_{\substack{X\psdge 0\\ \Tr(X)=1\\ \|X\|\le1/k}}\QRE(X,\wt{A})$.  In the language of \pref{rem:proj-def-p}, $M = p(\log A)$.
\begin{lemma}   \label{lem:QRE-close}
    $\QRE(\wt{M},A) \le \QRE(M,A)+8\eps$.
\end{lemma}
\begin{proof}
    We prove our claim with the following chain of inequalities:
    \begin{align*}
        \QRE(\wt{M},A) &= \QRE(\wt{M},\wt{A}) + \langle \wt{M}, \log \wt{A} - \log A\rangle \\
        &\le \QRE(\wt{M},\wt{A})+4\eps &\text{(from \pref{lem:logs-close})} \\
        &\le \QRE(M,\wt{A})+4\eps \\
        &= \QRE(M,A)+\langle M, \log A-\log \wt{A}\rangle+4\eps \\
        &\le \QRE(M,A)+8\eps &\text{(from \pref{lem:logs-close}).}
    \end{align*}
\end{proof}
\noindent Finally we prove:
\begin{lemma}   \label{lem:closeness-trace}
    $\|M-\wt{M}\|_* \le 4\sqrt{\eps}$.
\end{lemma}
\begin{proof}
    Define $f(X)\coloneqq \QRE(X,A)$.  From \pref{fact:vndiv-conc} $f$ is $1$-strongly convex and
    \[
        f(\wt{M}) \ge f(M) + \langle\grad f(M),\wt{M}-M\rangle+\frac{1}{2}\|\wt{M}-M\|_*^2.
    \]
    Since $M$ is the minimizer of $f$ in $\{X\in\psd[m]:\Tr(X)=1,\|X\|\le1/k\}$ and $f$ is $1$-strongly convex, $\langle\grad f(M),\wt{M}-M\rangle\ge 0$ and this implies
    \[
        f(\wt{M})-f(M)\ge\frac{1}{2}\|\wt{M}-M\|_*^2.
    \]
    From \pref{lem:QRE-close}, $f(\wt{M})-f(M) \le 8\eps$ and hence
    \[
        \frac{1}{2}\|\wt{M}-M\|_*^2 \le 16\eps
    \]
    and consequently
    \[
        \|M-\wt{M}\|_* \le 4\sqrt{\eps}.
    \]
\end{proof}

\subsection{Closeness in trace norm II: robustness to trace}    \label{sec:tn-closeness-2}
In this subsection, let $M$ be a $m\times m$ positive definite matrix with eigenvalues $\lambda_1\ge\lambda_2\ge\cdots\ge\lambda_{m} > 0$ and corresponding eigenvectors $v_1\dots,v_{m}$.  Let $k$ be an integer less than $m$, and let $T$ denote $\sum_{i=k+1}^m\lambda_i$.  We wish to show that all pairs in a certain set of matrices are close in trace norm.  Before we describe these matrices, we will need the following technical statement.
\begin{proposition} \label{prop:threshold-lipschitz}
    Let $f_1(t) = kt$, let $f_2(t) = \sum_{i=1}^m \min\{t,\lambda_i\}$.  $f_1(t)=f_2(t)+\Delta$ has a unique solution $\tau_{\Delta}$ on $[\lambda_k,\infty)$ for any $\Delta\in[-\eps T,\eps T]$.  Further, $|\tau_0-\tau_{\Delta}|\le|\Delta|$.
\end{proposition}
\begin{proof}
    Define functions $\{g_i\}_{i=0}^{k-1}$ defined on $[\lambda_k,\infty)$ where $g_i(t)=it+\sum_{j=i+1}^m\lambda_j$.  Observe that $f_2(t)$ is equal to $\min_{i\in[0,k-1]}g_i(t)$ on $[\lambda_k,\infty)$.  Thus its right-hand side derivatives must be bounded by $k-1$.  Since (i) $f_2(\lambda_k)-\Delta >f_1(\lambda_k)$, (ii) the right-hand derivative of $f_1$ is $k$ everywhere, and (iii) the right-hand derivative of $f_2$ at any point in $[\lambda_k,\infty)$ is at most $k-1$, there must be a unique $\tau_{\Delta}$ such that $f_1(\tau_{\Delta})=f_2(\tau_{\Delta})$.  The right-hand derivative of $f_1-f_2$ is at least $1$ on $[\lambda_k,\infty)$ and thus $\tau_{\Delta}$ must be contained in $[\tau_0-|\Delta|,\tau_0+|\Delta|]$ and thus $|\tau_0-\tau_{\Delta}|<|\Delta|$.
\end{proof}
\begin{definition}  \label{def:noisy-trunc}
We now define the \emph{noisy truncation operator}:
\[
    \Xi(M,C) \coloneqq \frac{1}{k\tau_C}\sum_{i=1}^m\min\{\tau_C,\lambda_i\}v_iv_i^{\top}
\]
where $C$ must be in range $[-\eps T, \eps T]$ and $\tau_C$ is as defined in the statement of \pref{prop:threshold-lipschitz}.
\end{definition}

\begin{lemma}   \label{lem:closeness-trace-norm}
    For every $C\in[-\eps T, \eps T]$, $\displaystyle\|\Xi(M,0)-\Xi(M,C)\|_* \le 2k\eps$.
\end{lemma}
\begin{proof}
Note that
\[
    \|\Xi(M,0)-\Xi(M,C)\|_* = \sum_{i=1}^m \left|\frac{\min\{\tau_0,\lambda_i\}}{k\tau_0} - \frac{\min\{\tau_C,\lambda_i\}}{k\tau_C}\right|
\]
First observe that $|\tau_0-\tau_C|\le |C| \le \eps T$.  Since $\tau_0\ge\lambda_k$, we have $f_1(\tau_0)=f_2(\tau_0)\ge T$, which implies $|\tau_0-\tau_C|\le \eps f_1(\tau_0) = \eps k\tau_0$.  This means $\tau_C=\gamma\tau_0$ for some $\gamma\in1\pm k\eps$.  As a result
\begin{align*}
    \|\Xi(M,0)-\Xi(M,C)\|_* &= \frac{1}{k\tau_0}\sum_{i=1}^m \left|\min\{\tau_0,\lambda_i\} - \min\left\{\tau_0,\frac{\lambda_i}{\gamma}\right\}\right|\\
    &\le \frac{1}{k\tau_0}\left(1-\frac{1}{\gamma}\right)\sum_{i=1}^m\min\{\tau_0,\lambda_i\}\\
    &= \left(1-\frac{1}{\gamma}\right)\frac{1}{f_1(\tau_0)}\cdot f_2(\tau_0)\\
    &= \left(1-\frac{1}{\gamma}\right)\\
    &\le 2k\eps.
\end{align*}
\end{proof}

\begin{remark}  \label{rem:approx-trace}
    We observe that $\Tr(\Xi(M,C)) = \frac{f_2(\tau_C)}{f_1(\tau_C)}$.  Since $f_1-f_2$ is increasing on $[\lambda_k,\infty)$, it follows that $\tau_C\le\tau_0$ when $C\le 0$, and consequently $f_1(\tau_C)\le f_2(\tau_C)$, which means $\Tr(\Xi(M,C))\ge 1$.
\end{remark}

\begin{remark}  \label{rem:spec-norm}
    $\|\Xi(M,C)\|$ is always at most $\frac{1}{k}$ by construction.
\end{remark}

\begin{remark}  \label{rem:noisy-trunc-is-projection}
    $\Xi(M,0) = p(\log M)$ where $p$ is the function from \pref{rem:proj-def-p}.
\end{remark}

\subsection{Closeness in trace norm III: wrap-up}   \label{sec:tn-closeness-3}
We will use the results of \pref{sec:tn-closeness-1}, \pref{sec:tn-closeness-2} and \pref{sec:pow} to prove guarantees of the output of \pref{alg:projApp}.  Given an input matrix $G$, we perform a sequence of transformations described below to get a matrix $\wt{p}(G)$.  Our goal is to prove that $\wt{p}(G)$ is close to $p(G)$, where $p$ is as defined in \pref{rem:proj-def-p}.
\begin{enumerate}
    \item Let $A_0$ be a matrix such that $(1-\eps)A_0\psdle \exp(G)\psdle (1+\eps)A_0$.
    \item We perform a $k$-PCA on $A_0$ and obtain vectors $v_1,\dots,v_k$ as output along with numbers $\wt{\lambda}_1,\dots,\wt{\lambda}_k$ where $\wt{\lambda}_i = v_i^{\top}A_0v_i$.
    \item Define $H$ as
    \[
        H \coloneqq (1-2\eps)\left(\Id-\sum_{i=1}^kv_iv_i^{\top}\right)A_0\left(\Id-\sum_{i=1}^kv_iv_i^{\top}\right)
    \]
    \item We define $A_1$ as
    \[
        A_1 \coloneqq \sum_{i=1}^k\wt{\lambda_i} v_iv_i^{\top} +  H.
    \]
    \item We run a $(1\pm\eps)$-approximate trace estimation algorithm on $H$ and obtain number $\wt{T}$.
    \item We solve for $t$ in the following equation and call the solution $\wt{\tau}$.
    \begin{align*}
        kt = (1-\eps)\wt{T} + \sum_{i=1}^k \min\{\wt{\lambda}_i,t\} & & t\in[\wt{\lambda}_k,\infty).
    \end{align*}
    \item We define $A_2$ as
    \[
        A_2 \coloneqq \sum_{i=1}^k\min\{\wt{\lambda}_i,\wt{\tau}\}v_iv_i^{\top} + H.
    \]
    \item Finally, we define $\wt{p}(W)$ as
    \[
        \wt{p}(W) \coloneqq \frac{(1-4k\eps)}{k\wt{\tau}}A_2.
    \]
\end{enumerate}
By a combination of \pref{thm:main-svd-guarantee} and the fact that
\[
    \frac{H}{1-2\eps}\psdle A_0,
\]
we know
\[
    (1-2\eps)^{k+2}A_1 \psdle \exp(G) \psdle (1+2\eps)^{k+2}A_1
\]
except with probability at most $O(k\delta)$.  Via \pref{lem:closeness-trace}, a consequence of the above is that for $\eps<\frac{1}{k^2}$:
\[
    \|p(G)-p(\log A_1)\|_* \le 4\sqrt{k\eps} \numberthis \label{eq:bound-from-psd}
\]
Now, we analyze closeness of $p(\log A_1)$ and $\wt{p}(G)$.  Let $T  = \Tr(H)$.  Then $(1-\eps)\wt{T} = T+C$ for some $C$ in the range $[-2\eps T, 0]$.  We now recall the noisy truncation operator $\Xi$ from \pref{def:noisy-trunc}.  By \pref{rem:eig-lower-bound} $\wt{\lambda}_1,\dots,\wt{\lambda}_k$ are the top $k$ eigenvalues of $A_1$ and thus $\wt{p}(G)$ is equal to the matrix $(1-4k\eps)\Xi(A_1,C)$.  First, from \pref{lem:closeness-trace-norm}:
\begin{align}
    \|\Xi(A_1,C)-\Xi(A_1,0)\|_*\le 4k\eps   \label{eq:close-trace-trunc}
\end{align}
Next, by \pref{rem:approx-trace}, $\Tr(\Xi(A_1,0))=1$ and thus by triangle inequality
\begin{align}
    \Tr(\Xi(A_1,C))\le 1+4k\eps \label{eq:small-trace}
\end{align}
Finally, we have
\begin{align*}
    \|p(\log A_1)-\wt{p}(G)\|_* &\le \|\Xi(A_1,C)-\Xi(A_1,0)\|_* + 4k\eps\|\Xi(A_1,C)\|_*\\
    &\le 9k\eps & \text{by \pref{eq:close-trace-trunc} and \pref{eq:small-trace}.}
\end{align*}
Combining the above with \pref{eq:bound-from-psd} via triangle inequality gives us:
\[
    \|p(G)-\wt{p}(G)\|_* \le 4\sqrt{k\eps} + 9k\eps.
\]
Finally, note that by \pref{rem:approx-trace}, $\Tr(\Xi(A_1,C))\ge 1$ and by \pref{rem:spec-norm},
\[
    \|\Xi(A_1,C)\|\le\frac{1}{k}\le\frac{\Tr(\Xi(A_1,C))}{k}.
\]
Multiplying the above inequality by $(1-4k\eps)$ lets us conclude that:
\[
    \|\wt{p}(G)\|\le\frac{\Tr(\wt{p}(G))}{k}
\]
and multiplying \pref{eq:small-trace} with $(1-4k\eps)$ lets us conclude
\[
    \Tr(\wt{p}(G)) \le 1.
\]
Thus, we have the following theorem about \pref{alg:projApp}.
\begin{theorem} \label{thm:main-thm-simp-proj}
    \pref{alg:projApp} takes in $G$ as input, and outputs a matrix $\wt{p}(W)$ such that except with probability $O(k\delta)$ the following three conditions hold:
    \begin{enumerate}
        \item $\|p(G)-\wt{p}(G)\|_* \le 4\sqrt{k\eps} + 9k\eps$.
        \item $\|\wt{p}(G)\|\le\frac{\Tr(\wt{p}(G))}{k}$.
        \item $\Tr(\wt{p}(G)) \le 1.$
    \end{enumerate}
\end{theorem}

\subsection{Full Approximate Projection}
In this section, we describe a fast algorithm to produce an approximate solution to the optimization problem \pref{eq:pcomp-orig}.  In particular, given $F,G\psdge 0$ let:
\[
    (M^*,W^*) = \argmax_{(M,W)\in\calS}\langle F, M\rangle + \langle G, W\rangle - \langle W, \log W\rangle - \langle M, \log M\rangle.
\]
\begin{remark}  \label{rem:def-q-full-proj}
    We say $q_1(F) = M^*$ and $q_2(G)=W^*$ and we use $(\wt{q}_1(F),\wt{q}_2(W))$ to refer to the output of \pref{alg:projAppComp}.
\end{remark}
Our goal is to bound the trace norm distance between $q_1(F)$ and $\wt{q}_1(F)$, and between $q_2(F)$ and $\wt{q}_2(F)$.

\begin{algorithm}[H]
\SetAlgoLined


\KwIn{Gain Matrices $F, G$}
\KwOut{$(M^*, W^*) \approx_{\eps} \argmax_{(M, W) \in \mathcal{S}} \inp{F}{M} + \inp{G}{W} + \vNE (M) + \vNE (W)$}

$\wh{W},\wt{\tau}, (v_i,\wt{\sigma}_i)_{i=1}^k \gets \text{SimpleApproximateProjection}(G)$\;

$Q \gets \exp (F)$\;

$\wt{Z}_1 = \Tr_{\eps}(Q),\ \wt{Z}_2 = k\wt{\tau}$\;
$\wh{M} = \wt{Z}_1^{-1} Q$\;

$\wt{\gamma} = \log \wt{Z}_1,\ \wt{\zeta} = \log \wt{Z}_2 + k^{-1} \sum_{i = 1}^k (\log \sigma_i - \log \min (\sigma_i, \wt{\tau}))$\;

$(M^*, W^*) = \left(\frac{e^{\wt{\gamma}}}{e^{\wt{\gamma}} + e^{\wt{\zeta}}} \wh{M}, \frac{e^{\wt{\zeta}}}{e^{\wt{\gamma}} + e^{\wt{\zeta}}} \wh{W}\right)$\;

\KwRet{Main output: $(M^*, W^*)$, Ancillary output: $\wt{\gamma},\wt{\zeta},\wt{\tau},(v_i,\sigma_i)_{i=1}^k$}
\caption{FullApproximateProjection}
\label{alg:projAppComp}
\end{algorithm}

We now prove that the output $(\wt{q}_1(F),\wt{q}_2(G))$ of \pref{alg:projAppComp} on input $F$ and $G$ is close to $(q_1(F),q_2(G))$ in trace norm.
\begin{theorem} \label{thm:approx-guarantees-comp-proj}
    We have the following guarantees:
    \begin{enumerate}
    \item $\|\wt{q}_1(F)-q_1(F)\|_* \le O(k\eps)$.
    \item $\|\wt{q}_2(G)-q_2(G)\|_* \le O(\sqrt{k\eps})$.
    \end{enumerate}
\end{theorem}
\begin{proof}
    \pref{alg:projI} computes $q_1$ and $q_2$ exactly.  We note that all trace estimates in \pref{alg:projAppComp} are up to a multiplicative $(1\pm\eps)$ factor.  All eigenvalue computations are also correct up to a multiplicative $(1\pm4k\eps)$.  As a consequence of the approximation guarantees on trace and eigenvalues, and the proof of \pref{lem:closeness-trace-norm}, $\wt{\tau}$ as computed in \pref{alg:projAppComp} is within a multiplicative $1\pm O(k\eps)$ factor of $\tau^*$ from \pref{alg:projI}.  Hence, $\wt{\zeta}$ and $\wt{\gamma}$ from the output of \pref{lem:closeness-trace-norm} must be within a multiplicative $1\pm O(k\eps)$ of $\zeta$ and $\gamma$ from the output of \pref{alg:projAppComp}.
    
    As a consequence, $\|\wt{q}_1(F)-q_1(F)\|_*\le O(k\eps)$.  The inequality $\|\wt{q}_2(G)-q_2(G)\|_*\le O(\sqrt{k\eps})$ follows from the above discussion combined with \pref{thm:main-thm-simp-proj}.
\end{proof}

\subsection{Implementation}
Now we describe the representation of the input and output, and give a runtime guarantee on \pref{alg:projAppComp}.  We are given $\ell\times\ell$ matrix $F$ and $m\times m$ matrix $G$ via the following oracles:
\begin{enumerate}
    \item An oracle that takes in $\ell$-dimensional vectors $v$ and outputs $Fv$ in time $t_F$.  Note that by \pref{lem:expapx} we can also implement an algorithm to compute $A_Fv$  in time $O(t_F\lambda_{\max}(F)\log(2\eps^{-1}))$ where $A_F$ is some matrix satisfying:
    \[
        (1-\eps)\exp(F) \psdle A_F \psdle (1+\eps)\exp(F).
    \]
    \item An oracle that takes in $m$-dimensional vectors $v$ and outputs $Gv$ in time $t_G$.  Note that by \pref{lem:expapx} we can also implement an algorithm to compute $A_Gv$  in time $O(t_G\lambda_{\max}(G)\log(2\eps^{-1}))$ where $A_G$ is some matrix satisfying:
    \[
        (1-\eps)\exp(G) \psdle A_G \psdle (1+\eps)\exp(G).
    \]
\end{enumerate}
\begin{observation} \label{obs:only-ancillary-needed}
    Given the oracle corresponding to input $G$ and the ancillary output of \pref{alg:projApp}, it is possible to implement an oracle that takes in $m$-dimensional vectors $v$ as queries and outputs $\wt{p}(G)v$ in $t_G+O(km)$ time.
\end{observation}
In light of \pref{obs:only-ancillary-needed}, we only need to analyze the runtime of producing the ancillary output; thus the runtime of \pref{alg:projApp} is 
\begin{displayquote}
    Runtime of the PCA algorithm + Runtime of the trace estimation algorithm + Runtime of computing $\wt{\tau}$. 
\end{displayquote}
The runtime of the PCA subroutine is $O(t_G(\log m + \log 1/\delta + \log 1/\eps)/\eps)$, the runtime of the trace estimation algorithm (from \pref{cor:esttr}) is $O\left((\poly(k)t_G + m)\log\left(\frac{1}{\eps}\right)\cdot\frac{\log m + \log(1/\delta)}{\eps^2}\right)$, and finally by using the characterization of $\wt{\tau}$ from the proof of \pref{prop:threshold-lipschitz}, $\wt{\tau}$ can be computed in $\poly(k)$ time.  Thus, we get that the runtime of \pref{alg:projApp} is:
\[
    O\left((t_G + m) \cdot\poly\left(k,\log m,\log\left(\frac{1}{\delta}\right),\frac{1}{\eps}\right)\right).
\]
Directly analogous to \pref{obs:only-ancillary-needed} is the following observation:
\begin{observation} \label{obs:only-ancillary-needed-full}
    Given the oracles corresponding to inputs $F,G$ and the ancillary output of \pref{alg:projAppComp}, it is possible to implement the following oracles:
    \begin{enumerate}
        \item An oracle that takes in $\ell$-dimensional vectors $v$ as queries and outputs $\wt{q}_1(F)v$ in $O(t_F)$ time.
        \item An oracle that takes in $m$-dimensional vectors $v$ as queries and outputs $\wt{q}_2(G)v$ in $O(t_G+km)$ time.
    \end{enumerate}
\end{observation}
From \pref{obs:only-ancillary-needed-full}, given that we only need to compute ancillary output, the runtime of \pref{alg:projAppComp} is:
\begin{displayquote}
    Runtime of \pref{alg:projApp} + Runtime of trace estimation + Runtime of computing $\wt{\gamma}$ and $\wt{\zeta}$.
\end{displayquote}
The runtime of trace estimation in this case is:
\[
    O\left(\left(\poly(k)t_F\log\left(\frac{1}{\eps}\right) + \ell\right)\cdot\frac{\log l + \log(1/\delta)}{\eps^2}\right)
\]
Since the third component is no more than the first or second, we have an overall runtime of:
\[
    T(l,m,k,\eps,\delta,t_G) \coloneqq O\left((t_G+t_F + \ell + m)\cdot\poly\left(k,\log (\ell + m),\log\left(\frac{1}{\delta}\right),\frac{1}{\eps}\right)\right).
\]
In summary, from the above discussion and a combination of \pref{thm:main-thm-simp-proj} we have proved:
\begin{theorem} \label{thm:fast-projection-main}
    There is an algorithm $\projalg$ which takes in matrices $F$ and $G$ of dimension $\ell\times\ell$ and $m\times m$ respectively, error parameter $\eps$, confidence parameter $\delta$, and outputs matrices $\wt{q}_1(F)$ and $\wt{q}_2(G)$ in time $T(\ell,m,k,\eps,\delta,t_G)$ such that except with probability $\delta$:
    \begin{enumerate}
        \item $\|\wt{q}_1(F)-q_1(F)\|_*\le \eps/2$.
        \item $\|\wt{q}_2(G)-q_2(G)\|_*\le \eps/2$.
        \item $\|\wt{q}_2(G)\|\le\frac{\Tr(\wt{q}_2(G))}{k}$.
    \end{enumerate}
    Further, $\wt{q}_1(F)$ is of the form $\gamma \exp(F)$, and hence the ``implicit representation'' the algorithm outputs is the scalar $\alpha$.  Similarly, $\wt{q}_2(F)$ is of the form $\beta\left(\sum_{i=1}^k\min\{\sigma_i,\tau\}v_iv_i^{\top}\right)+\beta'(\Id-\Pi_{V_k})\exp(G)(\Id-\Pi_{V_k})$ and hence the ``implicit representation'' the algorithm outputs is given by the scalars $\beta,\beta',\tau$ along with pairs $(\sigma_i,v_i)_{i=1}^k$.
\end{theorem}
\newcommand{\PPart}{\textsc{PlantedPartition}}
\section{Inference in semirandom graph models}
\label{sec:infgraph}

The technical content in this section follows the proof of Corollary 9.3 of \cite{CSV17}.
\paragraph{Problem setup.}  Let $V$ be a set of $n$ vertices, and let $S\subseteq V$ be a subset of size $\alpha n$.  A directed graph $G$ on vertex set $V$ is generated according to the following model:
\begin{enumerate}
    \item For every pair $(u,v)$ (possibly with $u=v$) such that $u\in S$ and $v\in S$, the directed edge $(u,v)$ is added to the edge set with probability $\frac{a}{n}$.
    \item For every pair $u\in S, v\notin S$, the directed edge $(u,v)$ is added to the edge set with probability $\frac{b}{n}$.
    \item For each remaining pair $(u,v)$, an adversary decides whether to make $(u,v)$ an edge or not.
\end{enumerate}
\begin{definition}
    In the \PPart~problem, we are given a graph $G$ generated according to the above model as input, and the goal is to produce a list of sets of vertices $\wt{S}_1,\wt{S}_2,\dots,\wt{S}_k$ where $k=O(1/\alpha)$ and there exists $i$ such that $|\wt{S}_i\Delta S| < O\left(\frac{\max\{a,b\}n}{\alpha^2(a-b)^2}\right) $.\footnote{We state our result for a simpler model than what \cite{CSV17} considers for simplicity of exposition -- an algorithm for the general model follows straightforwardly from one for this simplified model.}
\end{definition}
The result of \cite{CSV17} obtains a bound of $O\left(\frac{\max\{a,b\}\log(1/\alpha)n}{\alpha^2(a-b)^2}\right)$ on the size of the smallest $\wt{S}_i\Delta S$, and thus in addition to giving a significantly faster algorithm, we also give slightly improved statistical guarantees.
\begin{theorem}\label{thm:main-planted-part}
    We give an algorithm for the $\PPart$ problem that runs in $\wt{O}\left(n^2\cdot\poly(1/\alpha)\right)$.    
\end{theorem}
We will need the following concentration inequality from \cite{CSV17}.
\begin{lemma}[Proposition B.1 of \cite{CSV17}]  \label{lem:CSV-concentration}
    Let $\bX$ be a $\R^d$-valued random variable such that $\Cov[\bX]\psdle \sigma^2\cdot\Id$.  Let $\bX_1,\dots,\bX_m$ be $m$ independent copies of $\bX$.  Then there is a subset $J\subseteq[m]$ of size at least $(1-\eps)m$ such that
    \[
        \frac{1}{|J|}\sum_{i\in J}(\bX_i-\E\bX)(\bX_i-\E\bX)^{\transp}\psdle \frac{4\sigma^2}{\eps}\left(1+\frac{d}{(1-\eps)m}\right)
    \]
    except with probability at most $\exp\left(-\frac{\eps^2 m}{16}\right)$.
\end{lemma}

\begin{proof}[Proof of \cref{thm:main-planted-part}]
    Let $A_u$ denote the $n$-dimensional vector corresponding to outgoing edges of vertex $u$.  In particular
    \[
        A_u[v] =
        \begin{cases}
            1 &\text{$(u,v)$ is an edge}\\
            0 &\text{otherwise.}
        \end{cases}
    \]
    For $u\in S$,
    \[
        \E A_u[v] =
        \begin{cases}
            \frac{a}{n} &\text{if $v\in S$}\\
            \frac{b}{n} &\text{otherwise}
        \end{cases}
    \]
    and
    \[
        \Cov(A_u)[v,w] =
        \begin{cases}
            \frac{a}{n}\left(1-\frac{a}{n}\right) &\text{if $v=w$, $v,w\in S$}\\
            \frac{b}{n}\left(1-\frac{b}{n}\right) &\text{if $v=w$, $v,w\in S$}\\
            0 &\text{otherwise.}
        \end{cases}
    \]
    Let $c = \max\{a,b\}$; then $\Cov(A_u)[v,w]\psdle\frac{c}{n}\cdot\Id$ and from \pref{lem:CSV-concentration} there is a subset $S'\subseteq S$ of size $\alpha n/2$ such that
    \[
        \frac{1}{|S'|}\sum_{u\in S'}(A_u-\E A_u)(A_u-\E A_u)^{\transp} \psdle \frac{8c}{n}\left(1+\frac{n}{\alpha n/2}\right)\cdot\Id \psdle \frac{24c}{\alpha n}\cdot\Id
    \]
    except with probability $\exp\left(-\frac{\alpha n}{64}\right)$.  Let $\Sigma_{S'}$ be the covariance matrix and $\mu_{S'}$ be the mean of the uniform distribution on $\{A_u:u\in S'\}$.  The above can then be rewritten as
    \[
        \Sigma_S + (\E A_u-\mu_{S'})(\E A_u-\mu_{S'})^{\transp} \psdle \frac{24c}{\alpha n}\cdot\Id.
    \]
    Since $\Sigma_S$ is positive semidefinite,
    \[
        (\E A_u-\mu_{S'})(\E A_u-\mu_{S'})^{\transp} \psdle \frac{24c}{\alpha n}\cdot\Id
    \]
    and consequently
    \begin{align*}
        \|\E A_u-\mu_{S'}\|^2 &\le \frac{24c}{\alpha n} \numberthis \label{eq:closeness-means}
    \end{align*}
    We run the list-decodable mean estimation algorithm from \pref{thm:main-thm-list-dec-mean} on input $\left\{\sqrt{\frac{\alpha n}{24c}} A_u:u\in V(G)\right\}$ along with parameter $2/\alpha$ (where the scaling on input vectors is to ensure that the uniform distribution on the elements of $S'$ have unit covariance), and get a list $L$ of length $O(1/\alpha)$ as output in $O(n^2\cdot\poly(1/\alpha))$ time. Let $L'$ be the set obtained by scaling all elements of $L$ by $\sqrt{\frac{24c}{\alpha n}}$.  The guarantees of the algorithm in \pref{thm:main-thm-list-dec-mean} combined with the existence of the set $S'$ guarantees with high probability the existence of an element $\phi^*$ in $L'$ such that $\|\phi^*-\mu_{S'}\| \le O\left(\frac{1}{\alpha}\sqrt{\frac{c}{n}}\right)$.  Combining this with \pref{eq:closeness-means} and triangle inequality, we get
    \begin{align}
        \|\phi^*-\E A_u\| \le O\left(\frac{1}{\alpha}\sqrt{\frac{c}{n}}\right)    \label{eq:vecs-close}
    \end{align}
    We describe a procedure to translate vectors in $L'$ to sets in the following way:
    \begin{displayquote}
        Suppose $a < b$, then for each $\phi\in L'$, let $\wt{S}_{\phi} \coloneqq \left\{u:\phi_u<\frac{a+b}{2n}\right\}$; otherwise if $a > b$, we set $\wt{S}_{\phi}$ as $\left\{u:\phi_u>\frac{a+b}{2n}\right\}$.
    \end{displayquote}
    To show that this list of sets meet the required guarantee, we upper bound $|S\Delta\wt{S}_{\phi^*}|$.  Towards this goal, we establish a lower bound on $\|\phi^*-\E A_u\|$ as follows:
    \begin{align*}
        \|\phi^*-\E A_u\|^2 &= \sum_{v\in V(G)} (\phi^*[v]-(\E A_u)[v])^2\\
        &\ge \sum_{v\in S,v\notin\wt{S}_{\phi^*}}(\phi^*[v]-a/n)^2 + \sum_{v\in \wt{S}_{\phi^*},v\notin S}(\phi^*[v]-b/n)^2\\
        &\ge \sum_{v\in S,v\notin\wt{S}_{\phi^*}} \left(\frac{a-b}{2n}\right)^2 + \sum_{v\in \wt{S}_{\phi^*},v\notin S} \left(\frac{a-b}{2n}\right)^2\\
        &= |S\Delta\wt{S}_{\phi^*}|\cdot\left(\frac{a-b}{2n}\right)^2
    \end{align*}
    Combining the above with \pref{eq:vecs-close} tells us that $|S\Delta \wt{S}_{\phi^*}|\le O\left(\frac{cn}{\alpha^2(a-b)^2}\right)$.
\end{proof}

\section*{Acknowledgements}
We would like to thank Sam Hopkins and Prasad Raghavendra for helpful conversations.

\bibliographystyle{alpha}
\bibliography{cluster}

\appendix

\section{Algorithm Supporting Lemmas} 
\label{sec:miscellaneous}

\maincorollary*

\begin{proof} (Proof of Corollary)
We proceed by contradiction.  Assume $\norm{\hu - \mu} \geq r\frac{\sigma}{\sqrt \alpha}on$ for all $\hu \in \calL$.

We claim that the inlier weight at the start of iteration $t$ is $\sum_{i \in \calI} b_i  = 2 - \frac{(t-1)\alpha}{4}$.  We prove by induction.  The base case is true.  Now assume that this is true at iteration $t$. Since $t \leq \frac{4}{\alpha}$ the assumptions of \cref{thm:maintheorem} are satisfied and $DescendCost(X,b)$ outputs $(\hu, \bar w)$ satisfying $\sum_{i \in I}\bar{w}_i  \leq \frac{\alpha}{4}$.  Thus at the start of iteration $t+1$ the inlier weight is greater than $2 - \frac{t\alpha}{4}$.  This proves the claim.  

Therefore at the end of iteration $\frac{4}{\alpha}$ the inlier weight $\sum_{i \in \calI} b_i \geq 1$.  However, the algorithm runs for no more than $\frac{4}{\alpha}$ iterations removes at least $0.5$ weight per iteration until $\norm{b}_1 = 0$.  This is a contradiction as the inlier weight must be smaller than the total weight. This concludes the proof.         

\end{proof}

\begin{lemma} \label{lem:basecase} (Termination Base Case)
Let $\nu$ be a vector in $\R^d$.  Let $(\theta,\bar w)$ be the output of $ApproxCost_{X,b,\ell}(\nu)$ satisfying $\theta \leq \sigma^2$ for a positive integer $\ell > 1$ and for weight vector $b$ satisfying $\sum_{i \in \calI}b_i \geq 1$ and $b_i \in [0,\frac{2}{\alpha N}]$ for $i \in [N]$.  Then $(\nu,\bar w)$ is a sanitizing tuple.
\end{lemma}

\begin{proof} By assumption $\theta \leq \sigma^2$ or equivalently $\norm{\sum_{i=1}^N \bar w_i (x - \nu)(x - \nu)^T}_\ell \leq \sigma^2$ for $\bar w \in \Phi_b(1-\delta)$.  By the monotonicity of Ky Fan norm we also have $\norm{\sum_{i=1}^N \bar w_i (x_i - \nu)(x_i - \nu)^T} \leq \sigma^2$.  Applying \pref{fact:resilience} we obtain that if $\norm{\mu - \nu} \geq \frac{r\sigma}{\sqrt{\alpha}}$ then $\sum_{i \in I}\bar w_i \leq \frac{\alpha}{4}$.  Therefore $(\nu,\bar w)$ is a sanitizing tuple.
\end{proof}

\paragraph{Warm Start: } A warm start can be achieved simply by querying $ApproxCost_{X,b,\ell}(\nu)$ for $\nu = x_i$ for $\frac{\log(d)}{\log(\frac{1}{1 - \alpha})}$ randomly chosen $x_i \in X$ and taking the $\nu$ with minimum cost.  This procedure succeeds with high probability $1  - \frac{1}{d^{10}}$.  This follows directly from \pref{lem:costupper}.

\begin{lemma} \label{lem:costupper}(distance to true mean approximately upper bounds cost) For $X = \{x_1,...,x_N\}$ a dataset with $\alpha N$ inliers with covariance $Cov _{x \in I}(x) \preceq \sigma^2I$, and $b_i = \frac{2}{\alpha N}$ for all $i \in [N]$, we have $ApproxCost_{X,b,\ell}(\nu) \leq  \norm{\mu - \nu}^2 + 10\ell\sigma^2 $.    
\end{lemma}

\begin{proof}
Let $\widetilde w$ satisfy $\widetilde{w} \in \Phi_b(1)$ and $\langle \widetilde w, b^{\mathcal{O}} \rangle = 0$ then for $\widetilde \mu = \sum_{i=1}^N \widetilde{w}_i X_i$ we have 
\begin{multline*}
\text{ApproxCost}_{X,b,k}(\nu) \leq \max\limits_{M \in \calF_\ell} \min\limits_{w \in \Phi_b(1)}f( M, w) \leq \max\limits_{M \in \calF_\ell}f( M, \widetilde w)
\\= \max\limits_{M \in \calF_\ell}\langle M, \sum_{i=1}^N \widetilde{w}_i (x_i - \widetilde u)(x_i - \widetilde u )^T  + (\widetilde u - \nu)( \widetilde u - \nu)^T\rangle
= \max\limits_{M \in \calF_\ell} \langle M, \sum_{i=1}^N \widetilde{w}_i (x_i - \widetilde u)(x_i - \widetilde u )^T\rangle  + \langle M,(\widetilde u - \nu)( \widetilde u - \nu)^T\rangle
\end{multline*}
Where the first inequality follows by \pref{lem:informalapprox}, and second inequality follows because  \newline $\widetilde{w} \in \Phi_b(1)$.  Further upper bounding we obtain 
\begin{align*}
\leq \max\limits_{M \in \calF_\ell}\langle M, (\widetilde{\mu} - \nu)(\widetilde{\mu} - \nu)^T \rangle + 4\ell\sigma^2 \leq \norm{\widetilde \mu - \nu}^2 + 4\ell\sigma^2 \leq  \norm{\mu - \nu}^2 + (4\ell + 2)\sigma^2   
\end{align*}
The first inequality follows from $Tr(M) = \ell$ and that $Cov_{\widetilde{w}}(X) \preceq 4\sigma^2I$ by \cref{fact:momentfacts}, the second inequality follows by $M \preceq I$, the third inequality follows by $\norm{\mu - \widetilde\mu} \leq \sqrt{2}\sigma$ where we use \cref{fact:momentfacts}.       
\end{proof}

\begin{fact} \label{fact:resilience}(Resilience of Bounded Covariance Distributions) Let $w$ and $w'$ be two vectors in $\R^N$ where $w_i \geq 0$ and $w'_i \geq 0$ for all $i \in [N]$ and $\norm{w}_1 = 1$ and $\norm{w'}_1 = 1$ such that $\norm{\sum_{i=1}^N w_i(x_i - \mu)(x_i - \mu)^T} \leq \sigma_1^2$ and $\norm{\sum_{i=1}^N w'_i(x_i - \mu')(x_i - \mu')^T} \leq \sigma_2^2$ where $\mu$ and $\mu'$ are vectors in $\R^d$.  Then if $S\coloneqq \sum_{i = 1}^N \min (w_i, w'_i) \geq \gamma$, $\norm{\mu - \mu'} \leq \sqrt{\frac{2\sigma_1^2 + 2\sigma_2^2}{\gamma}}$.
\end{fact}
\begin{proof}
We have for any $\norm{u} = 1$
\begin{align*}
    \inp{u}{\mu - \mu'} &= \frac{1}{S} \inp*{u}{\sum_{i = 1}^N \min(w_i, w'_i) (\mu - \mu')} = \frac{1}{S} \sum_{i = 1}^N  \inp*{u}{\min(w_i, w'_i) (\mu - \mu')} \\
    &= \frac{1}{S} \sum_{i = 1}^N  \inp*{u}{\min(w_i, w'_i) (\mu - x_i)} + \inp*{u}{\min(w_i, w'_i) (x_i - \mu')} \\
    &\leq \sqrt{\frac{1}{S} \sum_{i = 1}^N \min(w_i, w'_i) \lprp{\inp*{u}{(\mu - x_i)} + \inp*{u}{(x_i - \mu')}}^2} \\
    &\leq \sqrt{\frac{2}{S} \sum_{i = 1}^N  \min(w_i, w'_i) \inp*{u}{(\mu - x_i)}^2 + \min(w_i, w'_i) \inp*{u}{(x_i - \mu')}^2} \\
    &\leq \sqrt{\frac{2 (\sigma_1^2 + \sigma_2^2)}{S}} \leq \sqrt{\frac{2 (\sigma_1^2 + \sigma_2^2)}{\gamma}}.
\end{align*}
By maximizing over $u$, the conclusion of the lemma follows.
\end{proof}

\begin{fact}\label{fact:momentfacts}
    (Moment Facts) For a set of points $X = \{x_1,...,x_N\}$ with mean $\mu$ satisfying $Cov(X) \preceq \sigma^2I$.  Let $w \in \R^N$ be a weight vector satisfying $w_i \in [0,\frac{1}{N}]$ for all $i \in [N]$.  Then for $\widetilde \mu := \frac{1}{\norm{w}_1} \sum_{i=1}^N w_ix_i$ we have $\norm{\mu - \widetilde\mu} \leq \frac{\sigma}{\sqrt{\norm{w}_1}}$ and $\frac{1}{\norm{w}_1}\sum_{i=1}^N w_i(x_i - \widetilde \mu)(x_i - \widetilde \mu)^T \preceq \frac{\sigma^2}{\norm{w}_1^2}I$     
\end{fact}

\begin{proof}
    First notice that 
    \begin{align*}
        \norm{\mu - \widetilde \mu}^2 =\norm{\mu - \frac{1}{\norm{w}_1} \sum_{i=1}^N w_i x_i}^2 = \norm{\frac{1}{\norm{w}_1} \sum_{i=1}^N w_i (x_i - \mu)}^2 = \max_{u \in S^{d-1}}\langle\frac{1}{\norm{w}_1} \sum_{i=1}^N w_i (x_i - \mu),u\rangle^2 \\
        \leq \max_{u \in S^{d-1}}\frac{1}{\norm{w}_1} \sum_{i=1}^N w_i \langle x_i -  \mu,u\rangle^2 \leq \max_{u \in S^{d-1}}\frac{1}{\norm{w}_1} \frac{1}{N}\sum_{i=1}^N \langle x_i -  \mu,u\rangle^2
        \leq \frac{\sigma^2}{\norm{w}_1}
    \end{align*}
    which implies $\norm{\mu - \widetilde\mu} \leq \frac{\sigma}{\sqrt{\norm{w}_1}}$ as desired.  Here the first inequality is Jensen's, and the second inequality follows by $w_i \le \frac{1}{N}$ for all $i \in [N]$, and the last inequality follows by $Cov(X) \preceq \sigma^2I$.  Furthermore, we have
    \begin{multline*}
        \frac{1}{\norm{w}_1}\sum_{i=1}^N w_i(x_i - \widetilde \mu)(x_i - \widetilde \mu)^T 
        \preceq \frac{1}{\norm{w}_1}\frac{1}{N}\sum_{i=1}^N (x_i - \widetilde \mu)(x_i - \widetilde \mu)^T 
        \\= \frac{1}{\norm{w}_1}(\frac{1}{N}\sum_{i=1}^N (x_i - \mu)(x_i - \mu)^T + (\mu - \widetilde \mu)(\mu - \widetilde \mu)^T)
         \preceq \frac{1}{\norm{w}_1}(\sigma^2 + \norm{\mu - \widetilde \mu}^2)I \preceq \frac{2\sigma^2}{\norm{w}_1^2}I 
    \end{multline*}
    Here the first inequality follows by $w_i \leq \frac{1}{N}$ for all $i \in [N]$, the second inequality follows by $Cov(X) \preceq \sigma^2I$, and the last inequality follows by using $\norm{\mu - \widetilde\mu} \leq \frac{\sigma}{\sqrt{\norm{w}_1}}$.  Thus, \newline $\frac{1}{\norm{w}_1}\sum_{i=1}^N w_i(x_i - \widetilde \mu)(x_i - \widetilde \mu)^T \preceq \frac{\sigma^2}{\norm{w}_1^2}I$ as desired.  
\end{proof}
\section{Sampling Based Methods for Trace and Inner Product Estimation}
\label{sec:appsketch}

In this section, we prove standard results enabling efficient procedures for estimating the trace and matrix inner products using variants of the Johnson-Lindenstrauss method. We first recall a Lemma from \cite{DBLP:journals/jacm/AroraK16}:

\begin{lemma}[\cite{DBLP:journals/jacm/AroraK16}]
    \label{lem:expapx}
    Let $B$ be a PSD matrix satisfying $\norm{B} \leq \kappa$. Then, the operator:
    
    \begin{equation*}
        \hat{B} = \sum_{i = 0}^k \frac{1}{i!} B^i \text{ where } k = \max\{e^2 \kappa, \log (2\eps^{-1})\}
    \end{equation*}
    
    satisfies
    
    \begin{equation*}
        (1 - \eps) \exp (B) \preccurlyeq \hat{B} \preccurlyeq \exp (B).
    \end{equation*}
\end{lemma}

Additionally, we include a variant on the Johnson-Lindenstrauss Lemma as stated in \cite{matouvsek2013lecture}:

\begin{lemma}[Lemma 2.3.1 from \cite{matouvsek2013lecture}]
    \label{lem:vecjl}
    Let $n, k$ be natural numbers and let $\eps \in (0, 1)$. Define the random linear map, $T: \R^n \to \R^k$ by:
    
    \begin{equation*}
        T(x)_i = \frac{1}{\sqrt{k}} \sum_{j = 1}^n Z_{ij} x_j
    \end{equation*}
    
    where $Z_{ij}$ are independent standard normal variables. Then we have for any vector $x \in \R^n$:
    
    \begin{equation*}
        \bm{P} \lbrb{(1 - \eps) \norm{x} \leq \norm{T(x)} \leq (1 + \eps) \norm{x}} \geq 1 - 2e^{-c\eps^2 k},
    \end{equation*}
    
    where $c > 0$ is a constant.
\end{lemma}

\begin{corollary}
    \label{cor:jl}
    Let $x_1, \dots, x_m \in \R^n$ and $T$ be defined as in \ref{lem:vecjl}. Then, we have:
    
    \begin{equation*}
        \bm{P} \lbrb{\forall i: (1 - \eps) \norm{x_1} \leq \norm{T(x_i)} \leq (1 + \eps) \norm{x_i}} \geq 1 - 2 m e^{-c\eps^2 k},
    \end{equation*}
    
    where $c > 0$ is a constant.
\end{corollary}

\begin{proof}
    The corollary follows through the union bound applied as follows:
    
    \begin{align*}
        &\bm{P} \lbrb{\exists i: \lnot ((1 - \eps) \norm{x_1} \leq \norm{T(x_i)} \leq (1 + \eps) \norm{x_i})} \\ 
        &\leq \sum_{i} \bm{P} \lbrb{\lnot ((1 - \eps) \norm{x_1} \leq \norm{T(x_i)} \leq (1 + \eps) \norm{x_i})} \leq 2m e^{-c\eps^2 k}.
    \end{align*}
\end{proof}

Next, we show to estimate matrix inner products using the above lemma. In this setup, one is given $m$ PSD matrix $M_1 \dots, M_l$ with $M_i = U_i U_i^\top$ and a single PSD matrix, $B$, and the goal is to obtain estimates of $\inp{M_i}{\exp (B)}$. We include the pseudo-code for the procedure below:

\begin{algorithm}[H]
\SetAlgoLined


\KwIn{PSD Matrix $B = WW^\top$ with $\norm{B} \leq \kappa$, Accuracy $\eps$, Failure Probability $\delta$, PSD Matrices $\{M_i = U_i U_i^\top\}_{i = 1}^m$}
\KwOut{Estimates of $\inp{M_i}{\exp (B)}$}

Let $k = \max\{4 e^2 \kappa, \log (4\eps^{-1})\}$ and $\wt{B} = \sum_{0 \leq i \leq k} \frac{B^i}{2^i \cdot i!}$\;
Let $l = O(\frac{\log m + \log n + \log 1 / \delta}{\eps^2})$\;
Let $\Pi \in \R^{l \times n}$ be distributed as $\Pi_{i, j} \sim \mc{N}(0, \frac{1}{l})$ independently for each $i, j$ and $Q = \Pi \wt{B}$\;

\KwRet{$\{z_i = \norm{Q U_i}^2\}_{i = 1}^m$}
\caption{InnerProductEstimation}
\label{alg:inpest}
\end{algorithm}

We now show that Algorithm~\ref{alg:inpest} produces estimates of $\inp{M_i}{\exp (B)}$ with high probability.

\begin{lemma}
    \label{lem:inpest}
    Algorithm~\ref{alg:inpest} when given input, $B = WW^\top$ with $\norm{B} \leq \kappa$ and $W \in \R^{n \times s}$ and $M_i = U_iU_i^\top$ with $U_i \in \R^{n \times r_i}$ and $\eps, \delta \in (0,1/4)$ returns estimates, $\{z_i\}_{i = 1}^m$ satisfying:
    
    \begin{equation*}
        \bm{P} \lbrb{\forall i: (1 - \eps) \inp{M_i}{\exp (B)} \leq z_i \leq (1 + \eps) \inp{M_i}{\exp (B)}} \geq 1 - \delta.
    \end{equation*}
    
    And furthermore, Algorithm~\ref{alg:inpest} runs in time $O(nl + klt_W + lt_U)$ where $t_W$ is the time required for a matrix-vector multiplication with the matrix $W$ or $W^\top$, $t_U$ is the time taken to compute $v^\top U_i$ for all $U_i$ and any vector $v$:
    
    \begin{equation*}
        k = \max \lbrb{4e^2 \kappa, \log (4\eps^{-1})}\ l = O\lprp{\frac{\log m + \log n + \log 1 / \delta}{\eps^2}}.
    \end{equation*}
\end{lemma}

\begin{proof}
     From Lemma~\ref{lem:expapx}, we have that:
     
     \begin{equation*}
         \lprp{1 - \frac{\eps}{4}} \exp \lprp{\frac{B}{2}} \preccurlyeq \wt{B} \preccurlyeq \exp \lprp{\frac{B}{2}}.
     \end{equation*}
     
     Let $\hat{B} = \exp \lprp{\frac{B}{2}}$. Observe that the eigenvectors of $\hat{B}$, $\wt{B}$ and $B$ coincide. Let $\hat{B} = \sum_{i = 1}^n \lambda_i v_iv_i^\top$ and $\wt{B} = \sum_{i = 1}^n \sigma_i v_iv_i^\top$ by the eigenvalue decompositions of $\hat{B}$ and $\wt{B}$. From the previous relationship, we have $(1 - \eps / 4) \lambda_i \leq \sigma_i \leq \lambda_i$. Therefore, we observe by squaring $\hat{B}$ and $\wt{B}$:
     
     \begin{equation*}
         \lprp{1 - \frac{\eps}{4}}^2 \exp \lprp{B} \preccurlyeq \wt{B}^2 \preccurlyeq \exp \lprp{B}.
     \end{equation*}
     
     Now, let $u^i_j$ for $j \in [r_i]$ denote the columns of $U_i$ and let $U = \{u^i_j: \forall i \in [m], j \in [r_i]\}$. Then, we have via a union bound from our settings of $l$ and Lemma~\ref{cor:jl} that with probability at least $1 - \delta$ for all $u \in U$:
     
     \begin{equation*}
         (1 - \eps / 4) \norm{\wt{B} u} \leq \norm{\Pi \wt{B} u} \leq (1 + \eps / 4) \norm{\wt{B} u}.
     \end{equation*}
     
     Now, conditioning on this event, we have by squaring both sides that and the previous conclusion for all $u\in U$:
     \begin{equation*}
         (1 - \eps / 4)^4 u^\top \exp (B) u \leq (1 - \eps / 4)^2 u^\top \wt{B}^2 u \leq \norm{\Pi \wt{B} u}^2 \leq (1 + \eps / 4)^2 u^\top \wt{B}^2 u \leq (1 + \eps / 4)^2 u^\top \exp (B) u.
     \end{equation*}
     
     From the previous inequality, using the fact that $(1 - \eps) \leq (1 - \eps/4)^4$ and $(1 + \eps/4)^2 \leq (1 + \eps)$ in our range of $\eps$, that for all $i \in [m]$:
     
     \begin{equation*}
         (1 - \eps) \sum_{j \in [r_i]} (u^i_j)^\top \exp(B) (u^i_j) = (1 - \eps) \inp{\exp(B)}{M_i} \leq \norm{\Pi \wt{B} U_i}^2 \leq (1 + \eps) \inp{\exp(B)}{M_i}. 
     \end{equation*}
     
     Since, the above event conditioned on occurs with probability at least $1 - \delta$, this concludes the proof of correctness of the output of the algorithm with probability at least $1 - \delta$. Finally the runtime of the algorithm is dominated by the time taken to compute $Q$ which takes time $O(lkt_W)$ and the time taken to compute $QU_i$ for all $i$ which takes time $O(lt_U)$. 
\end{proof}

Note, in our applications $r_i$ is typically $1$ and $l$ is typically small. Therefore, the runtime reduces to $\tilde{O}(mn)$ for large $m$ which is sufficient for our purposes. We now include the following corollary which will we will use at several points through the course of the paper:

\begin{corollary}
    \label{cor:esttr}
    Given $B = WW^\top$ with $W \in \R^{n \times s}$, vectors $\{v_i \in \R^n\}_{i = 1}^m$ and $\eps, \delta \in (0, 1/4)$ there exists a randomized algorithm, $\tipest{}$, which computes an estimate, $z$, satisfying:
    
    \begin{equation*}
        (1 - \eps) \sum_{i = 1}^m v_i^\top \exp (B) v_i \leq z \leq (1 + \eps) \sum_{i = 1}^m v_i^\top \exp (B) v_i
    \end{equation*}
    with probability at least $1 - \delta$. And furthermore, this algorithm runs in time $O(klt_W + nlm)$ where $t_W$ is the time required to compute a matrix-vector multiplication with the matrix $W$ or $W^\top$:
    
    \begin{equation*}
        k = \max \lbrb{4e^2 \kappa, \log (4\eps^{-1})}\ l = O\lprp{\frac{\log m + \log n + \log 1 / \delta}{\eps^2}}.
    \end{equation*}
    
    Furthermore, if $v_i = C e_i$, one obtains the same guarantees with the runtime reduced to $O(nl + klt_W + lt_C)$ where $t_C$ is the time is the time taken to compute a matrix vector multiplication with the matrix $C$.
\end{corollary}

\begin{proof}
     The first claim follows by summing up the output of \cref{alg:inpest} with input $M_i = v_iv_i^\top$, $B = WW^\top$, $\eps$ and $\delta$. The second follows by computing the Frobenius norm of $Q C$ in \cref{alg:inpest} which takes time $O(klt_W + lt_C)$.
\end{proof}
\newcommand{\psimp}[1][\delta]{\Delta_{#1}}
\newcommand{\lmax}{\lambda_{\max}}
\newcommand{\tO}{\tilde{O}}

\section{Fast Min-Max Optimization}
\label{sec:fstsolv}

We prove the existence of nearly linear time solvers for the class of SDPs required in our algorithms. Recall that given a set of points $X = \{x_i\}_{i = 1}^N$, vector $\nu$, set of weight budgets for each point $b = \{b_i > 0\}_{i \in [N]}$ and a rank $k$, we aim to solve the following optimization problem:

\begin{equation*}
    \min_{w \in \phib[1]} \max_{M \in \fan[k]} \inp*{M}{\sum_{i \in [N]} w_i (x_i - \nu)(x_i - \nu)^\top} = \max_{M \in \fan[k]} \min_{w \in \phib[1]} \inp*{M}{\sum_{i \in [N]} w_i (x_i - \nu)(x_i - \nu)^\top}.
\end{equation*}

We will first start by reformulating the above objective with the following mean adjusted data points instead $Z = \{z_i = x_i - \nu\}_{i \in [N]}$. Therefore, the objective reduces to the following reformulation which we will use throughout the rest of the section:

\begin{equation}
    \label{eq:minmax} \tag{MT}
    \min_{w \in \phib[1]} \max_{M \in \fan[k]} \inp*{M}{\sum_{i \in [N]} w_i z_iz_i^\top} = \min_{w \in \phib[1]} \norm{\sum_{i \in [N]} w_i z_iz_i^\top}_k.
\end{equation}

We solve this problem via a reduction to the following packing SDP by introducing an additional parameter $\lambda$:

\begin{equation*}
\label{eq:packingred}
\begin{gathered}
    \max_{w} \sum_{i \in [N]} w_i \\
    \text{Subject to: } 0 \leq w_i \leq b_i\ \forall i \in [N]\\
    \norm*{\sum_{i \in [N]} w_i z_iz_i^\top}_k \leq \lambda.
\end{gathered} \tag{Pack}
\end{equation*}

\newcommand{\optv}{\text{OPT}^*}
\newcommand{\sumlow}{l^*}
\newcommand{\optp}{\text{Pack}^*_{\lambda}}

Let $\optv{}$ denote the optimal value of the program \ref{eq:minmax}, \ref{eq:packingred}$(\lambda)$ denote the program $\text{\ref{eq:packingred}}$ instantiated with $\lambda$ and let $\optp$ denote its optimal value. The following quantity is useful throughout the section:
\begin{equation*}
    \sumlow = \min_{w \in \phib[1]} \sum_{i \in [N]} w_i \norm{z_i}^2.
\end{equation*}

This is equivalent to taking sorting the $z_i$ in terms of their lengths and computing their average squared length with respect to their budgets, $b_i$, such that their budgets sum to $1$. We introduce a technical result useful in the following analysis:

\begin{lemma}
    \label{lem:optval}
    \ref{eq:packingred}$(\optv)$ has optimal value at least $1$. 
\end{lemma}
\begin{proof}
    The lemma follows from the fact that a feasible solution for \ref{eq:minmax} achieving $\optv$ is a feasible point for \ref{eq:packingred}$(\lambda)$ for $\lambda \geq \optv$.
\end{proof}

The following lemma proves that $\sumlow$ gives an approximation to $\optv$ within a factor of $d$.

\begin{lemma}
    \label{lem:sumlowapx}
    The value $l^*$ satisfies:
    \begin{equation*}
        \optv \leq l^* \leq \frac{d}{k} \optv.
    \end{equation*}
\end{lemma}
\begin{proof}
    The upper bound on $\optv$ follows from that fact that:
    \begin{equation*}
        \norm{\sum_{i \in [N]} w_i z_iz_i^\top}_k \leq \Tr (\sum_{i \in [N]} w_i z_iz_i^\top) = \sum_{i \in [N]} w_i \norm{z_i}^2
    \end{equation*}
    and the lower bound follows from the inequality $\Tr M \leq \frac{d}{k} \norm{M}_k$ for any psd matrix $M$.
\end{proof}

\newcommand{\optlamb}[1][\lambda]{\text{OPT}_{#1}}

In what follows we prove that we can efficiently binary search over the value of $\lambda$ to find a good solution to \ref{eq:minmax}. We  refer to $\optlamb{}$ as the optimal value of \ref{eq:packingred} run with $\lambda$. 

\begin{lemma}
    \label{lem:monl}
    The function, $\optlamb$ when viewed as a function of $\lambda$ is monotonic in $\lambda$.
\end{lemma}
\begin{proof}
    The lemma follows from the observation that for $\lambda_1 \geq \lambda_2$, a feasible point for \ref{eq:packingred} with $\lambda_2$ is a feasible point for the program with $\lambda_1$. 
\end{proof}

\newcommand{\optld}[1][\lambda]{\text{OPT}_{#1, \epd}}

We now conclude with the main lemma of the section.

\begin{lemma} \label{lem:minmaxmain}
    \label{lem:fastsdp}
    There exists a randomized algorithm, $\packred$, which when given input $N$ data points $\{x_i\}_{i = 1}^N$, an arbitrary vector $\nu$, weight budgets $\{b_i > 0\}_{i = 1}^N$, error tolerance $\eps$ and failure probability $\delta$, computes a solution, $\hat{w}$ satisfying:
    
    \begin{gather*}
        \hat{w}_i \leq b_i\ \forall i \in [N] \\
        \norm*{\sum_{i = 1}^N \hat{w}_i z_i z_i^\top}_k \leq \min_{w \in \phib[1]} \norm*{\sum_{i \in [N]} w_i z_i z_i^\top}_k \text{ and } \\
        \sum_{i \in [N]} \hat{w}_i \geq (1 - \eps)
    \end{gather*}
    
    where $z_i = x_i - \nu$, with probability at least $1 - \delta$. Furthermore, $\packred$ runs in time at most:
    
    \begin{equation*}
        O\lprp{Nd \poly \lprp{\frac{1}{\eps},\, \log \frac{1}{\delta},\, k,\, \log (N + d)}}.
    \end{equation*}
\end{lemma}

\begin{proof}
     We first start by reducing to the following packing problem:
     \begin{equation}
         \label{eq:lamdred}
         \begin{gathered}
             \max_{w \geq 0} \sum_{i \in [N]} w_i \\
             \text{Subject to: } w_i \cdot \lprp{\frac{1}{(1 + \epd) b_i}} \leq 1 \tag{Pack-Red}\\
             \norm*{\sum_{i \in [N]} w_i \frac{z_iz_i^\top}{(1 + \epd)(\lambda/k)}}_k \leq k.
         \end{gathered}    
     \end{equation}

     To see that this is packing problem, notice that the above problem is equivalent to setting the constraint matrices $A_i$ and $B_i$ to:
     
     \begin{equation*}
         A_i = \lprp{\frac{1}{\sqrt{(1 + \epd) b_i}} \cdot e_i} \lprp{\frac{1}{\sqrt{(1 + \epd) b_i}} \cdot e_i}^\top \text{ and } B_i = \lprp{\frac{1}{\sqrt{(1 + \epd) (\lambda / k)}} \cdot z_i} \lprp{\frac{1}{\sqrt{(1 + \epd)(\lambda / k)}} \cdot z_i}^\top.
     \end{equation*}
     
     Let $\optld$ refer to the optimal value of \ref{eq:lamdred} and \ref{eq:lamdred}($\lambda$) denote the problem instantiated with $\lambda$. First notice that $\optld = (1 + \epd) \optlamb$ as for any feasible point of \ref{eq:lamdred}, $w$, $(1 + \epd)^{-1}w$ is a feasible point for \ref{eq:packingred} and vice-versa. 
     
     We will now perform a binary search procedure on the parameter, $\lambda$, to obtain a suitable solution to \ref{eq:lamdred} with our solver. Our binary search procedure will maintain two estimates, $(\lambda_l, \lambda_h)$ satisfying the following two properties which we will prove via induction:
     
     \begin{enumerate}
         \item We have a candidate solution, $w$, for \ref{eq:lamdred}$(\lambda_h)$ with $\sum_{i \in [N]} w_i \geq (1 - \epd/4)$.
         \item We have that $\optv \geq \lambda_l$.
     \end{enumerate}
     
     We will run our solver from \cref{thm:solthm}, $\solvalg$, with the error parameter set to $\epd / 4$ on \ref{eq:lamdred} for different values of $\lambda$ and failure probability to be determined subsequently. We instantiate $\lambda_h = \sumlow$ and $\lambda_l = \frac{k}{d} \sumlow$. We will now assume that the solver runs successfully and bound the failure probability at the end of the algorithm. To ensure that the first two conditions hold, we run the solver on \ref{eq:lamdred}$(\lambda_h)$. Note that the optimal value of \ref{eq:lamdred}$(\sumlow)$ is at least $(1 + \epd)$ from \cref{lem:optval,lem:sumlowapx,lem:monl} and the previous discussion. Therefore, the solver cannot return a primal feasible point, $(M, W)$, with objective value $1 + \epd / 4$. The second condition follows straightforwardly from \cref{lem:sumlowapx}. Now, in each step, we compute $\lambda_m = (\lambda_h + \lambda_l) / 2$ and run our solver on \ref{eq:lamdred}$(\lambda_m)$. We now have two cases:
     
     \begin{enumerate}
         \item If the solver returns a primal point, $(M, W)$, we set $\lambda_l = \lambda_m$. The first condition trivially holds true after this step. For the second condition, note that if $\lambda_m \geq \optv$, we have from \cref{lem:monl,lem:optval} that the optimal value of \ref{eq:lamdred}$(\lambda_m)$ is at least $(1 + \epd)$. Hence, the solver cannot return a primal point with objective value $(1 + \epd / 4)$ in this case. Therefore, we conclude that $\optv \geq \lambda_m$. This verifies the second condition of the induction hypothesis.
         \item If the solver returns a dual point, $w$, it must satisfy $\sum_i w_i \geq (1 - \epd / 4)$. This verifies the first condition and the second condition follows from the induction hypothesis.
     \end{enumerate}
     
     After $O(\log d / \epd)$ steps of binary search, we have that $(\lambda_h - \lambda_l) \leq \epd \cdot \optv$ from \cref{lem:sumlowapx}. From the second condition, we have that $\lambda_h \leq (1 + \epd)\optv$. Now, for the feasible $w$ at $\lambda_h$ with $\sum_{i = 1}^N w_i \geq 1 - \epd/4$, we have:
     
     \begin{equation*}
         \norm{\sum_{i \in [N]} w_i z_iz_i^\top}_k \leq (1 + \epd) \optv \implies \norm*{\sum_{i \in [N]} \frac{w_i}{(1 + \epd)} z_iz_i^\top}_k \leq \optv.
     \end{equation*}
     
     Letting $\wt{w} = \frac{w}{(1 + \epd)}$, we have that $\wt{w}$ is feasible for \ref{eq:packingred} with:
     
     \begin{equation*}
         \sum_{i \in [N]} \wt{w}_i = (1 + \epd)^{-1} \sum_{i \in [N]} w_i \geq (1 + \epd)^{-1} (1 - \epd / 4) \geq 1 - \frac{5\epd}{4}.
     \end{equation*}
     
     and furthermore, from the previous equation, we have that $\norm{\sum_{i \in [N]} \wt{w}_i z_iz_i^\top}_k \leq \optv$. Now, we set $\epd = \frac{4}{5} \delta$, and return $\wt{w}$ so obtained. 
     
     We now set the failure probability in $\solvalg$ is set to $O(\delta / (\log d / \epd))$ and therefore, the probability that the solver fails in any of the steps of the binary search is upper bounded by $\delta$ from the union bound. Finally, we bound the run time of the algorithm. Since, we only run $O(\log d / \epd)$ iterations of binary search, our overall running time bounded by:
     
     \begin{equation*}
         O\lprp{Nd \poly \lprp{\frac{1}{\eps},\, \log \frac{1}{\delta},\, k,\, \log (N + d)}}
     \end{equation*}
     
     as we have $n = N$, $l = N$, $m = d$, $t_C = N$ and $t_D = Nd$ in \cref{thm:solthm}.
\end{proof}

\end{document}